\documentclass{article}

     \usepackage[preprint,nonatbib]{neurips_2020}

\usepackage[utf8]{inputenc} % allow utf-8 input
\usepackage[T1]{fontenc}    % use 8-bit T1 fonts
\usepackage{url}            % simple URL typesetting
\usepackage{booktabs}       % professional-quality tables
\usepackage{amsfonts}       % blackboard math symbols
\usepackage{nicefrac}       % compact symbols for 1/2, etc.
\usepackage{microtype}      % microtypography
\usepackage{caption}        %for minipage with figures
\usepackage{lscape}       %for landscape

\usepackage{amsmath,amssymb}
\usepackage[mathscr]{euscript} %To be able to use script letters with \mathscr{}
\usepackage{amsthm} %To be able to write Definitions, Theorems, etc.
\usepackage{tikz} %To draw graphs.
\usepackage{graphicx} % To scale equations
\usepackage{framed} % To be able to frame theorems and stuff
\usepackage{multirow}
\usepackage{color}
\usepackage{subfigure}
\usepackage{longtable}
\usepackage{MnSymbol}
\usepackage{xcolor}  \definecolor{shadecolor}{rgb}{.95,.95,.95}  %To put a shaded region
\usepackage[hidelinks]{hyperref} %To be able to use links inside the document; [hidelinks removes the ugly red boxes]
\usepackage[ruled,vlined]{algorithm2e}

\newtheorem{myDefinition}{Definition}
\newtheorem{myTheorem}{Theorem}
\newtheorem{myLemma}{Lemma}

\newcommand{\bs}[1]{\boldsymbol{#1}}

\title{On the Identifiability of Phylogenetic \\ Networks
  under a
  Pseudolikelihood model}

\author{%
  Claudia Sol\'is-Lemus\thanks{corresponding author:
    solislemus@wisc.edu}\\
  Wisconsin Institute for Discovery \\
  Department of Plant Pathology \\
  University of Wisconsin-Madison\\
  \And
  Arrigo Coen\\
  Wisconsin Institute for Discovery \\
  University of Wisconsin-Madison\\
  \And
  C\'ecile An\'e \\
  Department of Statistics \\
  Department of Botany \\
  University of Wisconsin-Madison\\
}

\begin{document}

\maketitle

\begin{abstract}
  The Tree of Life is the graphical structure that represents the
  evolutionary process from single-cell organisms at the origin of
  life to the vast biodiversity we see today. Reconstructing this tree
  from genomic sequences is challenging due to the variety of
  biological forces that shape the signal in the data, and many of
  those processes like incomplete lineage sorting and hybridization
  can produce confounding information. Here, we present the
  mathematical version of the identifiability proofs of phylogenetic
  networks under the pseudolikelihood model in SNaQ
  \cite{Solis-Lemus2016}. We establish that the ability to detect
  different hybridization events depends on the number of nodes on the
  hybridization blob, with small blobs (corresponding to closely
  related species) being the hardest to be detected. Our work focuses
  on level-1 networks, but raises attention to the importance of
  identifiability studies on phylogenetic inference methods for
  broader classes of networks.
\end{abstract}

\section{Introduction}
\label{introSec}

% \noindent \textbf{Significance.}  Scientists world-wide are putting
% together massive efforts to understand how the biodiversity that we
% see on Earth evolved from a single-cell organism at the origin of
% life. This diversification process is represented through the Tree of
% Life, and no shortage of data, money or time are spared to reconstruct
% this tree from genomic data. Achieving the goal of estimating the Tree
% of Life would not only represent the greatest accomplishment in the
% history of evolutionary biology and systematics, but it would also
% allow us to fully understand the development and evolution of
% important biological traits in nature, in particular, those related to
% resilience to extinction when exposed to environmental threats such as
% climate change. That is, the Tree of Life would allow us not only to
% estimate ancestral traits in past extinct populations, but also to
% identify characteristics in fast-adapting organisms, and utilize this
% knowledge to bio-engineer assistance to those species in risk of
% extinction. The development of statistical and machine-learning theory
% that can help reconstruct the Tree of Life, especially those scalable
% to big data like the pseudolikelihood model, are paramount in
% evolutionary biology, systematics, and conservation efforts against mass
% extinctions.

\noindent \textbf{Background.} The Tree of Life is a massive graphical
structure which depicts all living organisms. Graphical structures
that represent evolutionary processes are denoted phylogenetic
trees. In mathematical terms, a phylogenetic tree is a fully
bifurcating tree in which internal nodes represent ancestral species
that over time differentiate into two separate species giving rise to
its two children nodes (see Figure \ref{netEx}). The evolutionary
process is then depicted by this bifurcating tree from the root (the
origin of life) to the external nodes of the tree (also denoted
leaves) which represent the immense biodiversity of living species in
present time.

Recently, scientists have challenged the notion that evolution can be
represented with a fully bifurcating process, as this process cannot
capture important biological realities like hybridization,
introgression or horizontal gene transfer, that require two fully
separated branches to join back. Thus, recent years have seen an
explosion of methods to reconstruct phylogenetic networks, which
naturally account for reticulate evolution (see \cite{Degnan2018} for
a review on network methods).

Estimating phylogenetic networks from genomic data is challenging
because the signal in the data is confounded by the myriad of
biological processes that shape it. Work on theoretical guarantees of
network estimation methods are still lacking (but see \cite{Pardi2015,
  Zhu2018, Francis2018, Banos2019}). Here, we characterize which types
of phylogenetic networks can be estimated from genomic data, thus
providing theoretical guarantees to phylogenetic inference methods.
In particular, we focus on the pseudolikelihood model in
\cite{Solis-Lemus2016}, and follow up on their work by providing more
mathematical details that had to be excluded for biological audiences.
Our rationale for the mathematical details of the proofs is to allow
further mathematical developments to be built upon what we
demonstrated.

\noindent \textbf{Main findings.} Hybridization events create blobs
(or cycles) in the network (see Figure \ref{netEx}). We prove that a
hybridization cycle can be detected under a pseudolikelihood model
only if it spans 4 or more nodes. This rules out gene flow between
closely related species, such as sister species.  Our work
provides theoretical guarantees to the pseudolikelihood estimation of
larger hybridization cycles, while bringing up attention to the need
for novel models and methods to estimate gene flow between closely
related species. The theoretical guarantees for the
pseudolikelihood estimation\cite{Solis-Lemus2016} are paramount since
it is a highly scalable and parallelizable method which has the
potential to analyze big genomic data with high accuracy.

\noindent \textbf{Comparison to \cite{Solis-Lemus2016}.}  In 2016,
\cite{Solis-Lemus2016} published the first pseudolikelihood estimation
method for phylogenetic networks, and this work studies the
detectability of hybridization events and the estimability of
numerical parameters for the restricted case of level-1 networks.
However, in this paper, the mathematical proofs of identifiability
could not presented with the necessary mathematical rigor given that
the manuscript was published in biological journal.  The lack of
mathematical details in \cite{Solis-Lemus2016} makes difficult to
expand the work to broader classes of networks.  Our work fills this
gap by providing the mathematical details of the identifiability
proofs in \cite{Solis-Lemus2016} in hopes to open the door to future
developments on this important topic.

% \subsection*{Organization of the paper}
% In \modelSec\ we formally state the problem and our main results.  In...

\begin{figure}
\centering
\includegraphics[scale=0.12]{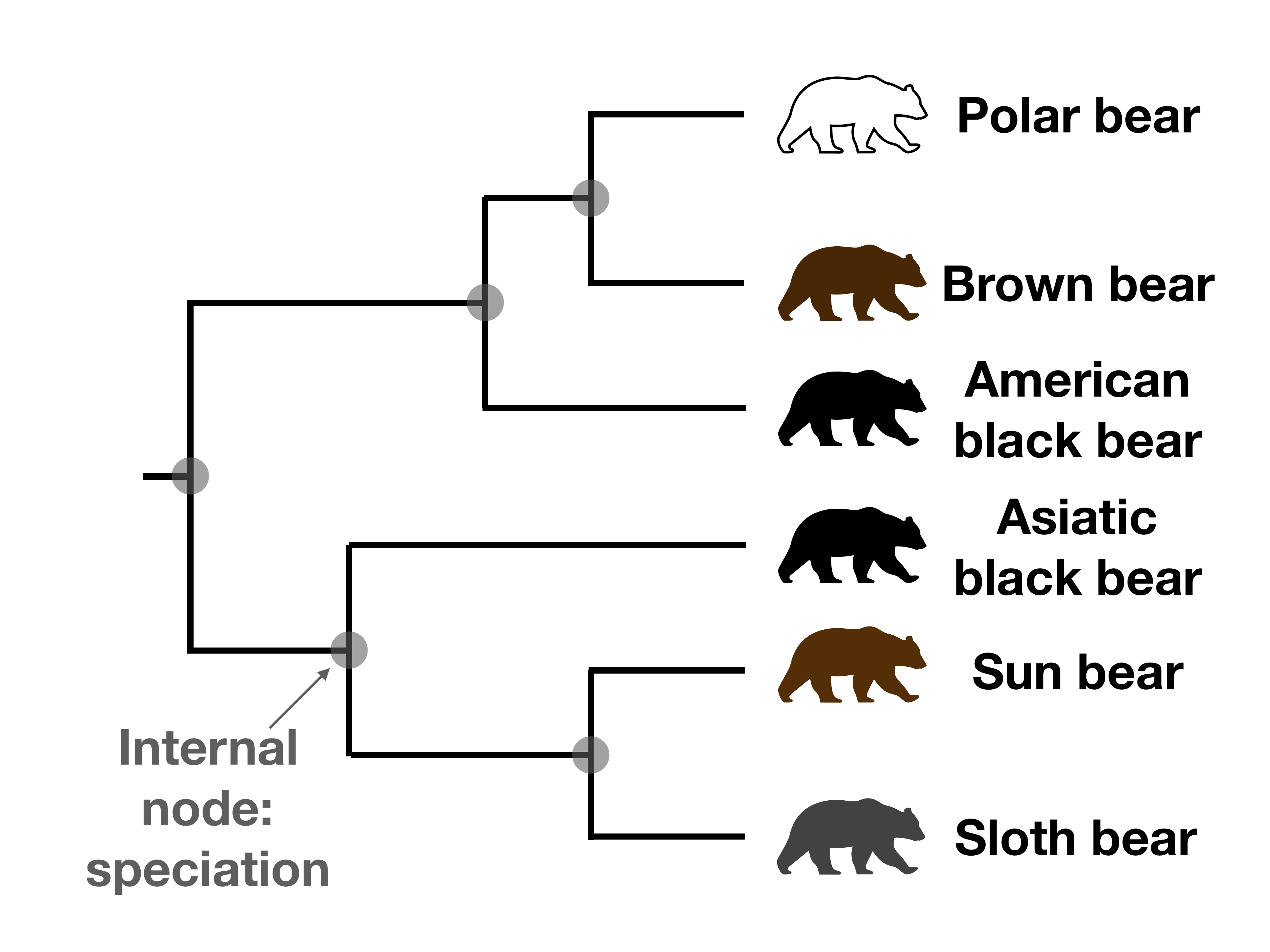}
\includegraphics[scale=0.12]{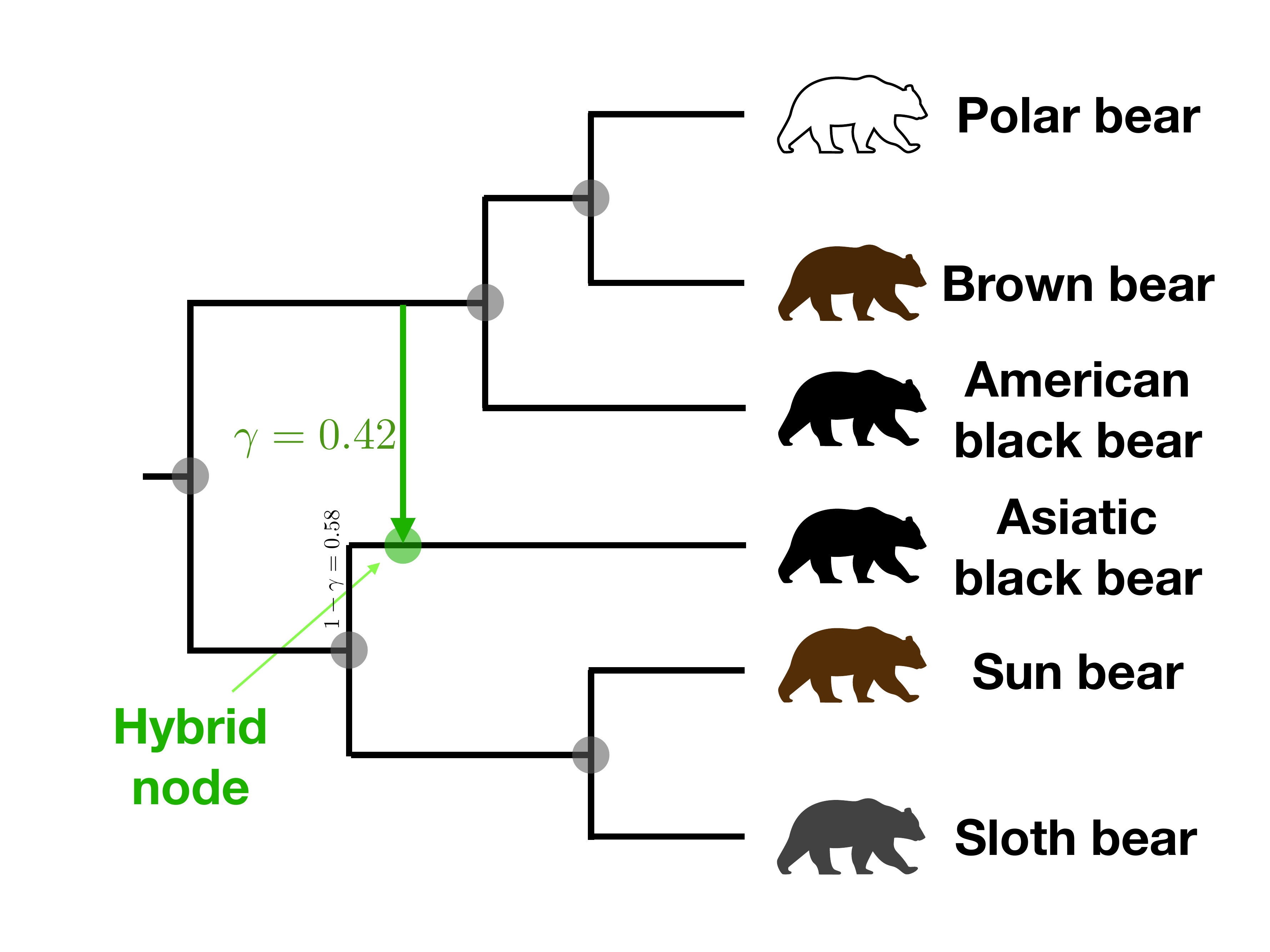}
\includegraphics[scale=0.12]{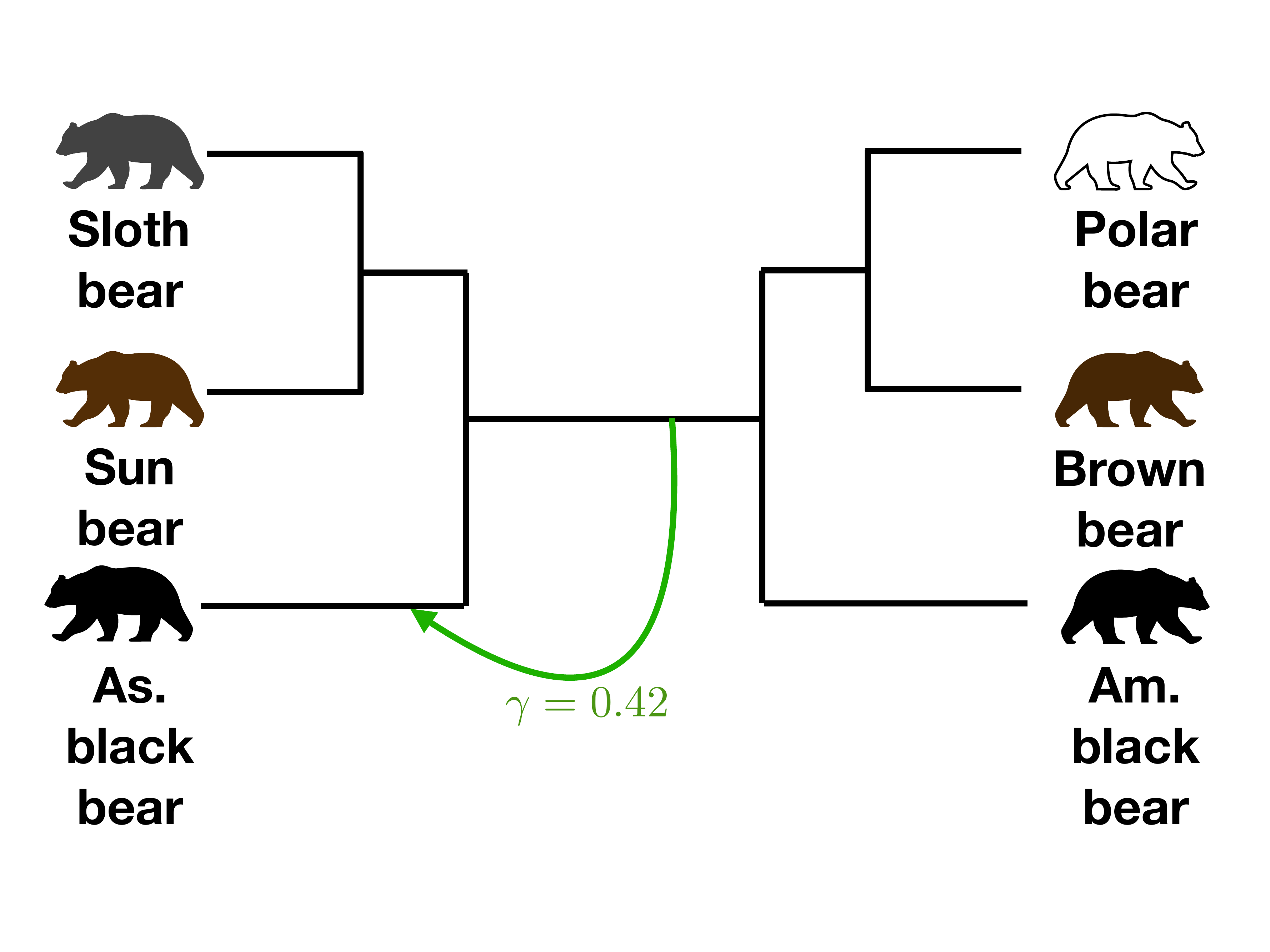}
\caption{Left: Rooted phylogenetic tree on bear
  species\cite{Kumar2017}. Internal nodes represent speciation events
  and are depicted with gray circles (usually omitted on phylogenetic
  trees).  Center: Rooted phylogenetic network on bear species
  \cite{Kumar2017} with $h=1$ hybridization event represented by the
  green arrow. Hybrid node represents a reticulation event and is
  depicted in green (usually omitted on phylogenetic networks). The
  hybrid node has two parent edges: minor hybrid edge in green labeled
  $\gamma=0.42$ and major hybrid edge in black labeled
  $1-\gamma=0.58$. The reticulation event can represent different
  biological processes: hybridization, horizontal gene transfer or
  introgression.  Right: Semi-directed network for the same biological
  scenario in center. Although the root location is unknown, its
  position is constrained by the direction of the hybrid edges. For
  example, the Asiatic black bear cannot be an outgroup.}
\label{netEx}
\end{figure}

\section{Topology identifiability: detectability of specific
  hybridizations}

Our main result characterizes which hybridization events in a level-1
phylogenetic network are detectable from a set of concordance factors
as input data under a pseudolikelihood model as in 
\cite{Solis-Lemus2016}.

Our main parameter of interest is the topology $\mathcal{N}$ of a
phylogenetic network along with the numerical parameters of the vector
of branch lengths ($\bs{t}$) and a vector of inheritance probabilities
($\bs{\gamma}$), describing the proportion of genes inherited by a
hybrid node from one of its hybrid parent (see Figure \ref{netEx}).
We assume this network is explicit, level-1 and semi-directed as in
\cite{Solis-Lemus2016}. See Supplementary Material (Section
\ref{parInt}) for more details on the network.

The data for our pseudolikelihood estimation method is a collection of
estimated gene trees $\{ G_i\}_{i=1}^{g}$ from $g$ loci (ortholog
region in genome with no recombination).  These gene trees are
unrooted and only topologies are considered (no branch lengths). To
account for estimation error in the gene trees, we do not consider the
gene trees directly as input data, but we summarize them into the
\textit{concordance factors}\cite{Baum2007}.  See the Supplementary
Material for more details on concordance factors (Section \ref{data}),
and on pseudolikelihood model (Section \ref{modelSec}).

\begin{myDefinition}[CF polynomials]
  Let $\mathcal{N}$ be $n$-taxon explicit level-1 semi-directed
  phylogenetic network with $h$ hybridizations. This network defines a
  set of $3 {n\choose 4}$ CF equations from the coalescent model with
  parameters $\bs{t}$ and $\bs{\gamma}$. Denote this system of
  equations as $CF(\mathcal{N},\bs{t},\bs{\gamma})$. If we change the
  variable of all branch lengths as $z_i=\exp{(-t_i)}$, then
  $CF(\mathcal{N},\bs{z},\bs{\gamma})$ is a system of polynomial
  equations.
  \label{cfeqs}
\end{myDefinition}

To show that a given hybridization is detectable, we compare the
theoretical formulas of the CFs given by a species network with $h$
hybridization events to those given by a species network without the
hybridization event of interest (that is, with $h-1$
hybridizations). We prove that these equations do not share any
feasible solutions, and thus, the same set of observed CFs cannot have
been generated by both the species network with $h$ hybridizations,
and the species network with $h-1$ hybridizations.

\begin{myDefinition}[Detectability]
  Let $\mathcal{N}$ be a $n$-taxon explicit level-1 semi-directed
  phylogenetic network with $h$ hybridizations. Let $\mathcal{N'}$ be
  a copy of network $\mathcal{N}$ without the $i^{th}$ hybridization.
\begin{itemize}
\item
  We say the $i^{th}$ hybridization event in
  $\mathcal{N}$ is \textbf{detectable} if the system of CFs from
  $\mathcal{N}$ does not match the system of CFs for $\mathcal{N'}$
  for any set of numerical parameters
  $(\bs{t},\bs{\gamma},\bs{t'},\bs{\gamma'})$.
 \item We say the $i^{th}$ hybridization event in $\mathcal{N}$ is
   \textbf{generically detectable} if the system of CFs from
   $\mathcal{N}$ matches the system of CFs for $\mathcal{N'}$ on a set
   of numerical parameters $(\bs{t},\bs{\gamma},\bs{t'},\bs{\gamma'})$
   of measure zero.
 \end{itemize}
\label{defdet}
\end{myDefinition}

Intuitively, Definition \ref{defdet} says that a hybridization is
detectable if the network that contains it and a network without it
(but identical to the original network in every other sense) cannot
produce the same set of CFs.

\begin{myDefinition}[$k$-cycle \cite{Huber2018}]
  Let $\mathcal{N}$ be a $n$-taxon explicit level-1 semi-directed
  phylogenetic network with $h=1$ hybridization. $\mathcal{N}$ is
  denoted a $\bs{k}$\textbf{-cycle network} if the
  cycle produced by the hybridization has $k$ nodes. The hybridization
  is also denoted a $k$-cycle hybridization.
\label{defcycle}
\end{myDefinition}

Figure \ref{kcycle} shows different $k$-cycle level-1 networks. The
subnetworks illustrated as triangles could have any number of
hybridization events as we will only focus on a given hybridization of
interest in the proofs of detectability.  That is, for a network with
$h>1$ hybridization events, as we focus on one hybridization at a time
to prove detectability (Definition \ref{defdet}), we will denote a
network $k$-cycle if the hybridization of interest has $k$ nodes on
the hybridization cycle.

\begin{figure*}
  \centering
  \includegraphics[scale=0.09]{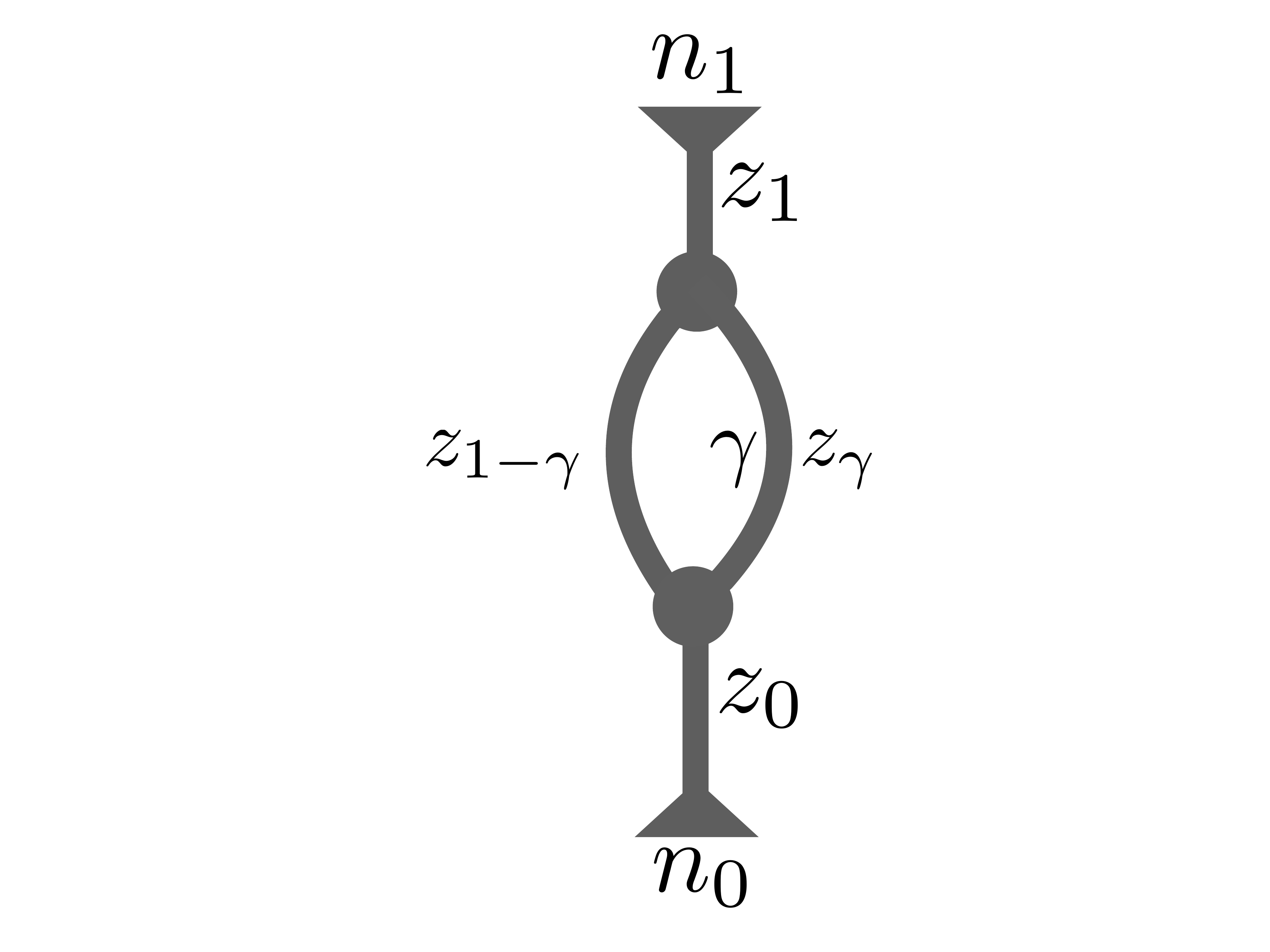}
  \includegraphics[scale=0.095]{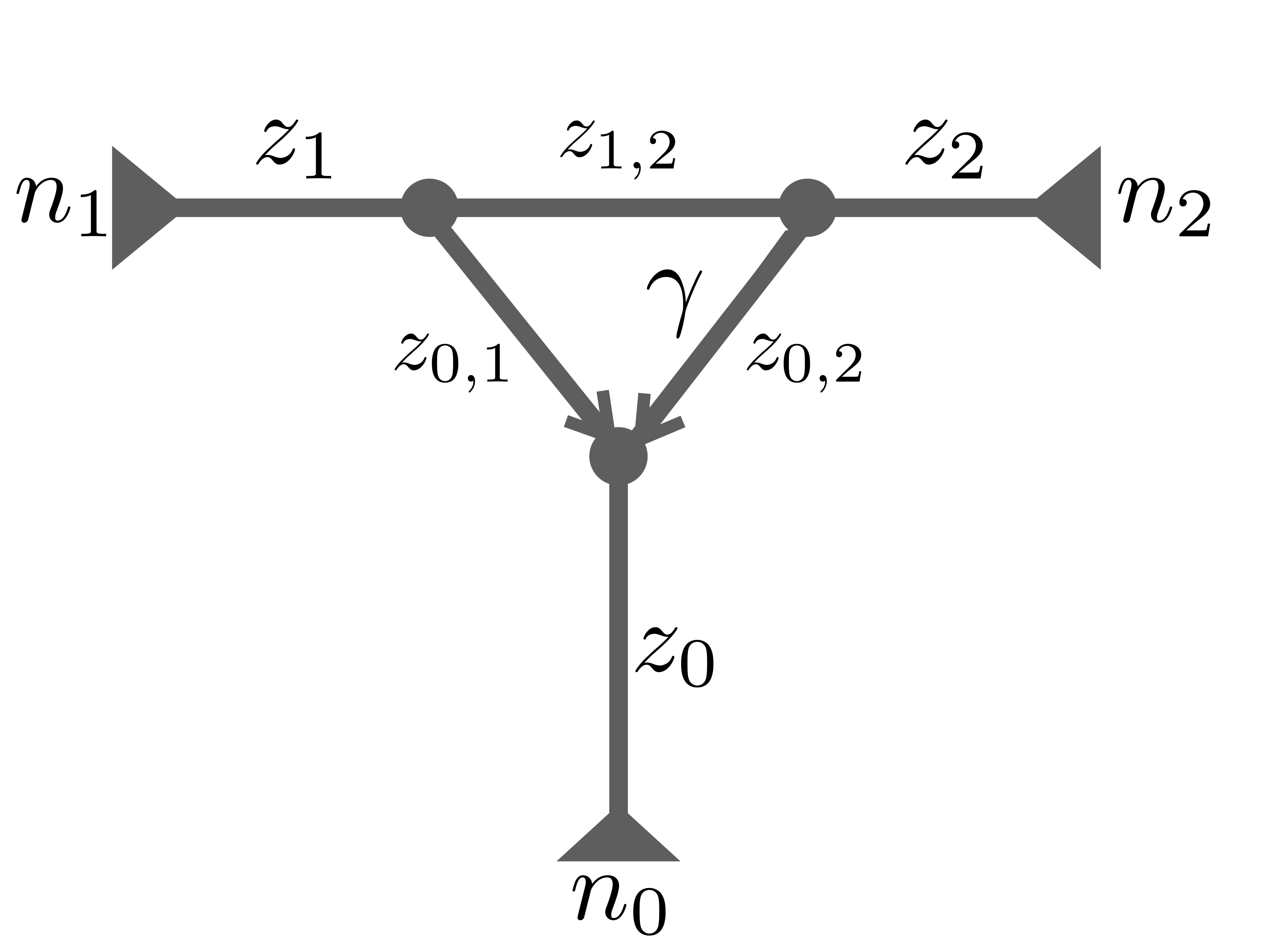}
  \includegraphics[scale=0.095]{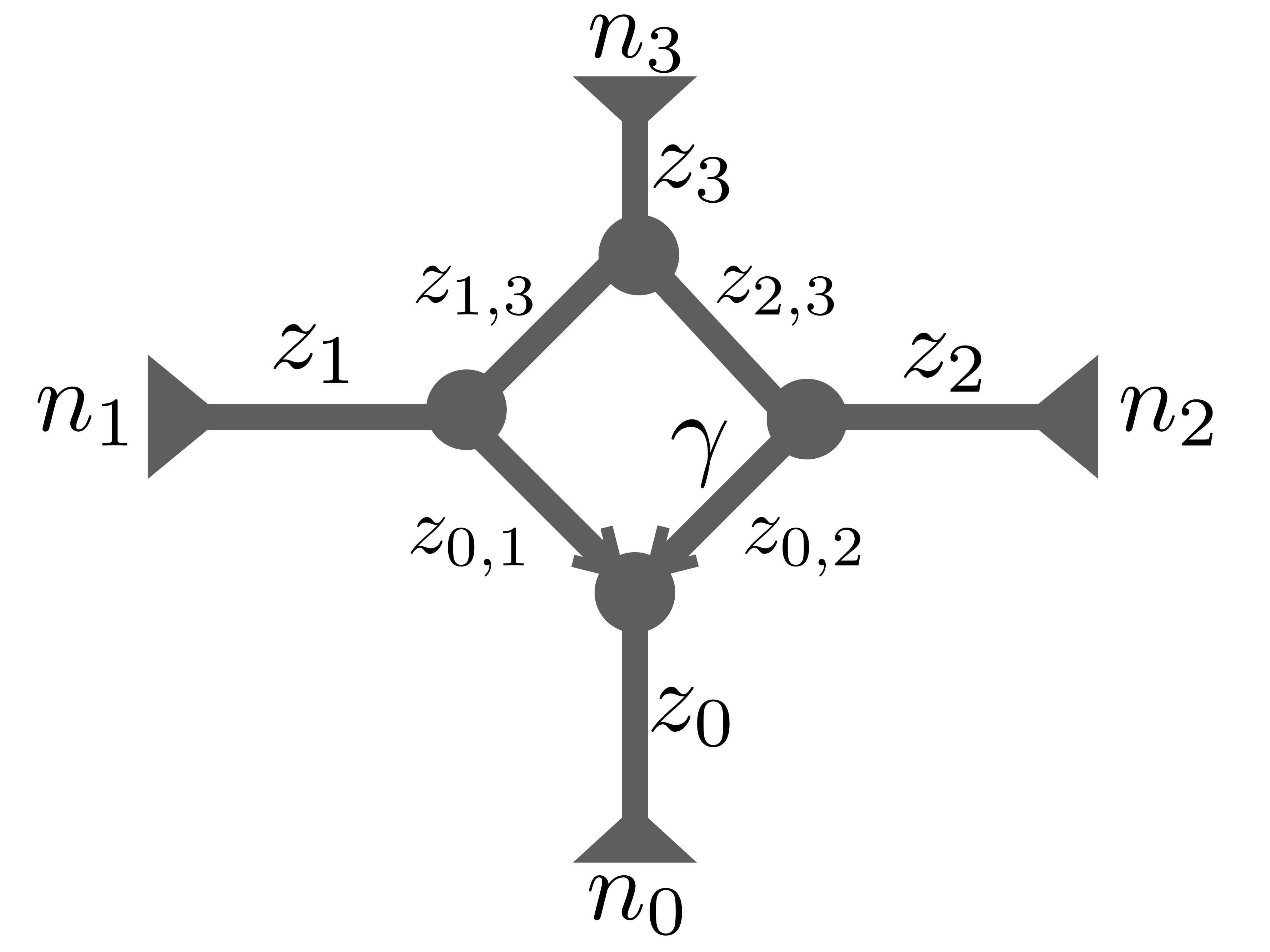}
  \includegraphics[scale=0.095]{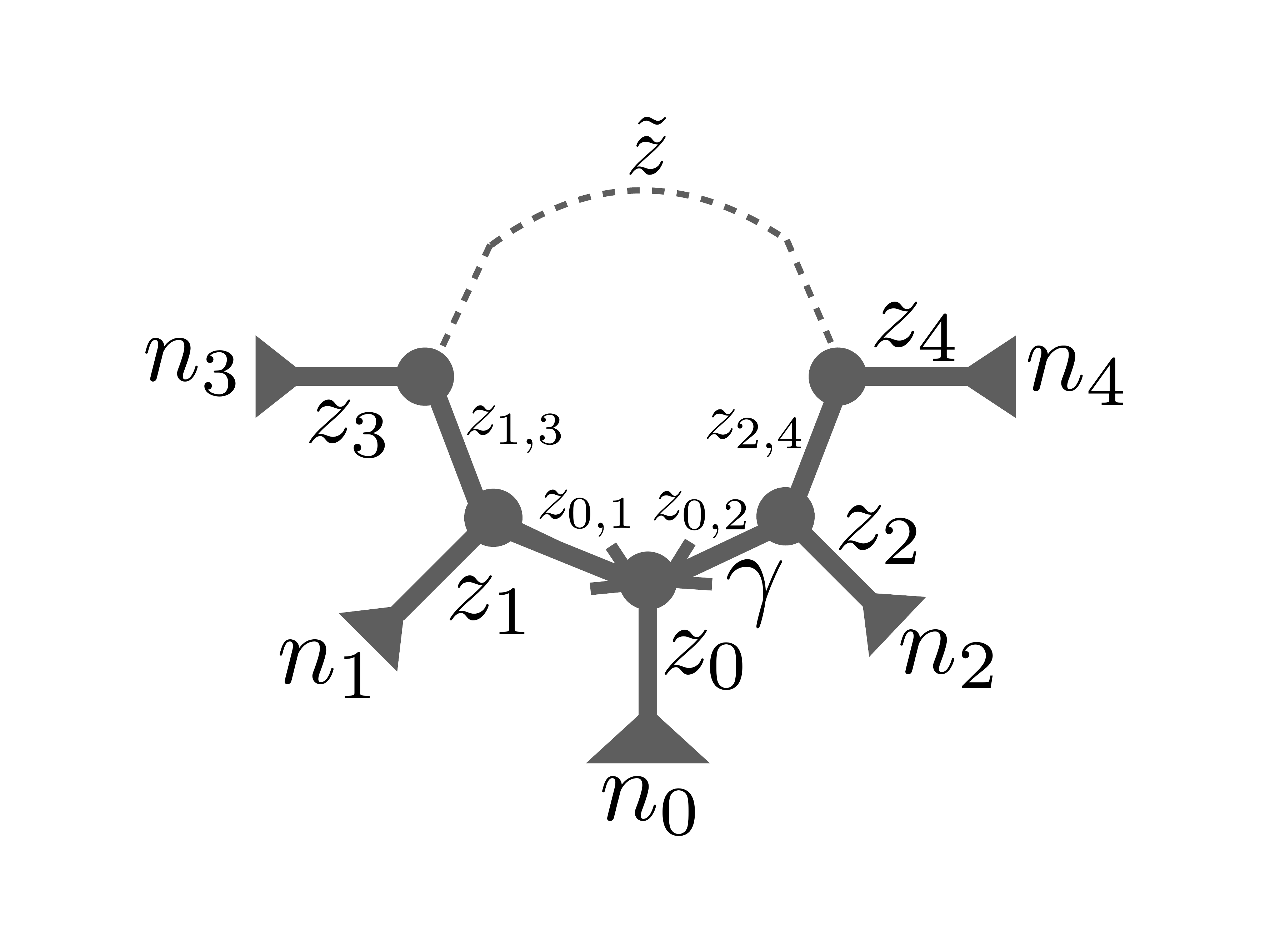}
  \begin{tabular}{cccc}
    %\textcolor{white}{.}
    $k=2$& \hspace{2cm} $k=3$ \hspace{2.5cm} & $k=4$ \hspace{2.5cm}
                                               & $k \geq 5$\\
%    &\textcolor{white}{.}
  \end{tabular}
  %   \begin{tabular}{cccc}
  %   \textcolor{white}{.}& \hspace{1cm} $k=4$ \hspace{4cm} & $k \geq 5$ \hspace{1cm}
  %   &\textcolor{white}{.}
  %   \\
  % \end{tabular}
  \caption{$k$-cycle level-1 semi-directed networks with $n$
    leaves. Here we label the edges with $z_i=\exp(-t_i)$ as will be
    used in the CF polynomial equations.}
    \label{kcycle}
\end{figure*}
Next, we list two assumptions for the main theorem on detectability of
hybridizations. The assumptions refer to a $n$-taxon explicit level-1
semi-directed phylogenetic network $\mathcal{N}$ with $h$
hybridizations.

\begin{shaded*}
\vspace{-.4cm}
\begin{itemize}
\item[(A1)]
  All branch lengths $t_i \in (0,\infty)$ for $i=1,\dots,n_e$
  and all inheritance probabilities $\gamma_j \in (0,1)$ for
  $j=1,\dots,n_h$. Here, $n_e$ represents the number of edges and
  $n_h$ the number of hybrid edges (Section \ref{parInt}).
\item[(A2)] Let the hybridization of interest define a $k$-cycle that
  divides the taxa on $k$ subsets, each with $n_k$ leaves (see Figure
  \ref{kcycle}), then we assume that $n_k \geq 2$ for every subgraph
  $k$.
\end{itemize}
\vspace{-.2cm}
\end{shaded*}

Next, we present our main theorem on topology identifiability. In
short, Theorem \ref{topId} shows that hybridization events with 3 or
more nodes in the cycle can be detected with the set of CFs as long as
we rule out inheritance probabilities of $\gamma=0$ or $1$, hard
polytomies ($t=0$) and branches without ILS ($t=\infty$).  In
addition, we need to rule out single-taxon sampled ($n_k=1$) for all
the subnetworks defined by the hybridization cycle (see Figure
\ref{kcycle}). All these cases are ruled out by assumptions (A1) and
(A2).  We will explore the single-taxon cases in Section
\ref{badcases}.

%\pagebreak

\begin{framed}
\begin{myTheorem}
\label{topId}
Let $\mathcal{N}$ be a $n$-taxon explicit level-1 semi-directed
phylogenetic network with $h$ hybridizations. Let $\mathcal{N}$ be a
$k_i$-cycle network on the $i^{th}$ hybridization.
\begin{itemize}
\item If $k_i=2$, then the $i^{th}$ hybridization is not detectable.
\item If $k_i\geq 3$, and A1-A2 hold, then the
  $i^{th}$ hybridization is detectable.
\end{itemize}
\end{myTheorem}
\end{framed}

The proof of Theorem \ref{topId} is in Section \ref{proofthm1}. 

\subsection{Single-taxon cases}
\label{badcases}

When there is a single taxon sampled from any of the subnetworks
defined by the hybridization cycle, assumption A2 is violated. That
is, this is the case when $n_k=1$ for some $k$ in Figure
\ref{kcycle}. We found that $k$-cycle hybridizations for $k \geq 4$
can be detected even if only one taxon is sampled from all the
subnetworks. However, 3-cycle hybridizations are no longer detectable,
as is presented in the following theorem.

\begin{framed}
  \begin{myTheorem}
    \label{singleThm}
    Let $\mathcal{N}$ be a $n$-taxon explicit level-1 semi-directed
    phylogenetic network with $h$ hybridizations. Let $\mathcal{N}$ be
    a $k_i$-cycle network on the $i^{th}$ hybridization. If assumption
    A1 holds:
    \begin{itemize}
    \item For $k_i=3$,
      \begin{itemize}
      \item if $n_j=n_k=1$ for any $j,k \in \{0,1,2\}$ ($j \neq k$), then this
        3-cycle hybridization is \textbf{not detectable}.
      \item if $n_j>1,n_k>1,n_l=1$ for $j,k,l \in \{0,1,2\}$, then this
        3-cycle hybridization is \textbf{generically detectable}.
      \end{itemize}
    \item For $k_i \geq 4$, if $n_k \geq 1$ for all $k$, then the
      $k_i$-cycle hybridization is \textbf{detectable}.
    \end{itemize}
    
  \end{myTheorem}
\end{framed}

The proof of Theorem \ref{singleThm} is in Section \ref{proofthm2}. In
short, there can be single taxon sampled from any of the subgraphs in
a $k$-cycle hybridization and the hybridization is still detectable
for $k \geq 4$. We note that the 3-cycle with only one single taxon
sampled on two out of the three subnetworks cannot be detected, and
this case corresponds to gene flow between sister taxa.

\section{Numerical parameters finite identifiability}

For the hybridization events that satisfy the conditions on Theorems
\ref{topId} and \ref{singleThm} and are thus detectable, we study
whether we can estimate the associated numerical parameters (Figure
\ref{kcycle}).
% We introduce a new assumption (A3) which is a
% relaxation of A2 by allowing single-taxon sampled for $k$-cycle
% networks with $k \geq 4$.

% \begin{shaded*}
% \vspace{-.4cm}
% \begin{itemize}
% \item[(A3)] \textit{Relaxation of A2.} Let $\mathcal{N}$ be a $n$-taxon
%   explicit level-1 semi-directed phylogenetic network with $h$
%   hybridizations. Let $\mathcal{N}$ be a $k$-cycle network on the
%   hybridization of interest. Assume that $n_k \geq 1$ for every subgraph $k$
%   defined by the hybridization cycle of interest.
% \end{itemize}
% \vspace{-.2cm}
% \end{shaded*}

Parameter identifiability implies that the CFs equations have a unique
solution $(\bs{t}^*,\bs{\gamma}^*)$. That is, there are not multiple
parameter values that can produce the same set of CFs. Proving
uniqueness of solution in polynomial equations is
challenging. Therefore, we will only prove \textit{finite
  identifiability}\cite{Pimentel2016}, which means that the set of CF
equations have finitely many solutions.

\begin{myDefinition}
  Let $\mathcal{N}$ be $n$-taxon explicit level-1 semi-directed
  phylogenetic network with $h$ hybridizations and let
  $\{\bs{t}_{h_i},\bs{\gamma}_{h_i}\}$ be the subset of numerical
  parameters $\{\bs{t},\bs{\gamma}\}$ in $\mathcal{N}$ that appear in
  the hybridization cycle defined by the $i^{th}$ hybridization. We
  say that these parameters are \textbf{finitely
    identifiable}\cite{Pimentel2016} if there are finitely many
  parameter values that can produce the set of CFs.
\end{myDefinition}

Next, we present our main theorem in parameter identifiability, which
describes that parameters in 3-cycle networks are not finitely
identifiable, but parameters in $k$-cycle networks for $k \geq 4$ are.

\newpage

\begin{framed}
\begin{myTheorem}
\label{parId}
Let $\mathcal{N}$ a $n$-taxon explicit level-1 semi-directed
phylogenetic network with a $k_i$-cycle on its $i^{th}$
hybridization. Let $k_i \geq 3$ to satisfy hybridization detectability
in Theorem \ref{topId}.  Let
$\theta_i = \{\bs{t}_{h_i},\bs{\gamma}_{h_i}\}$ be the subset of
numerical parameters associated with this hybridization cycle. Under
A1-A2:
%for every $k_i$, A2 hold for $k_i=3$, and A3 hold for
%$k_i\geq 4$.
\begin{itemize}
\item If $k_i=3$, then the numerical parameters in $\theta_i$ are \textbf{not finitely identifiable}.
\item If $k_i\geq 4$, then the numerical parameters in $\theta_i$ are \textbf{finitely identifiable}.
\end{itemize}
\end{myTheorem}
\end{framed}

The proof of Theorem \ref{parId} is in Section \ref{proofthm34}. In
short, Theorem \ref{parId} shows that the CFs contain enough
information to estimate the numerical parameters of hybridization
cycles with 4 or more nodes. Note that this result depends on
assumption (A2). We show in Section \ref{badcasespar} that 4-cycle
hybridizations can become not finitely identifiable when single taxa
are sampled from the subgraphs.

For the cases when the numerical
parameters are not finite identifiable ($3$-cycle networks and
single-taxon cases in Section \ref{badcasespar}), we propose a
reparametrization that will allow us to estimate a subset of the
numerical parameters in SNaQ (see \cite{Solis-Lemus2016}).

\subsection{Single-taxon cases}
\label{badcasespar}

It turns out that A2 can be violated in some cases, and parameters can
remain finitely identifiable. We only explore 4-cycle hybridizations,
since we already showed that parameters in a 3-cycle hybridization are
not even finitely identifiable under A2, and parameters for $k$-cycle
hybridizations ($k \geq 5$) can be estimated via sub-sampling with
4-cycle hybridizations (see Lemma \ref{k5par} in the Supplementary
Material Section \ref{auxsec}).
\begin{framed}
\begin{myTheorem}
\label{parIdSingle}
Let $\mathcal{N}$ be a $n$-taxon explicit level-1 semi-directed
phylogenetic network with a 4-cycle hybridization with parameters
associated to this cycle $\theta =
\{\bs{t}_{h},\bs{\gamma}_{h}\}$. Under A1:
\begin{itemize}
\item If $n_1>1$ or $n_2>1$ (Figure \ref{kcycle}), then the numerical
  parameters in $\theta$ are finitely identifiable. We denote these
  cases \textbf{good diamonds} \cite{Solis-Lemus2016}.
\item If $n_0=n_1=n_2=1,n_3>1$ (\textbf{bad diamond I}
  \cite{Solis-Lemus2016}), or if $n_0>1,n_1=n_2=n_3=1$ (\textbf{bad
    diamond II} \cite{Solis-Lemus2016}) then the numerical parameters
  in $\theta$ are not finitely identifiable.
\end{itemize}
\end{myTheorem}
\end{framed}

The proof of this theorem is combined with the proof of Theorem
\ref{parId} in Section \ref{proofthm34}. We note that the bad diamond
I had already been studied in \cite{Pickrell2012} in addition to
\cite{Solis-Lemus2016}.

\section{Discussion}
\label{conclusionsSec}

\noindent \textbf{Main findings.} We prove that hybridization cycles
of different sizes vary in their detectability potential. Cycles of 4
or more nodes are easily detected from concordance factors under a
pseudolikelihood model, while cycles of 2 nodes are totally
undetectable. 3-cycle hybridizations can be detected under certain
sampling schemes. In particular, we found that gene flow between
sister species -- common in real-life biological data -- cannot be
detected at all. We also prove that we can estimate numerical
parameters on the network (branch lengths and inheritance
probabilities) for hybridization cycles of 4 or more nodes.

Our work provides theoretical guarantees to the pseudolikelihood
estimation of larger hybridization cycles, while bringing up attention
to the need for novel models and methods to estimate gene flow between
closely related species (small cycles). Our work allows us to promote
pseudolikelihood estimation as a theoretically sound method, which is
also highly scalable and parallelizable to meet the evergrowing needs
of big genomic data.

\noindent \textbf{Limitations and Future work.} We restricted our work
to the class of level-1 networks. Our findings cannot be extended
directly to a broader class of networks, like level-k or tree-child
networks\cite{Huson2010}, because this assumption is crucial for the
proof of Lemma \ref{h1} which allows us to study hybridization events
separately from each other. More work needs to be done in order to
characterize identifiable networks of broader classes, as
hybridization events would more readily interact with one another,
creating a messier pattern in the CF equations. In addition, our work
does not tackle the identifiability of the network topology as a
whole, but rather the detectability of hybridization events. Future
work will try to provide theoretical guarantees for identifiability on
the whole topology which involve comparing whether different networks
can produce similar signal in the data.

Finally, our work was restricted to the pseudolikelihood model due to
its scalability potential. Yet the identifiability of the likelihood
model remains under studied to this day. In the future, we will
explore the identifiability of the likelihood model, which we suspect
can have the potential to estimate smaller cycles (3-cycles) that the
pseudolikelihood cannot.

\noindent \textbf{Comparison to \cite{Solis-Lemus2016} and
  \cite{Banos2019}.} As already mentioned, in 2016,
\cite{Solis-Lemus2016} published the first pseudolikelihood estimation
method for phylogenetic networks. In this work, we studied the
detectability of hybridization events and the estimability of
numerical parameters for the restricted case of level-1 networks. Many
of the mathematical proofs and derivations had to be excluded from the
manuscript because it was published in a biological journal. The lack
of mathematical details made it almost impossible to build upon that
work to extend to broader classes of networks. Our manuscript fills
this gap by providing the mathematical details of the identifiability
proofs in \cite{Solis-Lemus2016} in hopes to open the door to future
developments on this important topic.  In 2019, \cite{Banos2019}
published identifiability proofs for level-1 network
topologies. \cite{Banos2019} falsely claims that
\cite{Solis-Lemus2016} ``[...] arguments do not include investigations
on network properties such as cycle sizes'', yet \cite{Banos2019}
corroborates the results in \cite{Solis-Lemus2016} under a different
theoretical approach (graph theory), and identifies a combinatorial
strategy for inferring networks which later appeared in
\cite{Allman2019}. In addition to topology identifiability,
\cite{Solis-Lemus2016} and this present work prove finite
identifiability of numerical parameters (branch lengths and
inheritance probabilities), whereas \cite{Banos2019} focus exclusively
on topology identifiability.

% \subsection*{Broader Impact.}  Scientists world-wide are putting
% together massive efforts to understand how the biodiversity that we
% see on Earth evolved from a single-cell organism at the origin of
% life. This diversification process is represented through the Tree of
% Life, and no shortage of data, money or time are spared to reconstruct
% this tree from genomic data. Achieving the goal of estimating the Tree
% of Life would not only represent the greatest accomplishment in the
% history of evolutionary biology and systematics, but it would also
% allow us to fully understand the development and evolution of
% important biological traits in nature, in particular, those related to
% resilience to extinction when exposed to environmental threats such as
% climate change. That is, the Tree of Life would allow us not only to
% estimate ancestral traits in past extinct populations, but also to
% identify characteristics in fast-adapting organisms, and utilize this
% knowledge to bio-engineer assistance to those species in risk of
% extinction. The development of statistical and machine-learning theory
% that can help reconstruct the Tree of Life, especially those scalable
% to big data like the pseudolikelihood model, are paramount in
% evolutionary biology, systematics, and conservation efforts against mass
% extinctions.

\section{Proofs}

\subsection{Proof of Theorem \ref{topId}: detectability of
  hybridizations}
\label{proofthm1}

\begin{proof}[Proof of Theorem \ref{topId}.]
  By Lemma \ref{h1} (Section \ref{auxsec}), we know that we can focus
  on one hybridization at a time, and thus, the system of CF
  polynomial equations depend only on the parameters around the chosen
  hybridization. Without loss of generality, choose the $i^{th}$
  hybridization in $\mathcal{N}$. Let $\mathcal{N}$ be a $k_i$-cycle
  network on the $i^{th}$ hybridization. Let $\mathcal{T}$ be the
  subnetwork of $\mathcal{N}$ obtained by removing the minor hybrid
  edge on the $i^{th}$ hybridization. The gist of the proof is that we
  will find under which parameters (branch lengths and inheritance
  probabilities) does the system of equations
  $CF(\mathcal{N},\bs{z},\bs{\gamma})=CF(\mathcal{T},\bs{w})$ has a
  solution with aid of the Mathematica software \cite{Mathematica}
  (see the Reproducibility section for information on the
  scripts). Intuitively, if there exist a set in parameter space for
  which the network and the subnetwork produce the same CFs, then the
  hybridization is not detectable.

  For ease of reading, we placed all large figures and tables in
  Section \ref{figstabs}.

  \begin{proof}[Proof for 2-cycle]
    If $k_i=2$, then we obtain the following system of polynomial
    equations by setting
    $CF(\mathcal{N},\bs{z},\bs{\gamma})=CF(\mathcal{T},\bs{w})$. Recall
    that $z_i=\exp(-t_i)$ for $t_i$ branch length in $\mathcal{N}$,
    and $w_i=\exp(-t'_i)$ for $t'_i$ branch length in $\mathcal{T}$.

  \begin{minipage}{.3\textwidth}
    \centering
    \includegraphics[scale=.1]{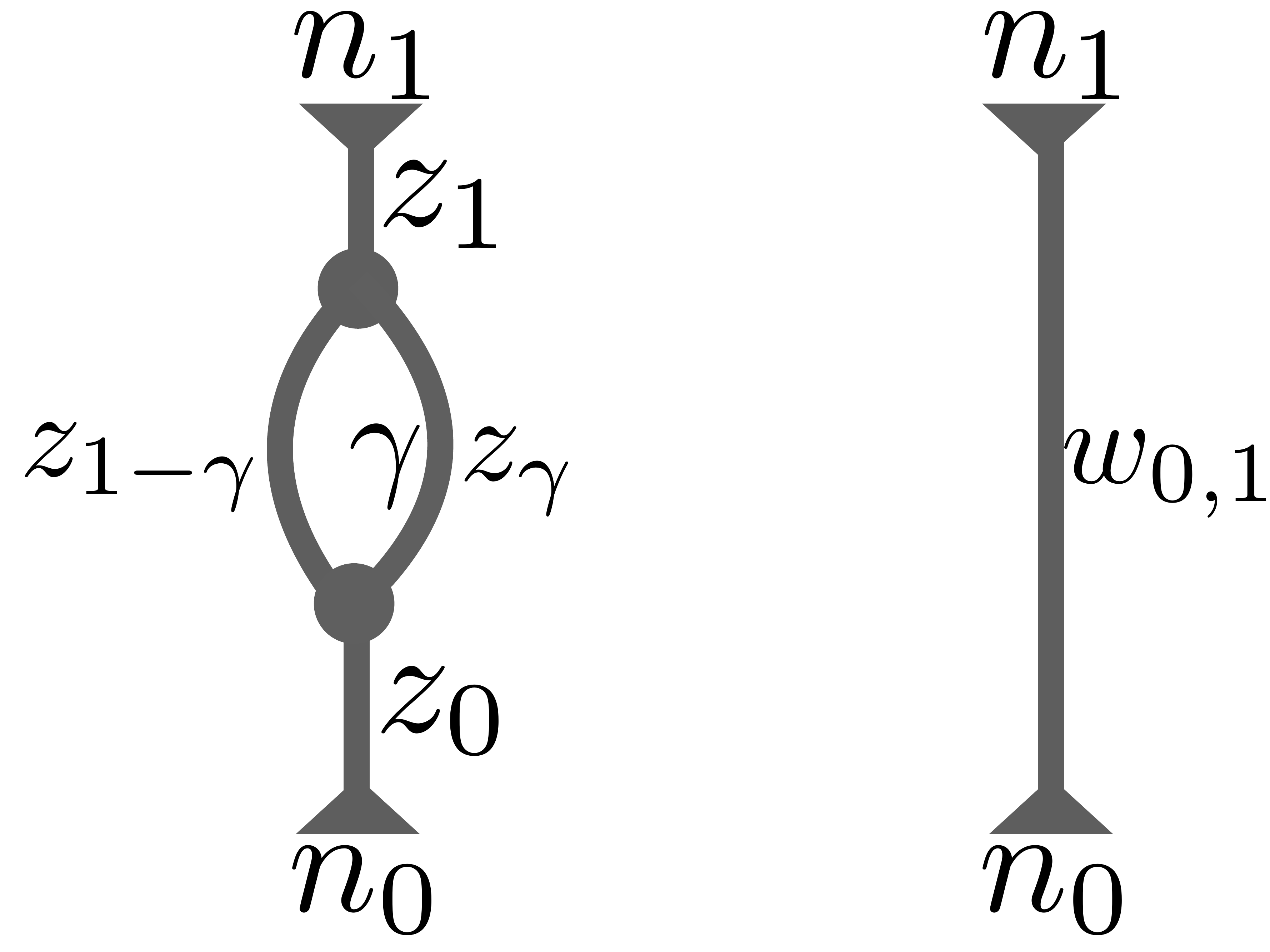}
  \end{minipage}%
  \begin{minipage}{.7\textwidth}
    \begin{align*}
      \gamma\left(1-\frac{2}{3}z_1z_{1-\gamma}z_0\right)+(1-\gamma)\left(1-\frac{2}{3}z_1z_{\gamma}z_0\right)&=1-\frac{2}{3}w_{0,1} \\ 
      \gamma\frac{1}{3}z_1z_{1-\gamma}z_0+(1-\gamma)\frac{1}{3}z_1z_{\gamma}z_0&=\frac{1}{3}w_{0,1} \\ 
      \gamma\frac{1}{3}z_1z_{1-\gamma}z_0+(1-\gamma)\frac{1}{3}z_1z_{\gamma}z_0&=\frac{1}{3}w_{0,1} 
    \end{align*}
  \end{minipage}

  This system has a solution when
  $w_{0,1} = \gamma z_1z_{1-\gamma}z_0-\gamma
  z_1z_{\gamma}z_0+z_1z_{\gamma}z_0$ for any values of
  $(\bs{z},\bs{\gamma})$. Thus, both $\mathcal{N}$ and $\mathcal{T}$
  can yield the same set of CFs for a properly chosen branch length
  $w_{0,1}$ for any values of $(\bs{t},\bs{\gamma})$. Therefore, the
  hybridization in this case is \textbf{not detectable}.
\end{proof}

For $k_i \geq 3$, by Lemma \ref{nk} (Section \ref{auxsec}), we only
need to consider the CF equations when $n_{k}=0,1,2$ in the $k_i$
subgraphs defined by the hybridization cycle. By Lemma \ref{h1}, these
subgraphs can be considered subtrees. By assumption A2, we know that
there are at least 2 taxa in each subtree, so all CF equations are
well-defined (but see Section \ref{badcases} for cases when there is
only a single taxon sampled from some of the subtrees).

  \begin{proof}[Proof for 3-cycle]
    For $k_i=3$ (Figure \ref{k3nettree}), the equations
    $CF(\mathcal{N},\bs{z},\bs{\gamma})=CF(\mathcal{T},\bs{w})$ (Table
    \ref{k3table}) have a solution if and only if
    \begin{itemize}
    \item $ \gamma=0 $ or $ \gamma=1 $ (which are the trivial cases), 
    \item or for $ \gamma\in(0,1) $, under scenarios that convert the
      3-cycle into a 2-cycle or that violate assumptions A1-A2 (see
      Figure \ref{k3path}).
\end{itemize}

That is, Figure \ref{k3path} shows that the 3-cycle hybridizations
that are not detectable for $\gamma \in (0,1)$ are those that i) by
contracting a branch on the cycle, one ends up with a 2-cycle
hybridization, or ii) by expanding certain branches to infinity, the
branch becomes ILS-free. This implies that all individuals that are
descendent to this branch (say $n_1$ in Figure \ref{k3path} top right)
will necessarily coalesce, so this scenario is equivalent to having
only one individual $n_1=1$, and this violates assumption
A2. Similarly, all non-detectable 3-cycle hybridizations violate the
A1 assumption.

\begin{shaded}
  We conclude that all non-detectable 3-cycle hybridizations
  correspond to 2-cycle or violations of assumptions A1-A2.
  Thus, assuming A1-A2, 3-cycle
  hybridizations are \textbf{detectable}.
\end{shaded}

\end{proof}
  
    \begin{proof}[Proof for 4-cycle]
      For $k_i=4$ (Figure \ref{k4nettree}), the equations
      $CF(\mathcal{N},\bs{z},\bs{\gamma})=CF(\mathcal{T},\bs{w})$
      (Tables \ref{k4table} and \ref{k4table-2}) have a solution if
      and only if
\begin{itemize}
      \item $ \gamma=0 $ or $ \gamma=1 $ (which are the trivial
        cases),
      \item or for $ \gamma\in(0,1) $, we need to have $z_{2,3}=1$
        ($t_{2,3}=0$), which necessarily takes us back to the 3-cycle
        case (see Figure \ref{k4path} left), in addition to other
        conditions (listed below) which simply result in the already
        studied undetectable cases in 3-cycle hybridizations.
\end{itemize}
      
\textbf{Set 1: 4-cycle with $\mathbf{t_{2,3}=0, \gamma \in (0,1)}$.}
The first set of additional conditions has $z_{1,3}=1$ ($t_{1,3}=0$)
which takes us back to the 2-cycle case (which we already knew was
undetectable):
\begin{enumerate}
  \item $z_{1,3}=1$
\item $z_{1,3}=1,  z_0=0 $
  \item $z_{1,3}=1,  z_2=0 $
  \item $z_{1,3}=1,  z_0=0, z_2=0$
  \item $z_{1,3}=1, \gamma^2 z_{0,2}=(1-\gamma) (\gamma z_{0,1}-2
    \gamma-z_{0,1}) $
  \item $z_{1,3}=1, z_2=0, \gamma^2 z_{0,2}=(1-\gamma) (\gamma z_{0,1}-2
    \gamma-z_{0,1}) $
  \end{enumerate}

  \textbf{Set 2: 4-cycle with $\mathbf{t_{2,3}=0, \gamma \in (0,1)}$.} The
  second set of additional conditions has $z_1=0$ ($t_1=\infty$),
  which takes us back to an undetectable 3-cycle hybridization in
  Figure \ref{k3path} (bottom right). The other condition in number 1:
  $\gamma (z_{0,2}-1)=(1-\gamma)(z_{1,3}+z_{0,1}-2)$ is equivalent to
  the one in Figure \ref{k3path} (bottom right) by replacing $z_{1,3}$
  with the corresponding branch in the 3-cycle case of $z_{1,2}$. The
  other condition on $z_{0,2}$ in number 2 is equivalent to the one in
  number 1 under the additional restriction
  (i.e. $z_{0,1}=\frac{\gamma}{\gamma-1}$).
  \begin{enumerate}
  \item $z_1=0, \gamma
    (z_{0,2}-1)=(1-\gamma)(z_{1,3}+z_{0,1}-2)$\\
  \item $z_1=0,z_{0,1}=\frac{\gamma}{\gamma-1},\gamma
    z_{0,2}=(1-\gamma)(z_{1,3}-2)$ \\
  \end{enumerate}

  \textbf{Set 3: 4-cycle with $\mathbf{t_{2,3}=0, \gamma \in (0,1)}$.} The
  third set of additional conditions has only one: $z_0=0,z_1=0$ which
  takes us to an undetectable 3-cycle case in Figure \ref{k3path}
  (center right).

  \textbf{Set 4: 4-cycle with $\mathbf{t_{2,3}=0, \gamma \in (0,1)}$.} The
  fourth set of additional conditions has a combination of
  $z_1=0,z_2=0$ which takes us to an undetectable 3-cycle case in
  Figure \ref{k3path} (top right).
  \begin{enumerate}
  \item $z_1=0, z_2=0, \gamma
    (z_{0,2}-1)=(1-\gamma)(z_{1,3}+z_{0,1}-2)$\\
  \item $z_1=0,z_2=0, z_{0,1}=\frac{\gamma}{\gamma-1},\gamma
    z_{0,2}=(1-\gamma)(z_{1,3}-2)$ \\
  \end{enumerate}

  \textbf{Set 5: 4-cycle with $\mathbf{t_{2,3}=0, \gamma \in (0,1)}$.} Finally,
  the last option for additional condition is equivalent to the
  $z_0=z_1=0$ case (Set 3 4-cycle): $z_0=0, z_1=0, z_2=0$. Even if
  $z_2=0$, since there are two subgraphs connected to the same cycle
  node (see Figure \ref{k4path} left), we are still in the case of
  $z_0=z_1=0$ in Figure \ref{k3path} (center right).

\begin{shaded}  
  We conclude that all non-detectable 4-cycle hybridizations
  correspond to 3-cycle or 2-cycle non-detectable cases, which cannot
  happen under assumptions A1-A2. Thus, assuming A1-A2, 4-cycle
  hybridizations are \textbf{detectable}.
\end{shaded}
  
\end{proof}

\begin{proof}[Proof for $k$-cycle ($k \geq 5$)]
  By Lemma \ref{k5} (Section \ref{auxsec}), we only need to focus on
  the case of $k=5$ to proof hybridization detectability for any case
  of $k \geq 5$. For $k_i=5$ (Figure \ref{k5nettree}), the equations
  $CF(\mathcal{N},\bs{z},\bs{\gamma})=CF(\mathcal{T},\bs{w})$ (Tables
  \ref{k5table-1}, \ref{k5table-2}, \ref{k5table-3}, \ref{k5table-4},
  \ref{k5table-5}) have a solution if and only if
\begin{itemize}
\item $ \gamma=0 $ or $ \gamma=1 $ (which are the trivial cases),
\item or for $ \gamma\in(0,1) $, we need to have $\tilde{z}=1$
  ($\tilde{t}=0$), and $z_{2,4}=1$ ($t_{2,4}=0$) which necessarily
  takes us back to the 3-cycle case (see Figure \ref{k4path} right),
  in addition to other conditions (listed below) which simply result
  in the already studied undetectable cases in 3-cycle hybridizations.
\end{itemize}

\textbf{Set 1: 5-cycle with
  $\mathbf{\tilde{t}=t_{2,4}=0, \gamma \in (0,1)}$.}  As in the
4-cycle, the first set of additional conditions has $z_{1,3}=1$
($t_{1,3}=0$) which takes us back to the 2-cycle case (which we
already knew was undetectable):
  \begin{enumerate}
  \item $z_{1,3}=1$
  \item $z_{1,3} = 1,z_2 = 0$
  \item $z_{1,3} = 1,z_2=0,\gamma^2 z_{0,2} = (1-\gamma) (\gamma z_{0,1}-2 \gamma-z_{0,1})$
\item $z_{1,3} = 1,z_0 = 0$
  \item $z_{1,3} = 1,z_0 = 0,z_2 = 0$
\item $z_{1,3} = 1,\gamma^2 z_{0,2} = (1-\gamma) (\gamma z_{0,1}-2
  \gamma-z_{0,1})$
    \end{enumerate}
    Note that these equations are comparable to the first set of
    additional conditions in the 4-cycle case (Set 1 4-cycle), with
    the difference that here we do not need the $z_2=1$ condition
    ($t_2=0$) which simply states that there are no coalescent events
    on this branch, as in the 5-cycle case, when we collapse the
    branches inside the cycle, three subgraphs are connected to the
    same node ($n_2,n_3,n_4$ in Figure \ref{k4path} right). So, in the
    4-cycle case, we required the $z_2=1$ so that more than two
    individuals would arrive at the cycle node. In the 5-cycle case, even
    if $z_3=z_4=z_2=1$, there would be three individuals reaching the
    cycle node. The extra condition on $z_{0,2}$ in numbers 3 and 6
    match that one in Set 1 4-cycle (number 3) as well.

    \textbf{Set 2: 5-cycle with
      $\mathbf{\tilde{t}=t_{2,4}=0, \gamma \in (0,1)}$.} The second
    set of additional conditions has $z_1=0$ ($t_1=\infty$), which
    takes us back to an undetectable 3-cycle hybridization in Figure
    \ref{k3path} (bottom right).  The condition
    $\gamma (z_{0,2}-1)=(1-\gamma)(z_{1,3}+z_{0,1}-2)$ matches the Set
    2 (4-cycle) and is similarly equivalent to the one in Figure
    \ref{k3path} (bottom right) by replacing $z_{1,3}$ with the
    corresponding branch in the 3-cycle case of $z_{1,2}$.
  \begin{enumerate}
\item $z_1=0,z_{0,1} = \frac{\gamma}{\gamma-1},\gamma z_{0,2} = (1-\gamma)(z_{1,3}-2)$
\item $z_1=0,\gamma(z_{0,2}-1) = (1-\gamma)(z_{1,3}+z_{0,1}-2)$
\end{enumerate}

\textbf{Set 3: 5-cycle with
  $\mathbf{\tilde{t}=t_{2,4}=0, \gamma \in (0,1)}$.}  Similarly to
4-cycle, the third set of additional conditions has a combination of
$z_1=0,z_0=0$ which takes us to an undetectable 3-cycle case in Figure
\ref{k3path} (center right).

\textbf{Set 4: 5-cycle with
  $\mathbf{\tilde{t}=t_{2,4}=0, \gamma \in (0,1)}$.} Similarly to
4-cycle, the fourth set of additional conditions has a combination of
$z_1=0,z_2=0$ which takes us to an undetectable 3-cycle case in Figure
\ref{k3path} (top right).
    \begin{enumerate}
\item $z_1 = 0,z_2 = 0,\gamma(z_{0,2}-1) = (1-\gamma)(z_{1,3}+z_{0,1}-2)$
\item $z_1 = 0,z_2 = 0,z_{0,1} = \frac{\gamma}{\gamma-1},\gamma
  z_{0,2} = (1-\gamma)(z_{1,3}-2)$
  \end{enumerate}

  \textbf{Set 5: 5-cycle with
    $\mathbf{\tilde{t}=t_{2,4}=0, \gamma \in (0,1)}$.} Finally, the
  last option for additional condition is equivalent to the
  $z_0=z_1=0$ case (Set 3): $z_1=0, z_2=0, z_0=0$. % 90 91
  Note that we are still in the case of $z_0=z_1=0$ in Figure \ref{k3path} (center right)
  even if $z_2=0$. This is because there are three subgraphs connected to the
  same cycle node (see Figure \ref{k4path} right).

\begin{shaded}  
  We conclude that all non-detectable 5-cycle hybridizations
  correspond to 3-cycle or 2-cycle non-detectable cases, which cannot
  happen under assumptions A1-A2. Thus, assuming A1-A2, 5-cycle
  hybridizations are \textbf{detectable}. By Lemma \ref{k5}, we
  conclude that all $k$-cycle hybridizations are \textbf{detectable}
  for $k \geq 5$.
\end{shaded}

\end{proof}

Given that the solutions to the systems
$CF(\mathcal{N},\bs{z},\bs{\gamma})=CF(\mathcal{T},\bs{w})$ for all
$k\geq 3$ depend on hard polytomies ($t_i=0$), ILS-free branches
($t_i=\infty$), or single taxon sampled for some subgraphs ($n_k=1$),
by assuming A1-A2, we can guarantee that the hybridizations are
\textbf{detectable} for $k\geq 3$.

\end{proof}

%\pagebreak

\subsection{Proof of Theorem \ref{singleThm}: detectability of
  hybridizations under single-taxon sampling}
\label{proofthm2}

\begin{proof}
  By Lemma \ref{h1} (Section \ref{auxsec}), we know that we can focus
  on one hybridization at a time, and thus, the system of CF
  polynomial equations depend only on the parameters around the chosen
  hybridization. Without loss of generality, choose the $i^{th}$
  hybridization in $\mathcal{N}$. Let $\mathcal{T}$ be the major tree
  of $\mathcal{N}$ obtained by removing the minor hybrid edge. The
  gist of the proof is that we will find under which parameters
  (branch lengths and inheritance probabilities) does the system of
  equations
  $CF(\mathcal{N},\bs{z},\bs{\gamma})=CF(\mathcal{T},\bs{w})$ has a
  solution (with aid of the Mathematica software
  \cite{Mathematica}). Intuitively, if there exist a set in parameter
  space for which the network and the tree produce the same CFs, then
  the hybridization is not detectable.

  By Lemma \ref{nk} (Section \ref{auxsec}), we only need to consider
  the CF equations when $n_{k} \leq 2$ in the $k_i$ subgraphs defined
  by the hybridization cycle.
  
  \begin{proof}[Proof for 3-cycle]
    Following the notation of Figure \ref{k3nettree} for $n_k$ for
    $k=0,1,2$, when two of the three subgraphs have only one taxon,
    the 3-cycle hybridization is not detectable because the system
    $CF(\mathcal{N},\bs{z},\bs{\gamma})=CF(\mathcal{T},\bs{w})$ has a
    solution regardless of the values of $\bs{t},\bs{\gamma}$:
    
    \begin{minipage}{.3\textwidth}
      $n_0=2,n_1=1,n_2=1$
    \end{minipage}%
    \begin{minipage}{.7\textwidth}
      The system has
      a solution when\\
      $w_0= z_0 (\gamma ^2 z_{0,1}+\gamma ^2 z_{0,2}+\gamma ^2
        z_{1,2}-2 \gamma z_{0,1}-\gamma z_{1,2}+z_{0,1}-3 \gamma ^2+3
        \gamma)$
      \end{minipage}

      \begin{minipage}{.3\textwidth}
        $n_0=1,n_1=2,n_2=1$
      \end{minipage}%
      \begin{minipage}{.7\textwidth}
        The system has a solution when $w_1= z_1 (\gamma  z_{1,2}-\gamma +1)$
      \end{minipage}
      
      \begin{minipage}{.3\textwidth}
        $n_0=1,n_1=1,n_2=2$
      \end{minipage}%
      \begin{minipage}{.7\textwidth}
        The system has a solution when $w_2= -\gamma  z_2 z_{1,2}+z_2 z_{1,2}+\gamma  z_2$
      \end{minipage}
      
      Note that these undetecable hybridization events rule out the
      potential to detect gene flow between sister taxa.

      \vspace{0.5cm}
      When only one of the three subgraphs has only
      one taxon, detectability depends on the sampling (that is, on
      which subgraph has only one sampled taxon).

      \vspace{0.25cm}
      For $n_0=2, n_1=2, n_2=1$, the equations
      $CF(\mathcal{N},\bs{z},\bs{\gamma})=CF(\mathcal{T},\bs{w})$ have
      a solution if and only if:
      \begin{itemize}
      \item $ \gamma=0 $ or $ \gamma=1 $ (which are the trivial cases),
      \item or for $ \gamma\in(0,1) $, we need to have one of the following:
        \begin{enumerate}
        \item $z_{1,2}=0$ which takes us back to the undetectable
          2-cycle hybridization,
        \item $z_1=0$ which is the same undetectable 3-cycle case as
          Figure \ref{k3path} bottom right
        \item $z_0=0$ which is the same undetectable 3-cycle case as
          Figure \ref{k3path} center left (since $n_2=1$ is equivalent
          to $z_2=0$)
        \item $(1-\gamma)(z_{0,1}-1)=\gamma(z_{0,2}+z_{1,2}-2)$ which
          is the same undetectable 3-cycle case as Figure \ref{k3path}
          bottom left (since $n_2=1$ is equivalent to $z_2=0$)
        \end{enumerate}
        The first three conditions are violated under assumption A1,
        but not the fourth one. Since the condition in number 4 has
        measure zero in parameter space, we conclude that the 3-cycle
        hybridization with taxa sampled $n_0=2, n_1=2, n_2=1$ is
        \textbf{generically detectable}.
      \end{itemize}

      \vspace{0.25cm}
      For $n_0=2, n_1=1, n_2=2$, the equations
      $CF(\mathcal{N},\bs{z},\bs{\gamma})=CF(\mathcal{T},\bs{w})$ have
      a solution if and only if:
      \begin{itemize}
      \item $ \gamma=0 $ or $ \gamma=1 $ (which are the trivial cases),
      \item or for $ \gamma\in(0,1) $, we need to have one of the
        following:
        \begin{enumerate}
        \item $z_{1,2}=0$ which takes us back to the undetectable
          2-cycle hybridization,
        \item $z_2=0$ which is the same undetectable 3-cycle case as
          Figure \ref{k3path} bottom left
        \item $z_0=0$ which is the same undetectable 3-cycle case as
          Figure \ref{k3path} center right (since $n_1=1$ is
          equivalent to $z_1=0$)
        \item $\gamma(z_{0,2}-1)=(1-\gamma)(z_{0,1}+z_{1,2}-2)$ which
          is the same undetectable 3-cycle case as Figure \ref{k3path}
          bottom right (since $n_1=1$ is equivalent to $z_1=0$)
        \end{enumerate}
        The first three conditions are violated under assumption A1,
        but not the fourth one. Since the condition in number 4 has
        measure zero in parameter space, we conclude that the 3-cycle
        hybridization with taxa sampled $n_0=2, n_1=1, n_2=2$ is
        \textbf{generically detectable}. Note that this case is
        symmetric to the previously described ($n_0=2, n_1=2, n_2=1$),
        so it makes sense that both are generically detectable.
      \end{itemize}

      \vspace{0.25cm}
      For $n_0=1, n_1=2, n_2=2$, the equations
      $CF(\mathcal{N},\bs{z},\bs{\gamma})=CF(\mathcal{T},\bs{w})$ have
      a solution if and only if:
      \begin{itemize}
      \item $ \gamma=0 $ or $ \gamma=1 $ (which are the trivial
        cases),
      \item or for $ \gamma\in(0,1) $, we need to have one of the
        following:
        \begin{enumerate}
        \item $z_{1,2}=0$ which takes us back to the undetectable
          2-cycle hybridization,
        \item $z_2=0$ which is the same undetectable 3-cycle case as
          Figure \ref{k3path} bottom left
        \item $z_1=0$ which is the same undetectable 3-cycle case as
          Figure \ref{k3path} bottom right
        \end{enumerate}
        All conditions are violated under assumption A1, and thus, we
        conclude that the 3-cycle hybridization with taxa sampled
        $n_0=1, n_1=2, n_2=2$ is \textbf{detectable}.
      \end{itemize}
    \end{proof}

    \begin{proof}[Proof for 4-cycle]
      We only need to prove that the case $n_0=n_1=n_2=n_3=1$ is
      detectable.  For $n_0=n_1=n_2=n_3=1$, the equations
      $CF(\mathcal{N},\bs{z},\bs{\gamma})=CF(\mathcal{T},\bs{w})$ have
      a solution if and only if:
      \begin{itemize}
      \item $ \gamma=0 $ (which is the trivial case),
      \item or for $ \gamma\in(0,1) $, we need to have $z_{2,3}=1$
        which takes us back to the undetectable 3-cycle case with
        $n_0=n_1=1,n_2=2$.
      \end{itemize}
      Given that the condition is violated under assumption A1, we
      conclude that the 4-cycle with only one taxon sampled from each
      subgraph is \textbf{detectable}. Note that this case does not
      have the $\gamma=1$ as possible solution. If we see the
      equations for $n=(1,1,1,1)$ in Table \ref{k4table}, we see that
      only $\gamma=0$ yields a solution in this case.
    \end{proof}

    \begin{proof}[Proof for $k$-cycle, $k \geq 5$]
      By Lemma \ref{k5}, we only need to prove detectability for the
      case $k=5$, and we only need to prove that the case
      $n_0=n_1=n_2=n_3=n_4=1$ is detectable.  For
      $n_0=n_1=n_2=n_3=n_4=1$, the equations
      $CF(\mathcal{N},\bs{z},\bs{\gamma})=CF(\mathcal{T},\bs{w})$ have
      a solution if and only if:
      \begin{itemize}
      \item $ \gamma=0 $ (which is the trivial case),
      \item or for $ \gamma\in(0,1) $, we need to have $z_{2,4}=1$ and
        $\tilde{z}=1$ which takes us back to the undetectable 3-cycle
        case with $n_0=n_1=1,n_2>1$.
      \end{itemize}
      
      Given that the conditions are violated under assumption A1, we
      conclude that the $k$-cycle with only one taxon sampled from
      each subgraph is \textbf{detectable} for every $k \geq 5$. Note
      that this case does not have the $\gamma=1$ as possible
      solution. This is true because the 4-cycle case is a subset of
      the 5-cycle case, and thus, similarly as for the 4-cycle case,
      the equations for $n=(1,1,1,1)$ in Table \ref{k4table} only have
      a matching solutions with $\gamma=0$.

    \end{proof}

    \begin{shaded}
      We conclude that any hybridization cycle with 4 or more nodes
      ($k \geq 4$) is \textbf{detectable} even if only one taxon is
      sampled from each of the subgraphs. 3-cycle hybridizations are
      \textbf{detectable} if the one taxon sampled is a child of the
      hybrid node (and more than one taxon sampled from the other two
      subgraphs), it is \textbf{generically detectable} if there are
      two or more taxa sampled from the child of the hybrid node, and
      only one in either of the other subgraphs. Finally, the 3-cycle
      hybridization is \textbf{not detectable} of two subgraphs
      contain only one taxon. This last condition rules out gene flow
      between sister taxa.
    \end{shaded}
    
  \end{proof}

%\newpage

\subsection{Proof of Theorems \ref{parId} and \ref{parIdSingle}:
  numerical parameters finite identifiability}
\label{proofthm34}

%Note that the proof of Theorem \ref{parIdSingle} is inside the same
%proof, on the case of $k$-cycle ($k \geq 4$).

\begin{proof}
  By Lemma \ref{h1} (Section \ref{auxsec}), we know that we can focus
  on one hybridization at a time. Without loss of generality, choose
  the $i^{th}$ hybridization on $\mathcal{N}$. Let
  $\mathcal{N}$ be a $k_i$-cycle network on the $i^{th}$
  hybridization. Let $\theta_i = \{\bs{t}_{h_i},\bs{\gamma}_{h_i}\}$
  be the subset of numerical parameters associated with this
  hybridization cycle.  We use the Macaulay2 software \cite{Grayson}
  for the proof as detailed below. See the Reproducibility section for
  information on the scripts.

\begin{proof}[Proof for $k_i=3$]
  The system of CF equations for a 3-cycle hybridization has 18
  polynomial equations (Table \ref{k3table}) in 7 numerical parameters
  $\{z_{0}, z_{1},z_{2}, z_{0,1},z_{1,2}, z_{0,2},\gamma\}$. We know
  from algebraic geometry \cite{Cox2007} that a system with the same
  number of algebraically independent equations as unknown parameters
  has finitely many solutions, and thus, the numerical parameters are
  finitely identifiable.

  We define the ideal $J$ of the 18 CF polynomials on 25 variables:
  the $\{a_i\}_{i=1}^{18}$ that correspond to the CF values, and the 7
  numerical parameters
  $\{z_{0},z_{1},z_{2},z_{0,1},z_{1,2},z_{0,2},\gamma\}$.

\begin{align*}
J=\{  1-\frac{2}{3} z_{1}z_{1,2}z_{2} &- a_1, \\
  \frac{1}{3} z_{1}z_{1,2}z_{2} &- a_2, \\
  \frac{1}{3} z_{1}z_{1,2}z_{2}&- a_3, \\
  (1- \gamma )\left(1-\frac{2}{3} z_{2}z_{1,2}
  \right) + \gamma \left(1-\frac{2}{3} z_{2}
  \right)&- a_4, \\
  (1- \gamma )\frac{1}{3} z_{2}z_{1,2} + \gamma \frac{1}{3}
  z_{2}&- a_5, \\
%  \phantom{ (1- \gamma )^2\left(1-\frac{2}{3}
%  z_{2}z_{0}z_{1,2}z_{0,1} \right) +2 \gamma (1-
%  \gamma )\left(1-\frac{2}{3} z_{2}z_{0} \right)
%  + \gamma ^2\left(1-\frac{2}{3}
%  z_{2}z_{0}z_{0,2} \right)}&\phantom{- a_{10}}
%\end{align*}
%\begin{align*}
  (1- \gamma )\frac{1}{3} z_{2}z_{1,2} + \gamma \frac{1}{3}
  z_{2}&- a_6, \\
  (1- \gamma )\left(1-\frac{2}{3} z_{1}
  \right) + \gamma \left(1-\frac{2}{3}
  z_{1}z_{1,2} \right)&- a_7, \\
  (1- \gamma )\frac{1}{3} z_{1} + \gamma \frac{1}{3}
  z_{1}z_{1,2}&- a_8, \\
  (1- \gamma )\frac{1}{3} z_{1} + \gamma \frac{1}{3}
  z_{1}z_{1,2}&- a_9, \\
  (1- \gamma )^2\left(1-\frac{2}{3}
  z_{2}z_{0}z_{1,2}z_{0,1} \right) +2 \gamma (1-
  \gamma )\left(1-\frac{2}{3} z_{2}z_{0} \right)
  + \gamma ^2\left(1-\frac{2}{3}
  z_{2}z_{0}z_{0,2} \right)&- a_{10}, \\
  (1- \gamma )^2\frac{1}{3} z_{2}z_{0}z_{1,2}z_{0,1}
  +2 \gamma (1- \gamma )\frac{1}{3} z_{2}z_{0} + \gamma
  ^2\frac{1}{3} z_{2}z_{0}z_{0,2} &- a_{11}, \\
  (1- \gamma )^2\frac{1}{3} z_{2}z_{0}z_{1,2}z_{0,1}
  +2 \gamma (1- \gamma )\frac{1}{3} z_{2}z_{0} + \gamma
  ^2\frac{1}{3} z_{2}z_{0}z_{0,2} &- a_{12}, \\
  (1- \gamma )^2\left(1-\frac{2}{3} z_{0}z_{0,1} \right) +2 \gamma (1-
  \gamma )\left(1- z_{0} +\frac{1}{3}
  z_{0}z_{1,2} \right) + \gamma
  ^2\left(1-\frac{2}{3} z_{0}z_{0,2} \right)&- a_{13}, \\ 
  (1- \gamma )^2\frac{1}{3} z_{0}z_{0,1} + \gamma (1-
  \gamma ) z_{0} \left(1-\frac{1}{3} z_{1,2} \right) +
  \gamma ^2\frac{1}{3} z_{0}z_{0,2} &- a_{14}, \\
  (1- \gamma )^2\frac{1}{3} z_{0}z_{0,1} + \gamma (1-
  \gamma ) z_{0} \left(1-\frac{1}{3} z_{1,2} \right) +
  \gamma ^2\frac{1}{3} z_{0}z_{0,2} &- a_{15}, \\
  (1- \gamma )^2\left(1-\frac{2}{3} z_{1}z_{0}z_{0,1} \right) +2 \gamma
  (1- \gamma )\left(1-\frac{2}{3} z_{1}z_{0}
  \right) + \gamma ^2\left(1-\frac{2}{3}
  z_{1}z_{0}z_{1,2}z_{0,2} \right)&- a_{16}, \\
  (1- \gamma )^2\frac{1}{3} z_{1}z_{0}z_{0,1} +2
  \gamma (1- \gamma )\frac{1}{3} z_{1}z_{0} + \gamma
  ^2\frac{1}{3} z_{1}z_{0}z_{1,2}z_{0,2} &- a_{17}, \\
  (1- \gamma )^2\frac{1}{3} z_{1}z_{0}z_{0,1} +2
  \gamma (1- \gamma )\frac{1}{3} z_{1}z_{0} + \gamma
  ^2\frac{1}{3} z_{1}z_{0}z_{1,2}z_{0,2} &- a_{18}\} \\
\end{align*}

It is evident that this system of equations is not algebraically
independent. For example, given that the three CFs for a given 4-taxon
subset need to sum up to one, we have that $a_1+a_2+a_3=1$ for the
first set of three. Furthermore, for tree-like 4-taxon subset, the two
minor CFs must be equal: $a_2=a_3$. Thus, focusing on the first three
CF equations, we see that there is only one degree of freedom in
$\{a_1,a_2,a_3\}$, which is enough to estimate the internal branch in
a quartet, but not enough by itself to estimate branch length and
inheritance probabilities on quarnets.

The two equations on the CF values ($a_1+a_2+a_3=1, a_2=a_3$) are two
examples of \textit{phylogenetic invariants} that the CFs need to
satisfy.  In order to know the number of algebraically independent
equations, we first need to identify the set of phylogenetic
invariants in the system.

To obtain the set of phylogenetic invariants, we use Macaulay2 to
obtain the ideal generated by eliminating the numerical
parameters. The resulting ideal is given by:
$I=\{a_1 +2 a_3 -1, a_2 -a_3, a_4 +2 a_6 -1, a_5 -a_6, a_7 +2 a_9 -1 ,
a_8 -a_9, a_{10} +2 a_{12} -1, a_{11} -a_{12}, a_{13} +2 a_{15} -1,
a_{14} -a_{15}, a_{16} +2 a_{18} -1, a_{17} -a_{18}\}$.

We denote $I$ as the set of 12 phylogenetic invariants that the CFs
need to satisfy on a 3-cycle hybridization. With Macaulay2, we obtain
that the dimension of $I$ is equal to 6, restricted to the space of
the $\{a_i\}_{i=1}^{18}$. The dimension of an ideal is the maximum
length of a regular sequence: set of algebraically independent
polynomials\cite{Cox2007}.

Thus, we have 6 algebraically independent equations in 7 parameters
$\{z_{0},z_{1},z_{2},z_{0,1},z_{1,2},z_{0,2},\gamma\}$, and thus, the
numerical parameters in the 3-cycle hybridization are \textbf{not
  finitely identifiable}.

%% -------------------
%% K3_allN222_out.txt
%% -------------------
%% 25 parameters: 18 a's, 7 z's
%% 18 equations (ideal) -> eliminate(z's), we get ideal with 12
%% invariants (G)
%% dim G = 13 (number of free parameters) => 25-13=12 fixed parameters
%% <=> 12 independent invariants
%% then, number of indep eqs = total eqs - num indep inv = 18-12 = 6

%% on the other side: 25-13=12, which should match to 18-dimG' if we
%% take the dimG on the space of 18 variables only (a's) => dimG'=6
%% this is because the number of free parameters should match:
%% 25-13 = 18-dimG'
%% and number of alg indep eqs: y = 18 - (25-13) = 18 - (18-dimG')
%% y=dimG'

%% thus, we can get the number of indep eqs in two ways:
%% 1) Total num of variables (25) - dimG (13) = number of indep
%% invariants (12); then Number of eqs (18) - number of indep
%% invariants (12) = 6!
%% 2) Find dimG' on the restricted space of only the a's (18):
%% Number of variables(25) - dimG (13) = Number of a's (18) - dimG'
%% => dimG'=6!

\end{proof}

\begin{proof}[$k_i$-cycle for $k_i \geq 4$]
  By Lemma \ref{k5par} (Section \ref{auxsec}), it is enough to prove
  finite identifiability of the parameters of a 4-cycle hybridization
  because any $k$-cycle hybridization ($k \geq 4$) can be reduced to a
  4-cycle by sub-sampling of taxa. We will show that by sampling two
  taxa from each of the four subgraphs $n_0, n_1, n_2, n_3$, we have
  enough equations to estimate all the parameters in the 4-cycle
  hybridization:
  $\{z_0,z_1,z_2,z_3,z_{0,1},z_{1,3},z_{2,3},z_{0,2},\gamma\}$. In the
  following proofs, we follow the same structure as the proof for
  $k_i=3$, but not as detailed. See proof of $k_i=3$ for more details.
  
  \noindent \textbf{Good diamond I:} $n_0=1,n_1=2,n_2=1,n_3=1$.  Here,
  we sample two individuals from the subgraph $n_1$ in Figure
  \ref{k4nettree}, and only one from the remaining subgraphs.

  We have 12 CF polynomials in the ideal on 16 variables:
  
  $\{a_{13}, a_{14}, a_{15}, a_{28}, a_{29}, a_{30},a_{34}, a_{35},
  a_{36},a_{37}, a_{38}, a_{39}\}$ and
  $\{z_{2,3}, z_{1,3}, z_1, \gamma\}$.

  \begin{align*}
    J=\{
 1-(2/3)z_{1,3}z_1-a_{13},\\
     (1/3)z_{1,3}z_1-a_{14},\\
     (1/3)z_{1,3}z_1-a_{15},\\
     (1-\gamma)(1-(2/3)z_{1,3})+\gamma(1/3)z_{2,3}-a_{28},\\
     (1-\gamma)(1/3)z_{1,3}+\gamma(1-(2/3)z_{2,3})-a_{29},\\
     (1-\gamma)(1/3)z_{1,3}+\gamma(1/3)z_{2,3}-a_{30},\\
     (1-\gamma)(1-(2/3)z_1)+\gamma(1-(2/3)z_{2,3}z_{1,3}z_1)-a_{34},\\
     (1-\gamma)(1/3)z_1+\gamma(1/3)z_{2,3}z_{1,3}z_1-a_{35},\\
     (1-\gamma)(1/3)z_1+\gamma(1/3)z_{2,3}z_{1,3}z_1-a_{36},\\
     (1-\gamma)(1-(2/3)z_1)+\gamma(1-(2/3)z_{1,3}z_1)-a_{37},\\
     (1-\gamma)(1/3)z_1+\gamma(1/3)z_{1,3}z_1-a_{38},\\
     (1-\gamma)(1/3)z_1+\gamma(1/3)z_{1,3}z_1-a_{39} \}
  \end{align*}
  We note that the set of $a$ values is not consecutive. This is
  because we use the order in Tables \ref{k4table} and \ref{k4table-2}
  to enumerate the $a_i$ so that they are consistent in all the
  4-cycle sampling cases.

  In this case, we have a set of 8 phylogenetic invariants:
\begin{align*}
  I=\{a_{38} - a_{39}, \\
  a_{37} + 2a_{39} - 1,\\
  a_{35} - a_{36},\\
  a_{34} + 2a_{36} - 1,\\
  a_{28} + a_{29} + a_{30} - 1,\\
  a_{14} - a_{15}, \\
  a_{13} + 2a_{15} - 1, \\
  a_{15}a_{29} - a_{15}a_{30} + a_{36} - a_{39}\}
\end{align*}
which result in four algebraically independent equations
($\mathrm{dim}(I)=4$ with Macaulay2), which are enough to solve for
the four numerical parameters $\{z_{2,3},z_{1,3},z_1,\gamma\}$.  Thus,
we are able to estimate these four parameters with this 4-cycle
sampling scheme.
%% -------------------
%% K4_allN1211_out.txt
%% -------------------
%% 12 equations, 16 variables, 8 invariants => dimG=8 free
%% 16vars-8=8 independent invariants
%% 12eqs-8=4 indep eqs

\noindent \textbf{Good diamond II:} $n_0=1,n_1=1,n_2=2,n_3=1$.  Here,
we sample two individuals from the subgraph $n_2$ in Figure
\ref{k4nettree}, and only one from the remaining subgraphs.

We have 12 CF polynomials in the ideal on 16 variables:

$\{a_4,a_5,a_6,a_{19},a_{20},a_{21},a_{25},a_{26},a_{27},a_{28},a_{29},a_{30}\}$
and $\{z_{2,3},z_{1,3},z_2,\gamma\}$.

\begin{align*}
    J=\{
    1-(2/3)z_{2,3}z_2-a_4,\\
     (1/3)z_{2,3}z_2-a_5,\\
     (1/3)z_{2,3}z_2-a_6,\\
     (1-\gamma)(1-(2/3)z_{2,3}z_2)+\gamma(1-(2/3)z_2)-a_{19},\\
     (1-\gamma)(1/3)z_{2,3}z_2+\gamma(1/3)z_2-a_{20},\\
     (1-\gamma)(1/3)z_{2,3}z_2+\gamma(1/3)z_2-a_{21},\\
     (1-\gamma)(1-(2/3)z_{1,3}z_{2,3}z_2)+\gamma(1-(2/3)z_2)-a_{25},\\
     (1-\gamma)(1/3)z_{1,3}z_{2,3}z_2+\gamma(1/3)z_2-a_{26},\\
     (1-\gamma)(1/3)z_{1,3}z_{2,3}z_2+\gamma(1/3)z_2-a_{27},\\
     (1-\gamma)(1-(2/3)z_{1,3})+\gamma(1/3)z_{2,3}-a_{28},\\
     (1-\gamma)(1/3)z_{1,3}+\gamma(1-(2/3)z_{2,3})-a_{29},\\
     (1-\gamma)(1/3)z_{1,3}+\gamma(1/3)z_{2,3}-a_{30} \}
  \end{align*}

  In this case, we have a set of 8 phylogenetic invariants:
  \begin{align*}
  I=\{  a_{28} + a_{29} + a_{30} - 1, \\
    a_{26} - a_{27}, \\
    a_{25} + 2a_{27} - 1, \\
    a_{20} - a_{21}, \\
    a_{19} + 2a_{21} - 1, \\
    a_{5} - a_6, \\
    a_{4} + 2a_6 - 1, \\
    a_{6}a_{29} + 2a_6a_{30} - a_6 + a_{21} - a_{27} \}
\end{align*}
which result in four algebraically independent equations
($\mathrm{dim}(I)=4$), which are enough to solve for the four
numerical parameters $\{z_{2,3},z_{1,3},z_2,\gamma\}$. Since we
already estimated $\{z_{2,3},z_{1,3},\gamma\}$ from the Good Diamond
I, we only need to estimate $z_2$ here.
%% -------------------
%% K4_allN1121_out.txt
%% -------------------
%% 12 equations, 16 variables, 8 invariants => dimG=8 free
%% 16vars-8=8 independent invariants
%% 12eqs-8=4 indep eqs

\noindent \textbf{Bad diamond I:} $n_0=1,n_1=1,n_2=1,n_3=2$.  Here, we
sample two individuals from the subgraph $n_3$ in Figure
\ref{k4nettree}, and only one from the remaining subgraphs.

We have 12 CF polynomials in the ideal on 16 variables:

$\{a_7,a_8,a_9,a_{22},a_{23},a_{24},a_{28},a_{29},a_{30},a_{31},a_{32},a_{33}\}$
and $\{z_{2,3},z_{1,3},z_3,\gamma\}$.

\begin{align*}
J=\{  1-(2/3)z_3-a_7,\\
     (1/3)z_3-a_8,\\
     (1/3)z_3-a_9,\\
     (1-\gamma)(1-(2/3)z_3)+\gamma(1-(2/3)z_{2,3}z_3)-a_{22},\\
     (1-\gamma)(1/3)z_3+\gamma(1/3)z_{2,3}z_3-a_{23},\\
     (1-\gamma)(1/3)z_3+\gamma(1/3)z_{2,3}z_3-a_{24},\\
     (1-\gamma)(1-(2/3)z_{1,3})+\gamma(1/3)z_{2,3}-a_{28},\\
     (1-\gamma)(1/3)z_{1,3}+\gamma(1-(2/3)z_{2,3})-a_{29},\\
     (1-\gamma)(1/3)z_{1,3}+\gamma(1/3)z_{2,3}-a_{30},\\
     (1-\gamma)(1-(2/3)z_{1,3}z_3)+\gamma(1-(2/3)z_3)-a_{31},\\
     (1-\gamma)(1/3)z_{1,3}z_3+\gamma(1/3)z_3-a_{32},\\
     (1-\gamma)(1/3)z_{1,3}z_3+\gamma(1/3)z_3-a_{33} \}
  \end{align*}

  In this case, we have a set of 10 phylogenetic invariants:
\begin{align*}
 I-\{ a_{32} - a_{33}, \\
  a_{31} + 2a_{33} - 1, \\
  a_{28} + a_{29} + a_{30} - 1, \\
  a_{23} - a_{24}, \\
  a_{22} + 2a_{24} - 1, \\
  a_{8} - a_{9}, \\
  a_{7} + 2a_{9} - 1, \\
  3a_9a_{30} + a_9 - a_{24} - a_{33}, \\
  a_{24}a_{29} + 2a_{24}a_{30} + a_{29}a_{33} - a_{30}a_{33} - a_{33}, \\
  3a_{9}a_{29} - 2a_9 + 2a_{24} - a_{33}\}
  \end{align*}
  which result in only three algebraically independent equations
  ($\mathrm{dim}(I)=3$), not enough to solve for the four numerical
  parameters $\{z_{2,3},z_{1,3},z_3,\gamma\}$.  However, given that we
  already estimated $\{z_{2,3},z_{1,3}, \gamma\}$ from Good Diamond I,
  we only need to estimate $z_{3}$.
%% -------------------
%% K4_allN1112_out.txt
%% -------------------
%% 12 equations, 16 variables, 10 invariants => dimG=7 free
%% 16vars-7=9 independent invariants
%% 12eqs-9=3 indep eqs

  \noindent \textbf{Bad diamond II:} $n_0=2,n_1=1,n_2=1,n_3=1$.  Here,
  we sample two individuals from the subgraph $n_0$ in Figure
  \ref{k4nettree}, and only one from the remaining subgraphs.

  We have 12 CF polynomials in the ideal on 18 variables:

  $\{a_{28},a_{29},a_{30},a_{43},a_{44},a_{45},a_{49},a_{50},a_{51},a_{52},a_{53},a_{54}\}$
  and $\{z_{2,3},z_{1,3},z_{0,1},z_{0,2},z_0,\gamma\}$.

  \begin{align*}
J=\{  (1-\gamma)(1-(2/3)z_{1,3})+\gamma(1/3)z_{2,3}-a_{28},\\
     (1-\gamma)(1/3)z_{1,3}+\gamma(1-(2/3)z_{2,3})-a_{29},\\
     (1-\gamma)(1/3)z_{1,3}+\gamma(1/3)z_{2,3}-a_{30},\\
     (1-\gamma)^2(1-(2/3)z_0z_{1,3}z_{0,1})+2\gamma(1-\gamma)(1-z_0+(1/3)z_0z_{2,3})+\gamma^2(1-(2/3)z_0z_{0,2})-a_{43},\\
     (1-\gamma)^2(1/3)z_0z_{1,3}z_{0,1}+\gamma(1-\gamma)z_0(1-(1/3)z_{2,3})+\gamma^2(1/3)z_0z_{0,2}-a_{44},\\
     (1-\gamma)^2(1/3)z_0z_{1,3}z_{0,1}+\gamma(1-\gamma)z_0(1-(1/3)z_{2,3})+\gamma^2(1/3)z_0z_{0,2}-a_{45},\\
     (1-\gamma)^2(1-(2/3)z_0z_{0,1})+2\gamma(1-\gamma)(1-z_0+(1/3)z_0z_{2,3}z_{1,3})+\gamma^2(1-(2/3)z_0z_{0,2})-a_{49},\\
     (1-\gamma)^2(1/3)z_0z_{0,1}+\gamma(1-\gamma)z_0(1-(1/3)z_{2,3}z_{1,3})+\gamma^2(1/3)z_0z_{0,2}-a_{50},\\
     (1-\gamma)^2(1/3)z_0z_{0,1}+\gamma(1-\gamma)z_0(1-(1/3)z_{2,3}z_{1,3})+\gamma^2(1/3)z_0z_{0,2}-a_{51},\\
     (1-\gamma)^2(1-(2/3)z_0z_{0,1})+2\gamma(1-\gamma)(1-z_0+(1/3)z_0z_{1,3})+\gamma^2(1-(2/3)z_0z_{0,2}z_{2,3})-a_{52},\\
     (1-\gamma)^2(1/3)z_0z_{0,1}+\gamma(1-\gamma)z_0(1-(1/3)z_{1,3})+\gamma^2(1/3)z_0z_{0,2}z_{2,3}-a_{53},\\
     (1-\gamma)^2(1/3)z_0z_{0,1}+\gamma(1-\gamma)z_0(1-(1/3)z_{1,3})+\gamma^2(1/3)z_0z_{0,2}z_{2,3}-a_{54}\}
  \end{align*}

  In this case, we have a set of 7 phylogenetic invariants:
  \begin{align*}
 I\{   a_{53} - a_{54}, \\
    a_{52} + 2a_{54} - 1, \\
    a_{50} - a_{51}, \\
    a_{49} + 2a_{51} - 1, \\
    a_{44} - a_{45}, \\
    a_{43} + 2a_{45} - 1, \\
    a_{28} + a_{29} + a_{30} - 1 \}
  \end{align*}
  which result in only five algebraically independent equations
  ($\mathrm{dim}(I)=5$), not enough to solve for the six numerical
  parameters $\{z_{2,3}, z_{1,3}, z_{0,1},z_{0,2},z_0,\gamma\}$.
  However, given that we already estimated
  $\{z_{2,3},z_{1,3}, \gamma\}$ from Good Diamond I, we only need to
  estimate $z_{0,1},z_{0,2},z_{0}$.
%% -------------------
%% K4_allN2111_out.txt
%% -------------------
%% 12 equations, 18 variables, 7 invariants => dimG=11 free
%% 18vars-11=7 independent invariants
%% 12eqs-7=5 indep eqs

  Thus, we can estimate the numerical parameters from different
  sub-sampling in the 4-cycle hybridization:
\begin{itemize}
\item Good diamond I: $z_{1,3},z_{2,3},z_1,\gamma$
\item Good diamond II: $z_2$
\item Bad diamond I: $z_3$
\item Bad diamond II: $z_{0,1},z_{0,2},z_0$
\end{itemize}

\end{proof}

\begin{shaded}
  We conclude that with the exception of 3-cycle hybridizations and
  bad diamonds, all numerical parameters associated with any level-1
  network are \textbf{finitely identifiable}.
\end{shaded}
  
\end{proof}

\subsection{Auxiliary lemmas for the proofs of theorems}
\label{auxsec}

The lemmas below are used in the proofs of the topology
identifiability and numerical parameter finite identifiability
theorems.  With Lemma \ref{nk}, we show that we only need to sample
two individuals for each subgraph defined by the hybridization cycle
(Figure \ref{kcycle}) in order to have enough polynomial equations for
all the numerical parameters in the cycle.  With Lemma \ref{type1}, we
show that some hybridizations on 4-taxon subnetworks (quarnets)
produce CFs that can be matched with a properly chosen internal branch
on a quartet tree. That is, these quarnets are not distinguishable
from a quartet tree with the CFs.  With Lemma \ref{h1}, we show that
with the level-1 assumption, we can focus on the detectability of each
hybridization at a time, without considering other hybridizations in
the network.  With Lemma \ref{k5}, we show that proving detectability
of a cycle of $k=5$ nodes is enough to prove detectability of any
cycle of $k \geq 5$ nodes.  Finally, with Lemma \ref{k5par}, we show
that we can prove parameter identifiability of any hybridization cycle
of $k \geq 4$ nodes by showing parameter identifiability of
hybridization cycles of $k=4$ nodes.

\begin{myLemma}
  Let $\mathcal{N}$ be a $n$-taxon explicit level-1 semi-directed
  phylogenetic network with $h$ hybridizations. Let
  $CF(\mathcal{N}, \bs{t}, \bs{\gamma})$ be the system of CF equations
  defined by the coalescent model (Definition \ref{cfeqs}).  Let $h_i$
  be the $i^{th}$ hybrid node in $\mathcal{N}$, which defines a
  $k_i$-cycle with $k_i$ subgraphs attached, each with $n_k$ taxa for
  $k=1,\dots,k_i$ (see Figure \ref{kcycle}).  Let
  $\theta_i = \{\bs{t}_{h_i},\bs{\gamma}_{h_i}\}$ be the subset of
  numerical parameters associated with this hybridization cycle. Let
  $CF_{h_i}(\mathcal{N}) = CF(\mathcal{N},
  \bs{t}_{h_i},\bs{\gamma}_{h_i})$ be the subset of CF equations that
  involve $\theta_i$.

  Then, we only need to sample $n_k=2$ from each subgraph to define
  all the CF equations in $CF_{h_i}(\mathcal{N})$.
  \label{nk}
\end{myLemma}

\begin{proof}
  Without loss of generality, we choose a hybridization cycle with 3
  nodes ($k_i=3$). See Figure \ref{kcycle}. This 3-cycle defines 3
  subgraphs with number of taxa $n_0,n_1,n_2$ respectively. The
  complete system of CF equations of the network is produced by taking
  every possible 4-taxon subset $\{a,b,c,d\}$ and write its
  corresponding three CF equations: $CF_{ab|cd},CF_{ac|bd},CF_{ad|bc}$
  as described in Section \ref{modelSec}. However, we are interested
  only in the CF equations that include the parameters related to the
  $i^{th}$ hybridization event:
  $\theta_i=\{z^{(i)}_1, z^{(i)}_{1,2}, z^{(i)}_2,
  z^{(i)}_{0,1},z^{(i)}_{0,2},z^{(i)}_0, \gamma^{(i)}\}$.

  Note that if we take 4 taxa from the same subgraph (say the one
  labeled $n_0$), then none of the CF equations will involve any
  parameters in $\theta_i$. Similarly, if we take 3 taxa from the same
  subgraph (say $n_0$), and one taxon from a different subgraph (say
  $n_1$), then the edge connecting the one taxon in $n_1$ to the
  others will indeed contain parameters in $\theta_i$
  ($z^{(i)}_0, z^{(i)}_{0,1}, z^{(i)}_1$). However, this edge will be
  an external edge, and as we know from the coalescent model, only
  internal edges play a role in the CF formulas. Thus, no parameters
  in $\theta_i$ will actually appear in the CF formulas when taking
  all 4 taxa from the same subgraph, or when taking 3 taxa from one
  subgraph and one from another.  This means that only the CF
  equations from 4-taxon subsets in which at most 2 taxa belong to the
  same subgraph include the parameters in $\theta_i$.  Thus, the
  subset of CF equations corresponding to the hybridization event
  ($CF_{h_i}(\mathcal{N})$) can be fully defined by having
  $n_0=n_1=n_2=2$.
\end{proof}

\begin{myLemma}
  Let $\mathcal{Q}$ be a 4-taxon explicit level-1 semi-directed
  network with one hybridization event. Assume that $\mathcal{Q}$ is
  of type 1, 2, 4 or 5 in Figure \ref{5quartets}. Let
  $CF(\mathcal{Q},\bs{t},\gamma)$ be the set of three CF equations for
  $\mathcal{Q}$. Let $\mathcal{T}_Q$ be the major quartet obtained
  by removing the minor hybrid edge ($\gamma<0.5$) in
  $\mathcal{Q}$, and let $CF(\mathcal{T}_Q,\bs{\tilde{t}})$ be the set of
  three CF equations for $\mathcal{T}_Q$.

  Then, $CF(\mathcal{Q},\bs{t},\gamma) = CF(\mathcal{T}_Q,\tilde{t})$
  for a properly chosen branch length $\tilde{t}$ in $\mathcal{T}_Q$.
  \label{type1}
\end{myLemma}

\begin{proof}
If $\mathcal{Q}$ is a type 5 network, the proof is trivial as the
CF match directly those from a quartet tree.

For type 1, 2 and 4, below we see the formulas to choose the internal
branch in a quartet tree ($\tilde{t}_1$ in the Figures) to match the
CFs of these quarnets. As in the remainder of the paper, here
$z_i=\exp(-t_i)$.

\noindent\begin{minipage}{0.35\textwidth}
\centering
\includegraphics[scale=0.1]{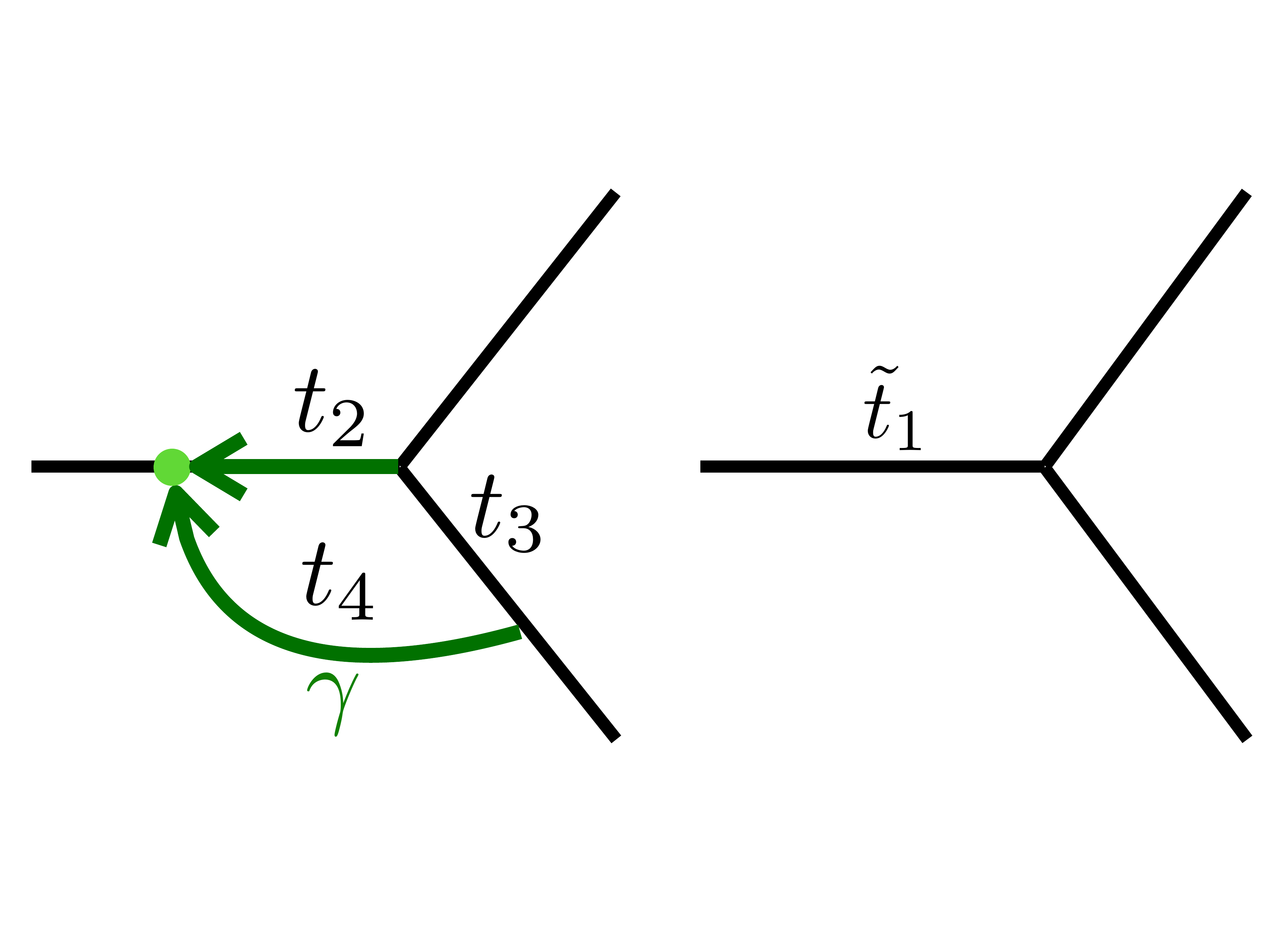}
\\Type 1
%\captionof{figure}{Quartet 1 (type 1)}\label{fig:caseI2rep}
\end{minipage}
\begin{minipage}{0.55\textwidth}
\centering{\footnotesize
$\exp(-\tilde{t}_1):=1+\gamma z_3-\gamma^2 z_4-\gamma^2 z_3-(1-\gamma)^2 z_2$}
\end{minipage}

%----------------------Case II
\noindent\begin{minipage}{0.35\textwidth}
\centering
\includegraphics[scale=0.1]{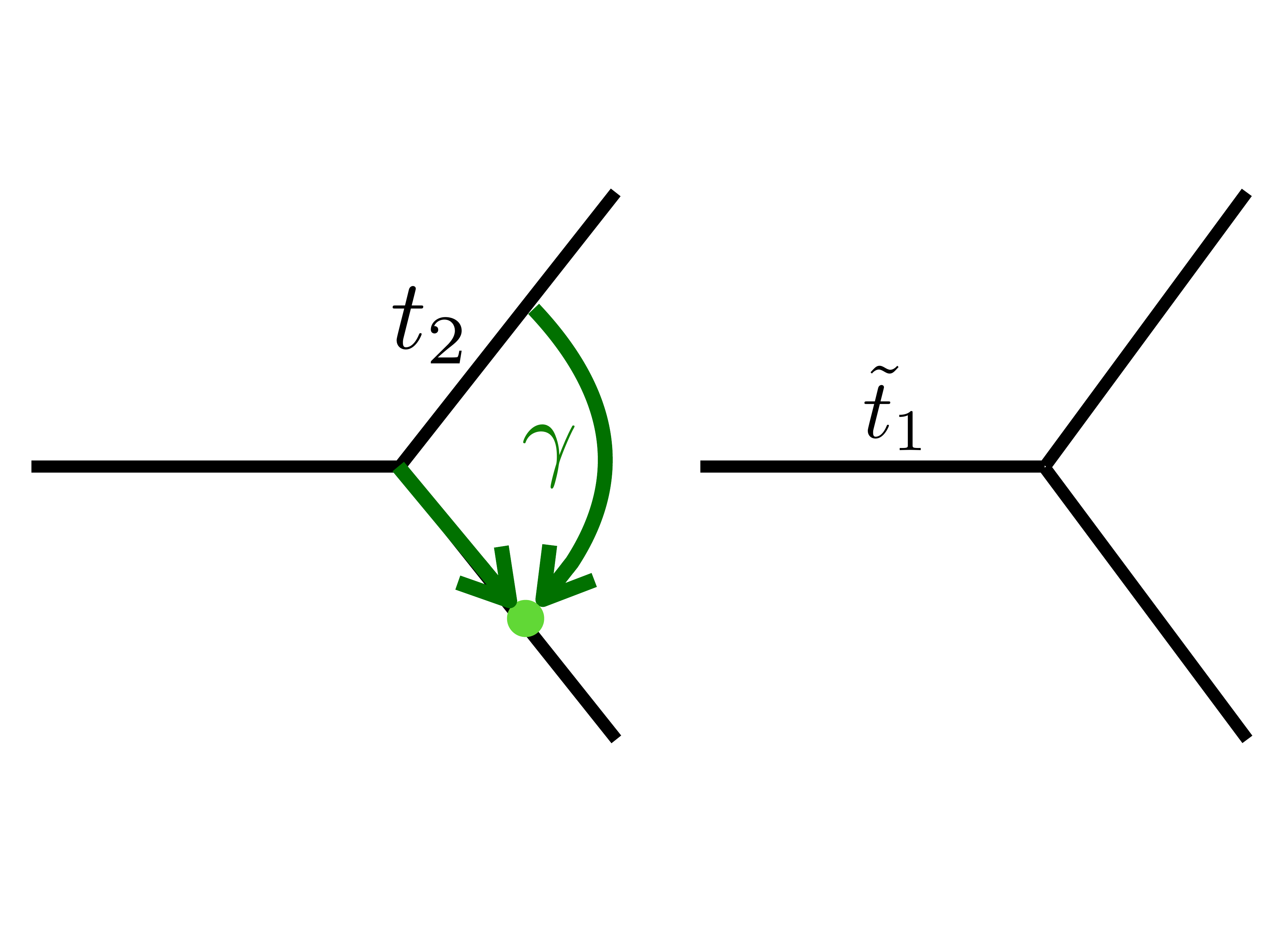}
\\Type 2
%\captionof{figure}{Quartet 2 (type 1)}\label{fig:caseIIrep}
\end{minipage}
\begin{minipage}{0.55\textwidth}
\centering{\footnotesize
$\exp(-\tilde{t}_1):=1-\gamma z_2$}
\end{minipage}

%----------------------Case IV
\noindent\begin{minipage}{0.35\textwidth}
\centering
\includegraphics[scale=0.1]{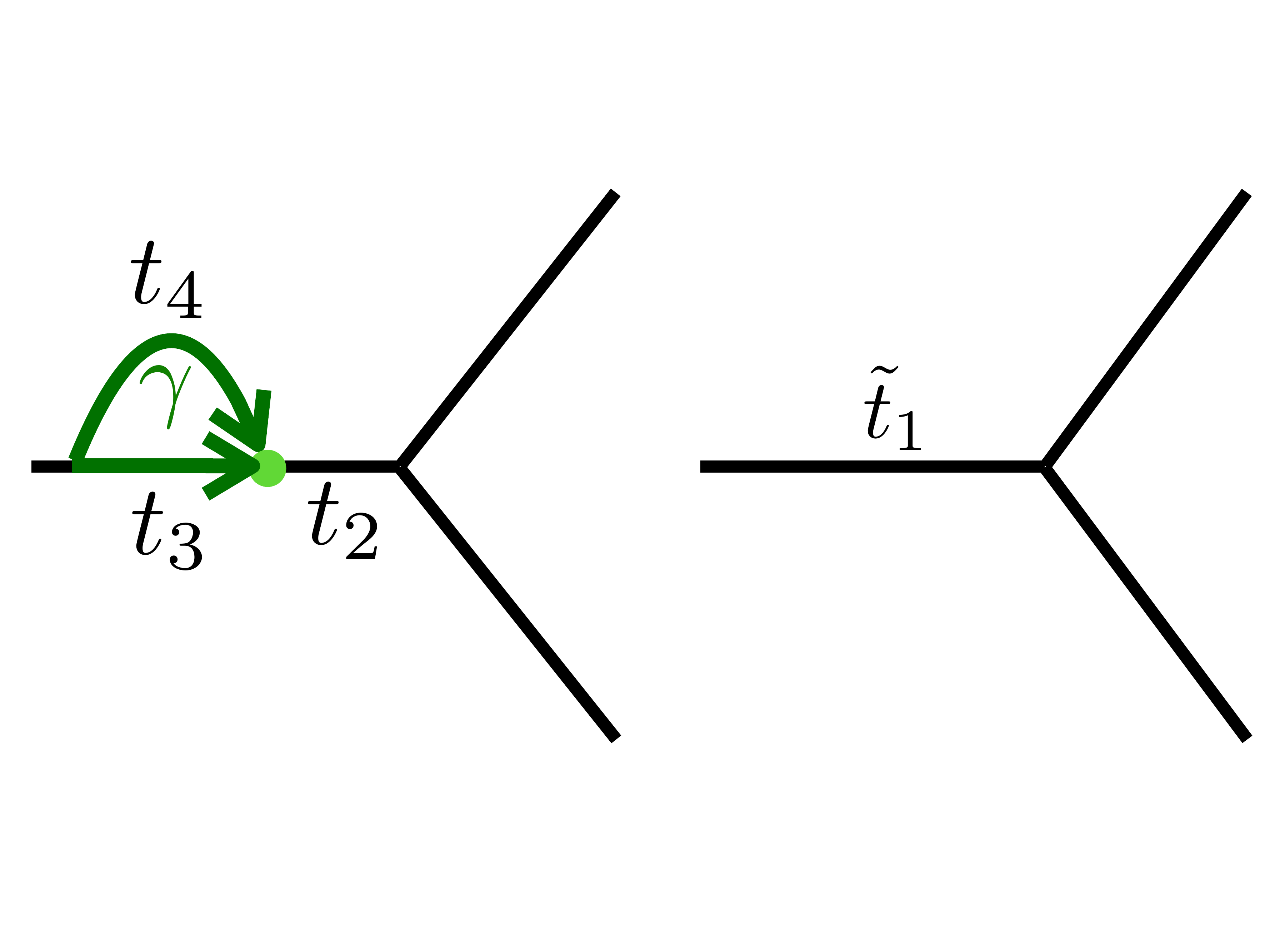}
\\Type 4
\end{minipage}
\begin{minipage}{0.55\textwidth}
\centering{\footnotesize
$\exp(-\tilde{t}_1):=\exp(-t_2)[1-\gamma^2 z_4 -(1-\gamma)^2 z_3]$}
\end{minipage}
  
\end{proof}

\begin{myLemma}
  Let $\mathcal{N}$ be a $n$-taxon explicit level-1 semi-directed
  phylogenetic network with $h$ hybridizations. Let
  $CF(\mathcal{N},\bs{t},\bs{\gamma})$ be the system of CF equations
  defined by the coalescent model (Definition \ref{cfeqs}).  Let $h_i$
  be the $i^{th}$ hybrid node of interest, which defines a
  $k_i$-cycle, and let $\theta_i = \{\bs{t}_{h_i},\bs{\gamma}_{h_i}\}$
  be the subset of numerical parameters associated with this
  hybridization cycle.  Define
  $CF_{h_i}(\mathcal{N})=CF(\mathcal{N},\bs{t}_{h_i},\bs{\gamma}_{h_i})$
  as the subset of CF equations that involve $\theta_i$.  Let
  $\mathcal{N}_{h_i}$ be a network with only $h=1$ hybridization, the
  hybridization of interest in $\mathcal{N}$, and define
  $CF_{h_i}(\mathcal{N}_{h_i})=CF(\mathcal{N}_{h_i},\bs{t}_{h_i},\bs{\gamma}_{h_i})$
  as the CF equations concerning the one hybridization in
  $\mathcal{N}_{h_i}$.

  Then, $CF_{h_i}(\mathcal{N}_{h_i})=CF_{h_i}(\mathcal{N})$. That is,
  the CF equations corresponding to the $i^{th}$ hybridization on the
  whole network $\mathcal{N}$ are the same as the CF equations from
  $\mathcal{N}_{h_i}$ (network with only one hybridization).
  \label{h1}
\end{myLemma}

\begin{proof}
  By Lemma \ref{nk}, $CF_{h_i}(\mathcal{N})$ and
  $CF_{h_i}(\mathcal{N}_{h_i})$ are completely defined by 4-taxon
  subsets with at most 2 taxa per subgraph defined by the $k_i$-cycle.
  Let $s=\{a,b,c,d\}$ be one of those 4-taxon subsets. Without loss of
  generality, assume that $a,b$ belong to subgraph $n_k$ and $c,d$
  belong to subgraph $n_{k'}$. Since $\mathcal{N}$ has other
  hybridization events, let's assume that $a,b$ are linked by another
  hybridization inside the $n_k$ subgraph. By the level-1 assumption,
  the hybridization event involving $a,b$ needs to be contained inside
  the $n_k$ subgraph, and thus, it can only be of type 1, 2, 4 or 5 in
  Figure \ref{5quartets} (because type 3 hybridization involves three
  taxa, not only two).  By Lemma \ref{type1}, this hybridization has
  the same CF equations as a properly chosen tree, and thus, we can
  remove all hybridization events in the subgraphs of $\mathcal{N}$
  and replace the subgraphs by subtrees. We obtain $\mathcal{N}_{h_i}$
  from $\mathcal{N}$ in this way, and thus,
  $CF_{h_i}(\mathcal{N}_{h_i})=CF_{h_i}(\mathcal{N})$.
\end{proof}

\begin{myLemma}
  Let $\mathcal{N}$ be a $n$-taxon explicit level-1 semi-directed
  phylogenetic network with $h=1$ hybridization that defines a
  hybridization cycle of $k \geq 5$ nodes. Let
  $CF(\mathcal{N},\bs{t},\bs{\gamma})$ be the system of CF equations
  defined by the coalescent model (Definition \ref{cfeqs}).  Let
  $\mathcal{T}$ be the major tree obtained by removing the minor
  hybrid edge in $\mathcal{N}$.  Proving
  $CF(\mathcal{N},\bs{t},\bs{\gamma})=CF(\mathcal{T},\bs{t}')$
  (hybridization detectability) is equivalent to proving hybridization
  detectability on a 5-cycle network ($\mathcal{\widetilde{N}}$)
  constructed from $\mathcal{N}$ by removing the $k-5$ subgraphs
  connected to the hybridization cycle farthest from the hybrid node
  (compare Figure \ref{k5nettree} to \ref{k8nettree}).
  \label{k5}
\end{myLemma}

\begin{proof}
  Without loss of generality, we will consider $\mathcal{N}$ to have an 8-cycle hybridization
  (Figure \ref{k8nettree} left) with major tree $\mathcal{T}$ (Figure
  \ref{k8nettree} right). When matching
  $CF(\mathcal{N},\bs{t},\bs{\gamma})=CF(\mathcal{T},\bs{t}')$, we can
  select 4-taxon subsets that would provide equalities for the
  parameters far from the hybrid node. For example, by selecting the
  4-taxon subset with $ n_2=1 $, $ n_4=2 $ and $ n_6=1 $, we get
  directly the equality: $z_4=w_4$. Similarly, we can choose 4-taxon
  subsets to produce the equalities: $ z_6=w_6 $, $ z_7=w_7 $,
  $ z_5=w_5 $ and $ z_3=w_3$.  In addition, if we select the 4-taxon
  subset with $n_4=n_6=2$, we get the equality
  $ z_6z_{4,6}z_4= w_6w_{4,6}w_4 $ which yields $ z_{4,6}=w_{4,6} $
  given the $z_4=w_4$ and $z_6=w_6$ equalities.  Following the same
  procedure, we obtain the equalities related to the other edges in
  the cycle that are far from the hybrid node: $ z_i=w_i $ for
  $ i=(6,7),(5,7),(3,5) $.

  We conclude that when considering the above equalities between
  $\mathcal{N}$ and $\mathcal{T}$, we only need equalities relating
  the parameters $ z_i $ for $ i=(1,3),(0,1)(0,2),(2,4),0,1,2 $, which
  corresponds to the 5-cycle hybridization of
  $\mathcal{\widetilde{N}}$.  Thus, the CF equations for a $k$-cycle
  for $k\geq 5$ are equivalent to a 5-cycle hybridization when
  studying the detectability of the hybridization.

  \begin{figure}[h!]
    \centering
    \begin{minipage}{.5\textwidth}
      \centering
      \includegraphics[scale=.15]{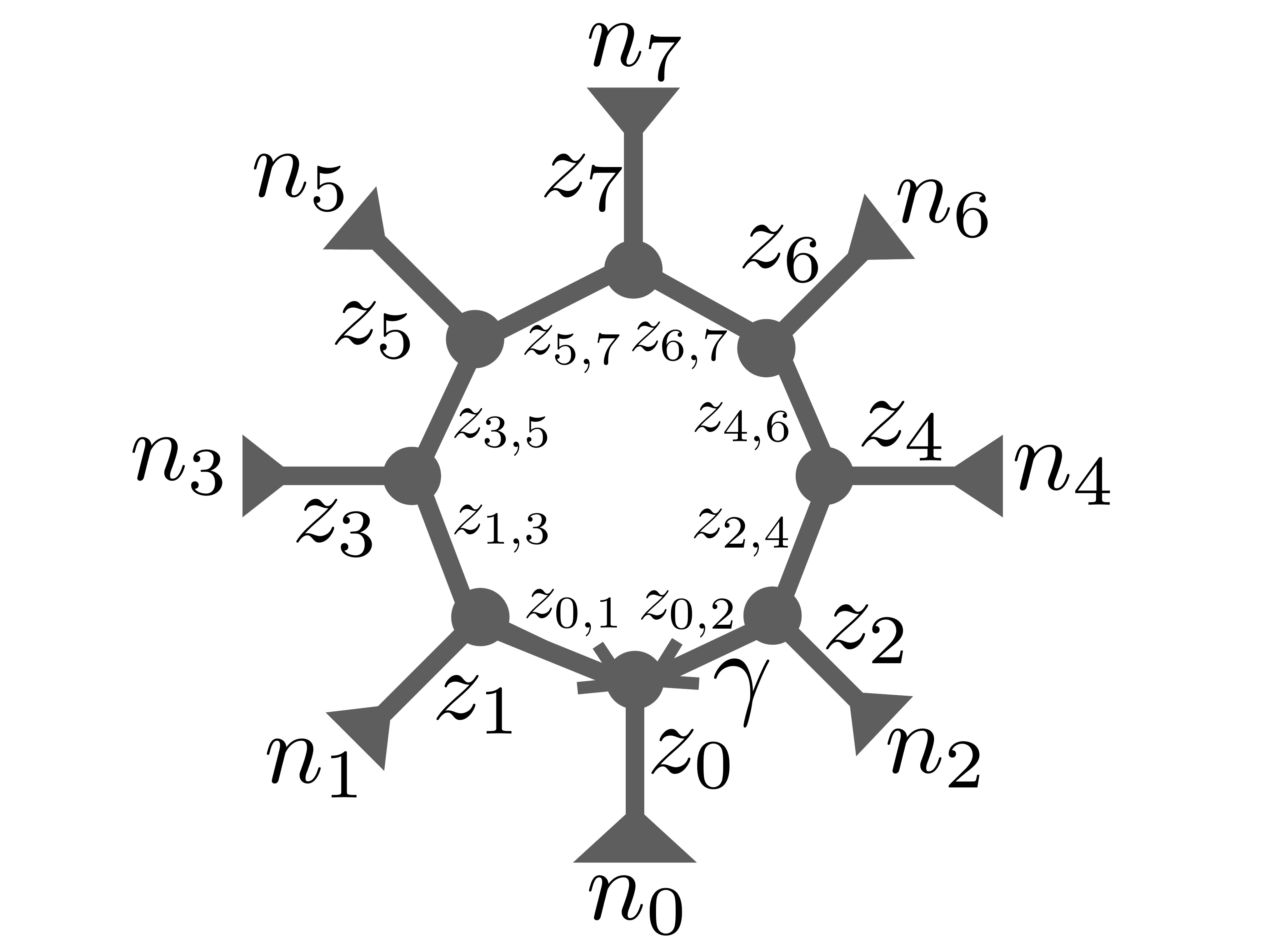}
    \end{minipage}%
    \begin{minipage}{.5\textwidth}
      \centering
      \includegraphics[scale=.15]{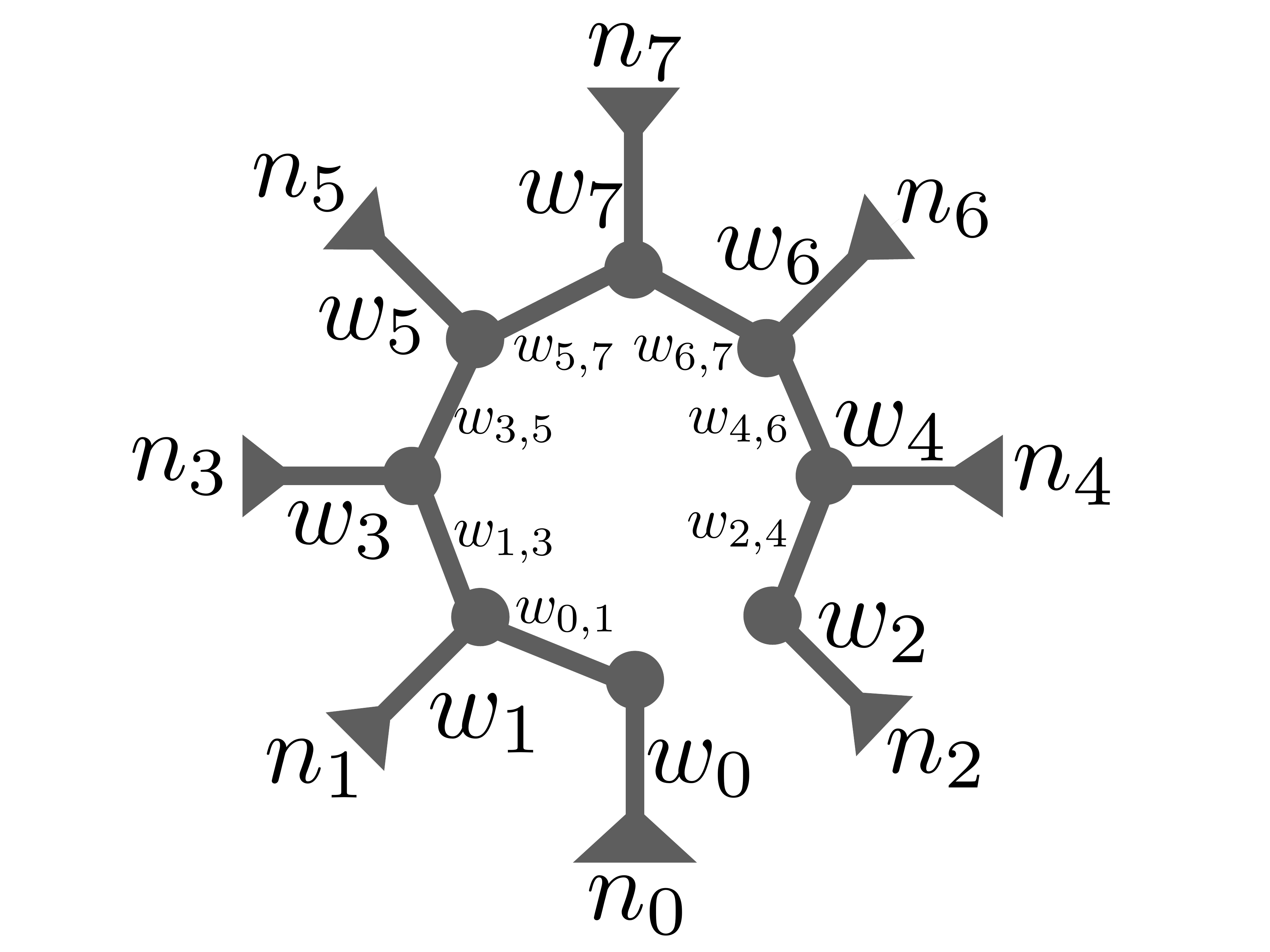}
    \end{minipage}
    \caption{8-cycle network and its corresponding tree
      representation.}
    \label{k8nettree}
  \end{figure}

\end{proof}

\begin{myLemma}
  Let $\mathcal{N}$ be a $n$-taxon explicit level-1 semi-directed
  phylogenetic network with $h=1$ hybridization that defines a
  hybridization cycle of $k \geq 4$ nodes. Let
  $CF(\mathcal{N},\bs{t},\gamma)$ be the system of CF equations
  defined by the coalescent model around the hybridization cycle
  (Definition \ref{cfeqs}).  Then, the polynomial equations in
  $CF(\mathcal{N},\bs{t},\gamma)$ can be decomposed into subsets of
  equations:
  $\{CF_i(\mathcal{N_i},\bs{t}_i,\gamma)\}$ each
  corresponding to a 4-cycle hybridization obtained by ignoring
  certain $n_k$ subgraphs at a time.
  \label{k5par}
\end{myLemma}

  \begin{proof}
    It is enough to prove finite identifiability of the parameters of
    a 4-cycle hybridization because any $k$-cycle hybridization
    ($k \geq 4$) can be reduced to a 4-cycle by sub-sampling of
    taxa. For example, in the 8-cycle hybridization in Figure
    \ref{k8nettree}, we can estimate any parameters with a properly
    chosen sub-sample of taxa.

    By Lemma \ref{nk}, we only need to sample 2 individuals per
    subgraph. Let us choose the 4-cycle hybridization with
    $n_0=n_1=n_2=n_3=2, n_4=n_5=n_6=n_7=0$, then the CF equations
    ($CF_{0,1,2,3}(\mathcal{N})$ for short) of this 4-cycle
    hybridization will only depend on the parameters
    $\{z_0,z_1,z_2,z_3\}$ along with the branch lengths inside the
    cycle $\{z_{1,3},z_{0,1},z_{0,2},\tilde{z}\}$ with
    $\tilde{z}=z_{2,4}+z_{4,6}+z_{6,7}+z_{5,7}+z_{3,5}$. Thus, we can
    study the finite identifiability of
    $\{z_0,z_1,z_2,z_3,z_{1,3},z_{0,1},z_{0,2},\tilde{z}\}$ using only
    the equations in $CF_{0,1,2,3}(\mathcal{N})$.  We can repeat this
    process with a different 4-taxon subset covering a different set
    of parameters until every parameter in $(\bs{t},\bs{\gamma})$ is
    included in at least one subset of CF equations.

\end{proof}

\subsection*{Reproducibility}
All Macaulay2 and Mathematica scripts are available in the GitHub
repository
\url{https://github.com/solislemuslab/snaq-identifiability} (link will
be made publicly available after publication).

\bibliography{identifiability.bib}

\begin{thebibliography}{10}

\bibitem{Allman2019}
Elizabeth~S. Allman, Hector Ba{\~n}os, and John~A. Rhodes.
\newblock Nanuq: a method for inferring species networks from gene trees under
  the coalescent model.
\newblock {\em Algorithms for Molecular Biology}, 14(1):24, 2019.

\bibitem{Allman2011}
Elizabeth~S. Allman, James~H. Degnan, and John~A. Rhodes.
\newblock Identifying the rooted species tree from the distribution of unrooted
  gene trees under the coalescent.
\newblock {\em Journal of Mathematical Biology}, 62(6):833--862, 2011.

\bibitem{Ane2007}
C{\'{e}}cile An{\'{e}}, Bret Larget, David~A Baum, Stacey~D Smith, and Antonis
  Rokas.
\newblock {Bayesian estimation of concordance among gene trees.}
\newblock {\em Molecular biology and evolution}, 24(2):412--26, mar 2007.

\bibitem{Banos2019}
Hector Ba{\~n}os.
\newblock Identifying species network features from gene tree quartets under
  the coalescent model.
\newblock {\em Bulletin of mathematical biology}, 81(2):494--534, 02 2019.

\bibitem{Baum2007}
David~A Baum.
\newblock {Concordance trees, concordance factors, and the exploration of
  reticulate genealogy}.
\newblock {\em Taxon}, 56(May):417--426, 2007.

\bibitem{Cox2007}
David Cox, John Little, and Donal O'Shea.
\newblock {\em {Ideals, varieties, and algorithms}}.
\newblock Springer, third edition, 2007.

\bibitem{Degnan2018}
James~H Degnan.
\newblock {Modeling Hybridization Under the Network Multispecies Coalescent}.
\newblock {\em Systematic Biology}, 67(5):786--799, 05 2018.

\bibitem{Francis2018}
Andrew Francis and Vincent Moulton.
\newblock Identifiability of tree-child phylogenetic networks under a
  probabilistic recombination-mutation model of evolution.
\newblock {\em Journal of Theoretical Biology}, 446:160--167, 2018.

\bibitem{Grayson}
DR~Grayson and ME~Stillman.
\newblock Macaulay2, a software system for research in algebraic geometry.
\newblock Available at \url{http://www.math.uiuc.edu/Macaulay2/}.

\bibitem{Huber2018}
Katharina~T. Huber, Vincent Moulton, Charles Semple, and Taoyang Wu.
\newblock Quarnet inference rules for level-1 networks.
\newblock {\em Bulletin of Mathematical Biology}, 80(8):2137--2153, 2018.

\bibitem{Huson2010}
Daniel Huson, Regula Rupp, and Celine Scornavacca.
\newblock {\em {Phylogenetic Networks}}.
\newblock Cambridge University Press, New York, NY, first edition, 2010.

\bibitem{Mathematica}
Wolfram~Research{,} Inc.
\newblock Mathematica, {V}ersion 12.0.
\newblock Champaign, IL, 2019.

\bibitem{Kumar2017}
Vikas Kumar, Fritjof Lammers, Tobias Bidon, Markus Pfenninger, Lydia Kolter,
  Maria~A. Nilsson, and Axel Janke.
\newblock The evolutionary history of bears is characterized by gene flow
  across species.
\newblock {\em Scientific Reports}, 7(1):46487, 2017.

\bibitem{Meng2009}
Chen Meng and Laura~Salter Kubatko.
\newblock {Detecting hybrid speciation in the presence of incomplete lineage
  sorting using gene tree incongruence: a model.}
\newblock {\em Theoretical population biology}, 75(1):35--45, mar 2009.

\bibitem{Pardi2015}
Fabio Pardi and Celine Scornavacca.
\newblock {Reconstructible Phylogenetic Networks: Do Not Distinguish the
  Indistinguishable}.
\newblock {\em PLOS Computational Biology}, 11(4):e1004135, 2015.

\bibitem{Pickrell2012}
JK~Joseph~K. Pickrell and JK~Jonathan~K. Pritchard.
\newblock {Inference of population splits and mixtures from genome-wide allele
  frequency data.}
\newblock {\em PLoS genetics}, 8(11):e1002967, jan 2012.

\bibitem{Pimentel2016}
D.~L. {Pimentel-Alarc\'on}, N.~{Boston}, and R.~D. {Nowak}.
\newblock A characterization of deterministic sampling patterns for low-rank
  matrix completion.
\newblock {\em IEEE Journal of Selected Topics in Signal Processing},
  10(4):623--636, 2016.

\bibitem{SolisLemus2017}
C.~Sol\'{i}s-Lemus, P.~Bastide, and C.~An\'{e}.
\newblock {PhyloNetworks: a package for phylogenetic networks}.
\newblock {\em Molecular Biology and Evolution}, 34(12):3292--3298, 2017.

\bibitem{Solis-Lemus2016}
Claudia Sol{\'{i}}s-Lemus and C{\'{e}}cile An{\'{e}}.
\newblock {Inferring phylogenetic networks with maximum pseudolikelihood under
  incomplete lineage sorting}.
\newblock {\em PLOS Genetics}, 12(3):e1005896, 2016.

\bibitem{Yu2012}
Yun Yu, James~H. Degnan, and Luay Nakhleh.
\newblock {The probability of a gene tree topology within a phylogenetic
  network with applications to hybridization detection.}
\newblock {\em PLoS genetics}, 8(4):e1002660, jan 2012.

\bibitem{Yu2014}
Yun Yu, Jianrong Dong, Kevin~J Liu, and Luay Nakhleh.
\newblock {Maximum Likelihood Inference of Reticulate Evolutionary Histories}.
\newblock {\em PNAS}, 111(46):16448--16453, 2014.

\bibitem{Zhu2018}
Sha Zhu and James~H. Degnan.
\newblock {Displayed Trees Do Not Determine Distinguishability Under the
  Network Multispecies Coalescent}.
\newblock {\em Systematic Biology}, 66(2):283--298, 12 2016.

\end{thebibliography}

\newpage

\appendix

\begin{center}
\large{\textbf{Supplementary Material}}\\
On the Identifiability of Phylogenetic Networks under a
  Pseudolikelihood model
\end{center}

\section{Definitions and notations}
\label{parInt}
Our main parameter of interest is the topology $\mathcal{N}$ of a
phylogenetic network along with the numerical parameters of the vector
of branch lengths ($\bs{t}$) and a vector of inheritance probabilities
($\bs{\gamma}$), describing the proportion of genes inherited by a
hybrid node from one of its hybrid parent (see Figure \ref{netEx}).

% There are two main classes of phylogenetic networks: implicit networks
% and explicit networks. Implicit networks \cite{Huson2010} -- also
% called split networks -- describe the discrepancy in gene trees, but
% they lack biological interpretation as the internal nodes do not
% represent ancestral species. In this work, we will only focus on
% explicit networks, but see \cite{Huson2010} for more information on
% implicit networks.

A rooted explicit phylogenetic network $\mathcal{N}$ on taxon set $X$ is a
connected directed acyclic graph with vertices
$V = \{r\} \cup V_L \cup V_H \cup V_T$ , edges $E = E_H \cup E_T$ and
a bijective leaf-labeling function $f : V_L \rightarrow X$ with the
following characteristics:
\begin{itemize}
\item{The root $r$ has indegree 0 and outdegree 2}
\item{Any leaf $v \in V_L$ has indegree 1 and outdegree 0}
\item{Any tree node $v \in V_T$ has indegree 1 and outdegree 2}
\item{Any hybrid node $v \in V_H$ has indegree 2 and outdegree 1}
\item{A tree edge $e \in E_T$ is an edge whose child is a tree node}
\item{A hybrid edge $e \in E_H$ is an edge whose child is a hybrid
    node}
\item{A hybrid edge $e \in E_H$ has an inheritance probability
    parameter $\gamma_e < 1$ which represents the proportion of the
    genetic material that the child hybrid node received from this
    parent edge. For a tree edge $e$, $\gamma_e=1$}
\end{itemize}

In a rooted explicit network, every internal node represents a
biological mechanism: speciation for tree nodes and hybridization for
hybrid nodes. However, other types of phylogenetic network also exist
in literature, such as unrooted networks \cite{Huson2010} and
semi-directed networks \cite{Solis-Lemus2016}.  Unrooted phylogenetic
networks are typically obtained by suppressing the root node and the
direction of all edges.  In semi-directed networks, on the other hand,
the root node is suppressed and we ignore the direction of all tree
edges, but we maintain the direction of hybrid edges, thus keeping
information on which nodes are hybrids. The placement of the root is
then constrained, because the direction of the two hybrid edges to a
given hybrid node inform the direction of time at this node: the third
edge must be a tree edge directed away from the hybrid node and
leading to all the hybrid’s descendants. Therefore the root cannot be
placed on any descendant of any hybrid node, although it might be
placed on some hybrid edges.  See Figure \ref{netEx} for the example
of a rooted explicit phylogenetic network (center), and its semi-directed
version (right).

Assume that $\mathcal{N}$ has $n$ species and $h$ hybridization events
(that is, $|V_h|=h$). Each hybridization event creates a blob (or
cycle) which represents
% the union of two different paths from two distinct vertices $s,t$
% (see figure \missing{xx}). A blob can also be described as
a subgraph with at least two nodes and no cut-edges. A cut-edge is any
edge in the network whose removal disconnects the network.
Traditionally, the term \textit{cycle} is used for directed (rooted)
networks, and the term \textit{blob} is used for undirected (unrooted)
networks. We make no distinction here as we will focus on
\textit{semi-directed} networks, and thus, some edges in the cycle
(blob) will be directed, and some will not. Let $k_i$ be the number of
nodes in the $i^{th}$ hybridization cycle (for $i=1,2,\cdots,h$), so
that this hybridization is denoted $k_i$-cycle hybridization. For
example, in Figure \ref{netEx} (right) $n = 6, h = 1, k_1 = 3$.

We further assume that the phylogenetic network is of level-1
\cite{Huson2010}, i.e. any given edge can be part of at most one
cycle. This means that there is no overlap between any two cycles (see
Figure \ref{level1}).

\begin{figure}[ht!]
\centering
\includegraphics[scale=0.1]{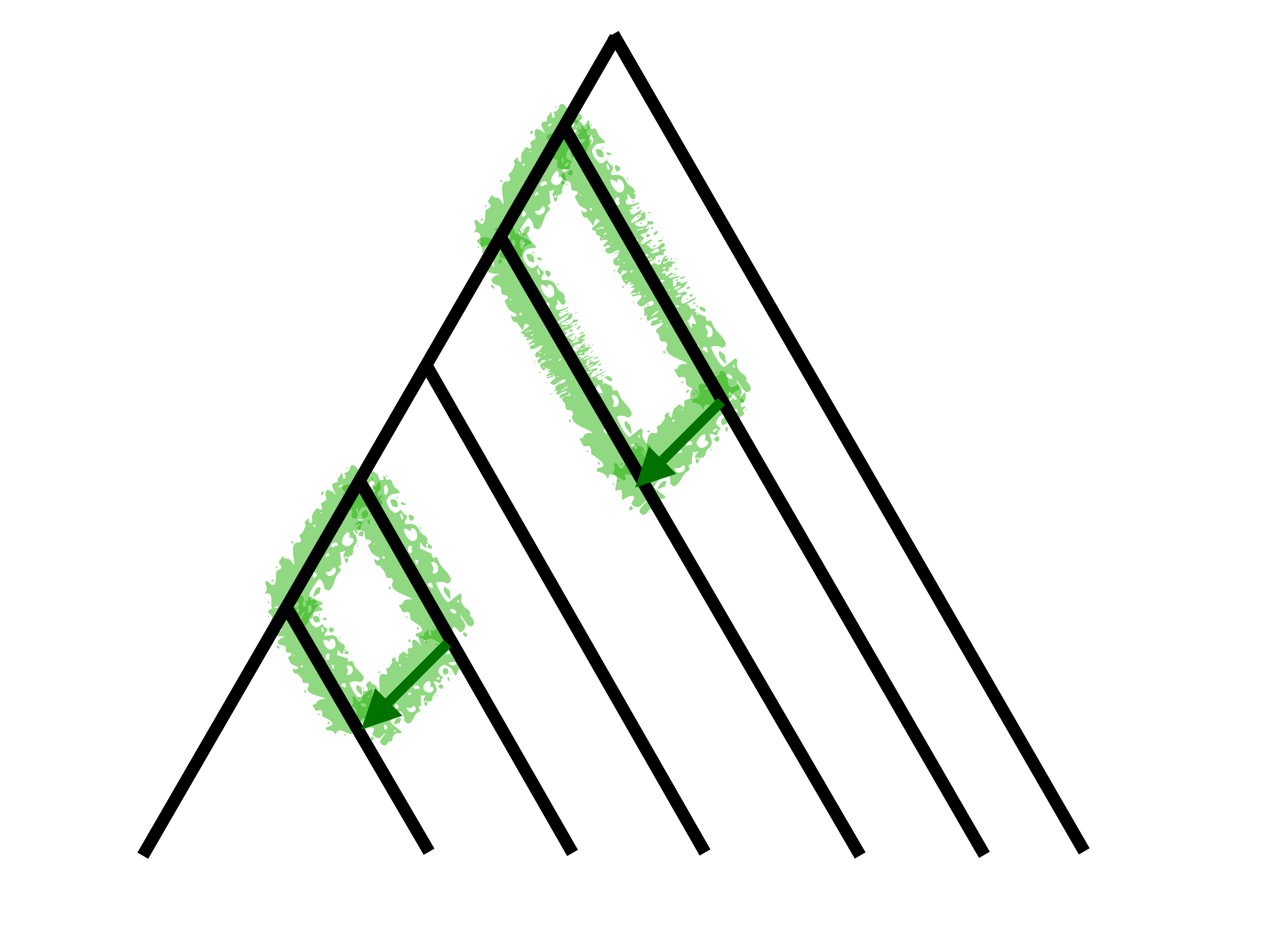}
\includegraphics[scale=0.1]{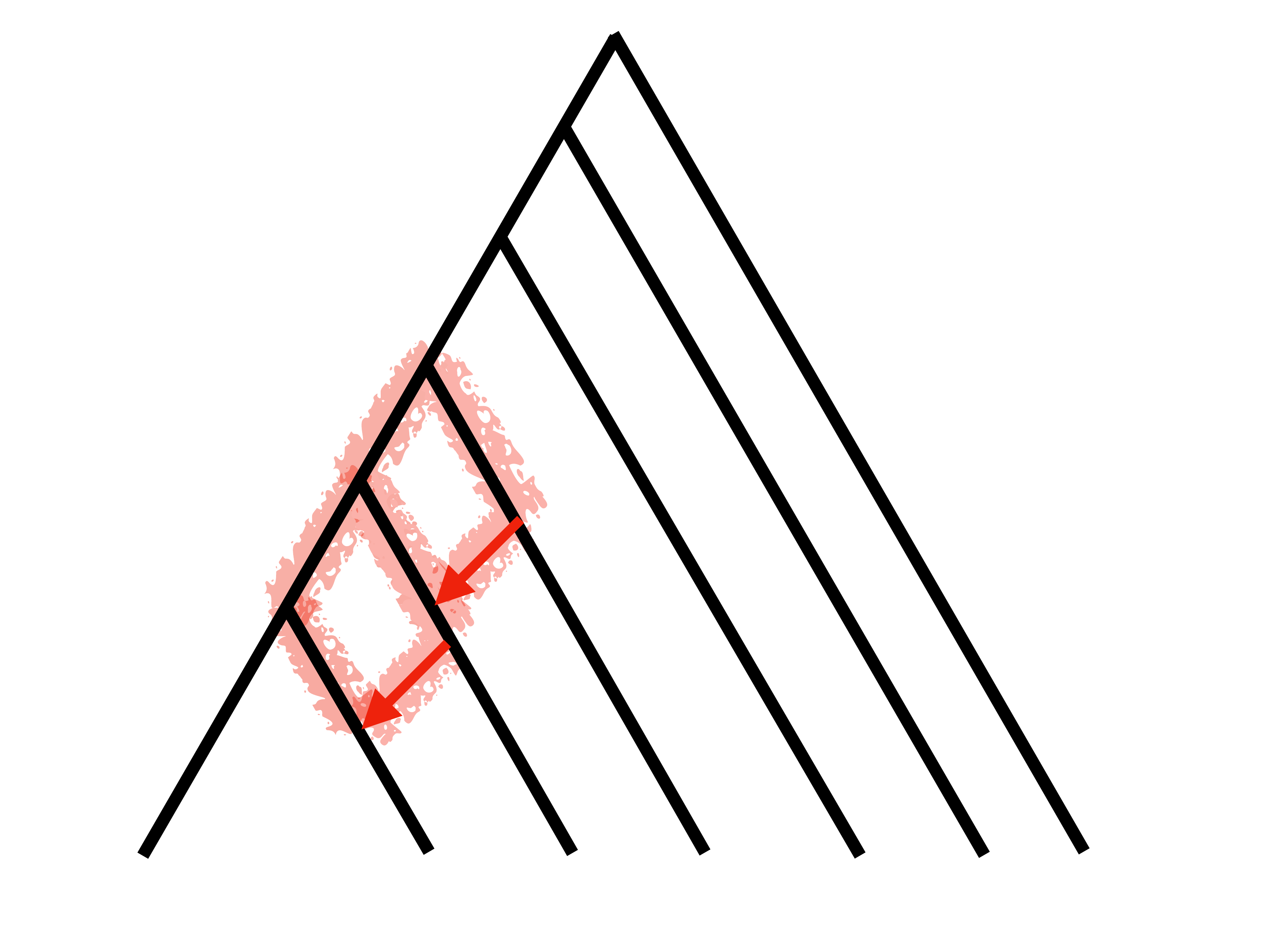}
\caption{Left: Level-1 network vs Right: non-level-1 network}
\label{level1}
\end{figure}

Thus, our parameters of interest are
$(\mathcal{N}, \bs{t}, \bs{\gamma})$ where
\begin{itemize}
\item $\mathcal{N}$ is an explicit semi-directed level-1 phylogenetic
  network that links the $n$ species under study, and has $h$
  hybridization events. This network has two vectors of numerical
  parameters:
\item branch lengths $\bs{t} \in [0,\infty)^{n_e}$ for $n_e$ branches
  in the network, and
\item inheritance probabilities $\bs{\gamma} \in [0,1]^{n_h}$ for
  $n_h$ minor hybrid edges.
\end{itemize}
 
Note that for every hybridization event, there are two parent hybrid
edges connected to the hybrid node: 1) major hybrid edge with
inheritance probability $\gamma > 0.5$, and 2) a minor hybrid edge
with inheritance probability $\gamma < 0.5$. Both edges are
parametrized with the same $\gamma$.

\section{Concordance factor (CF) data} 
\label{data}
The data for our pseudolikelihood estimation method is a collection of
estimated gene trees $\{ G_i\}_{i=1}^{g}$ from $g$ loci (ortholog
region in genome with no recombination). These gene trees are unrooted
and only topologies are considered (no branch lengths).  To account
for estimation error in the gene trees, we do not consider the gene
tree directly as input data, but we summarize them into the
\textit{concordance factors} (CF) \cite{Baum2007} for every possible
quartet. A quartet is a 4-taxon unrooted tree. For example, for taxon
set $s = \{a, b, c, d\}$, there are only three possible quartets,
represented by the splits $q_1 = ab|cd$, $q_2 = ac|bd$ and
$q_3 = ad|bc$.  The CF of a given quartet is the proportion of genes
whose true tree displays that quartet (see Figure
\ref{quartetCFfig}). When these CFs are estimated in BUCKy
\cite{Ane2007}, they represent true genomic discordance, and naturally
account for the estimation error in the gene trees.  Our
pseudolikelihood method (SNaQ \cite{Solis-Lemus2016, SolisLemus2017})
uses these observed CFs as input data to estimate a species
network. SNaQ finds the best network that fits the data by identifying
the network whose theoretical CFs are close to the observed CFs. More
details about the model are presented in the Section \ref{modelSec}.

\begin{figure}[ht!]
    \centering
    \includegraphics[scale=0.25]{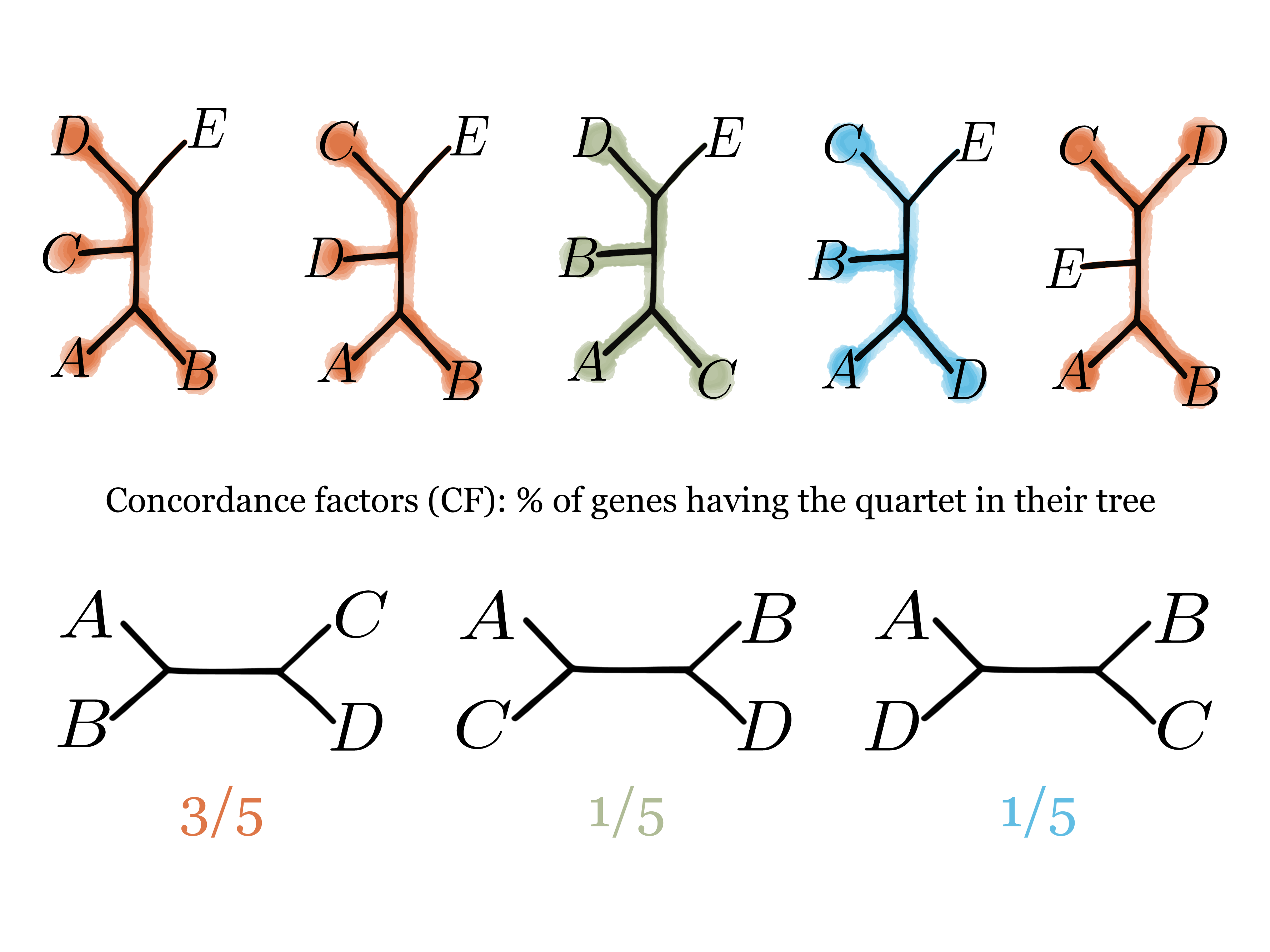}
    \caption{Example of computation of concordance factors from a
      sample of 5 gene trees. For the 4-taxon subset $\{A,B,C,D\}$,
      there are three possible quartets: $AB|CD$, $AC|BD$, $AD|BC$
      with observed CFs of $(3/5, 1/5, 1/5)$. This estimation of CFs
      is not robust to estimation error as we consider the gene trees
      as perfectly reconstructed. Alternatively, we can estimate the
      CFs with BUCKy\cite{Ane2007} which accounts for gene tree
      estimation error.}
    \label{quartetCFfig}
\end{figure}

% good for snaq paper, but not this paper
% \begin{figure}
% \centering
% \includegraphics[scale=0.25]{figures/snaq-flowchart.pdf}
% \caption{SNaQ flowchart from sequence alignments to estimated
%   phylogenetic network using MrBayes\cite{Huelsenbeck2001} to estimate
%   gene trees, BUCKy\cite{Ane2007} to estimate CFs, and
%   QuartetMaxCut\cite{Snir2012} to estimate a starting tree to aid in
%   optimization.}
% \label{flowchart}
% \end{figure}

\section{Probability model for gene trees}
\label{modelSec}

Here, we define the pseudolikelihood function for
$(\mathcal{N}, \bs{t}, \bs{\gamma})$ given the estimated CFs as
data. Most of this section is an extract from \cite{Solis-Lemus2016}
(except for Section \ref{equivalence}) and presented here for the sake
of completeness.

The pseudolikelihood of a network is based on the likelihood of its
4-taxon subnetworks (\textit{quarnets} in \cite{Huber2018}). That is,
for a given network $\mathcal{N}$ with $n \geq 4$ taxa, we consider
all 4-taxon subsets
$\mathcal{S} = \{s = \{a,b,c,d\} : a,b,c,d \in X\}$ and combine the
likelihood of each 4-taxon subnetworks to form the full network
pseudolikelihood:
\begin{equation}
    \tilde{L}(\mathcal{N}) = \prod_{s \in \mathcal{S}} L(s)
\label{pseudolik1}
\end{equation}
where $L(s)$ is the likelihood of the subnetwork of a given 4-taxon
subset $s$. These 4-taxon likelihoods are not independent, which is
why we get a pseudolikelihood when we multiply them, instead of a true
likelihood.

To calculate the likelihood of a given 4-taxon subset $s=\{a,b,c,d\}$,
let $Y=(Y_{q_1}, Y_{q_2}, Y_{q_3})$ denote the number of gene trees
that match each of the three possible quartet resolutions:
$q_1=ab|cd, q_2=ac|bd, q_3=ad|bc$, then $Y$ follows a multinomial
distribution with probabilities $(CF_{q_1}, CF_{q_2}, CF_{q_3})$, the
theoretical CFs expected under the coalescent on the 4-taxon
subnetwork which were derived in \cite{Solis-Lemus2016} under the
multispecies coalescent model on networks \cite{Meng2009, Yu2012,
  Yu2014}.

Thus, the likelihood of a given 4-taxon subset is
\ref{pseudolik1}, we get:
\begin{equation}
    L(s) \propto (CF_{q_1})^{Y_{q_1}} (CF_{q_2})^{Y_{q_2}} (CF_{q_3})^{Y_{q_3}}
\end{equation}

The data are summarized in the $Y$ values through the estimated CFs,
and the candidate network governs the CF values, which we explain
below.

\subsection{The multispecies coalescent network (MSCN) model}

The theoretical CFs expected under the coalescent model are already
derived for a species tree in \cite{Allman2011}, and for a species
network in \cite{Solis-Lemus2016}. In both cases, the CFs do not
depend on the position of the root. For the tree, the major CF is
defined for the quartet that agrees with the species tree. That is, if
the species tree has the split $ab|cd$ with internal edge $t$, then
the major CF would be $CF_{ab|cd}=1-2/3 \exp(-t)$. The CF for the
minor resolutions (in disagreement with the species tree $ab|cd$)
would then be $CF_{ac|bd}=CF_{ad|bc}=1/3 \exp(-t)$.

For the case of a 4-taxon network, the theoretical CFs are weighted
averages of CFs on trees.  For example, Figure \ref{cfs} center shows
a semi-directed 4-taxon and one possible rooting in Figure \ref{cfs}
left. The CFs on this network are given by the weighted averages of CF
on the two trees on the right, with weights given by
$1-\gamma,\gamma$:
\begin{itemize}
\item
  $CF_{ab|cd} = (1-\gamma)(1-2/3\exp(-t_1))+\gamma(1/3
  \exp(-t_1-t_2))$ for the major resolution, and
\item $CF_{ac|bd} = CF_{ad|bc}=(1-\gamma)(1/3\exp(-t_1))+\gamma(1/3
\exp(-t_1-t_2))$ for the minor resolutions.
\end{itemize}

\begin{figure}
    \centering
    \includegraphics[scale=0.1]{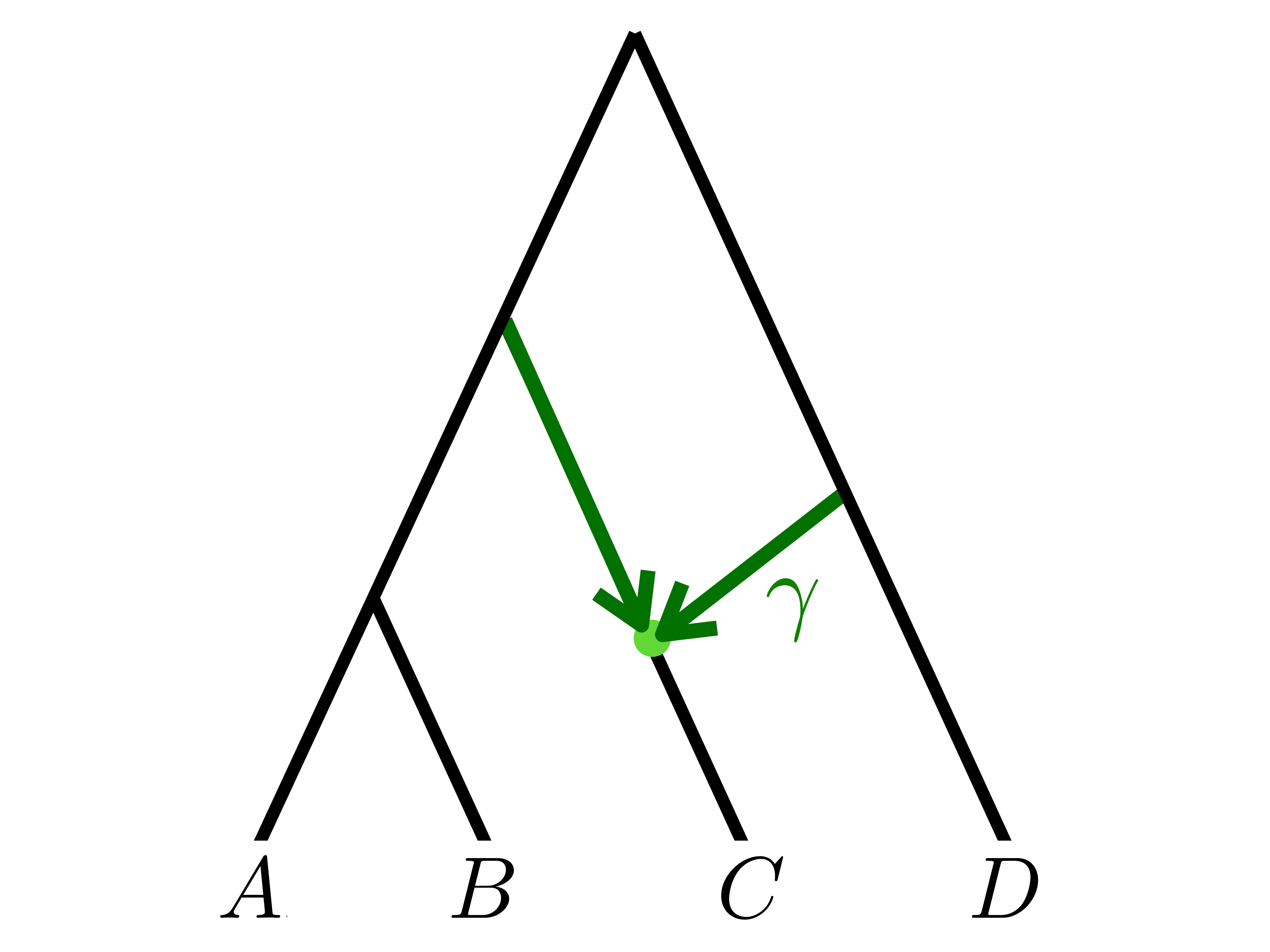}
    \includegraphics[scale=0.1]{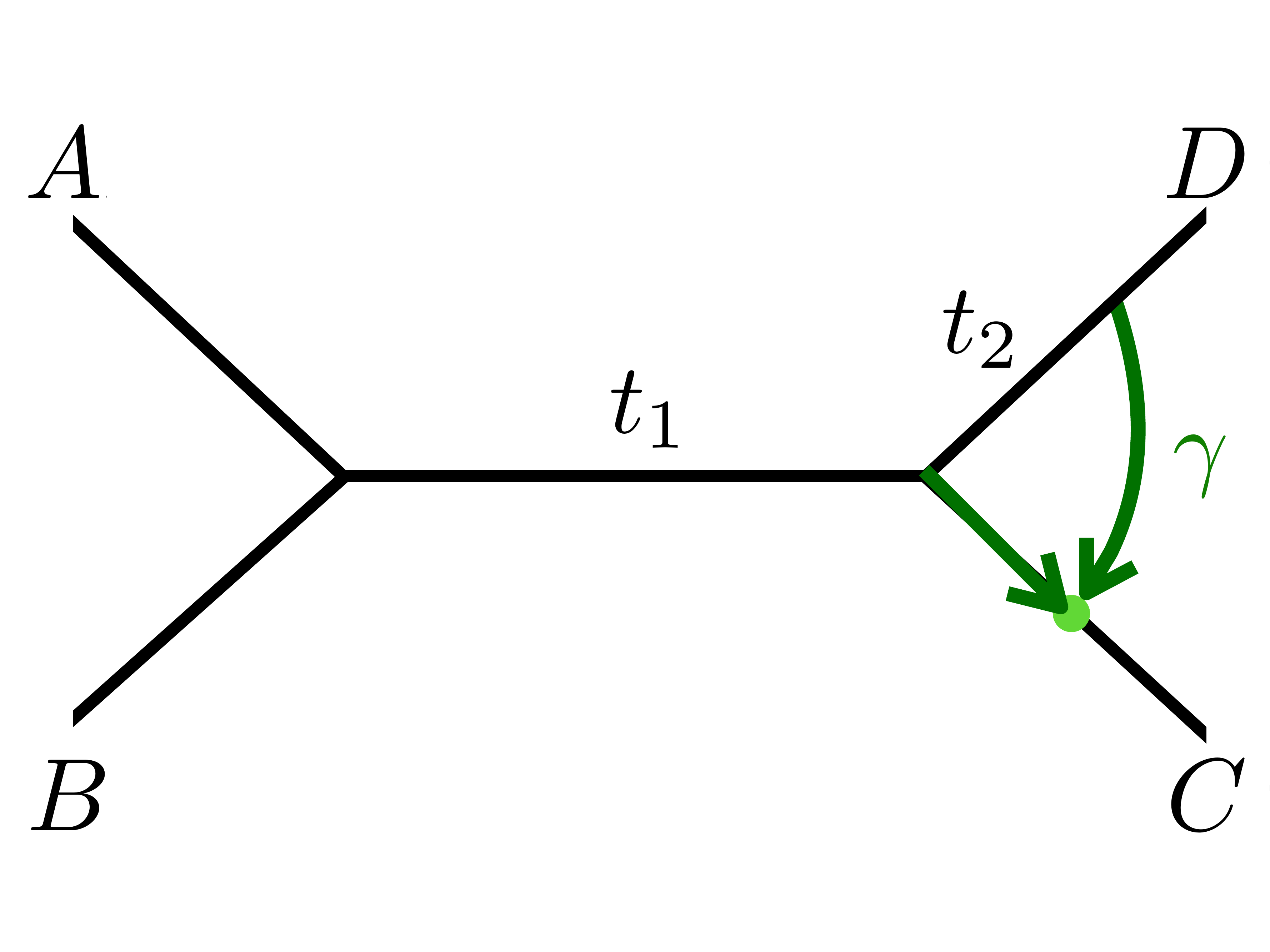}
    \includegraphics[scale=0.1]{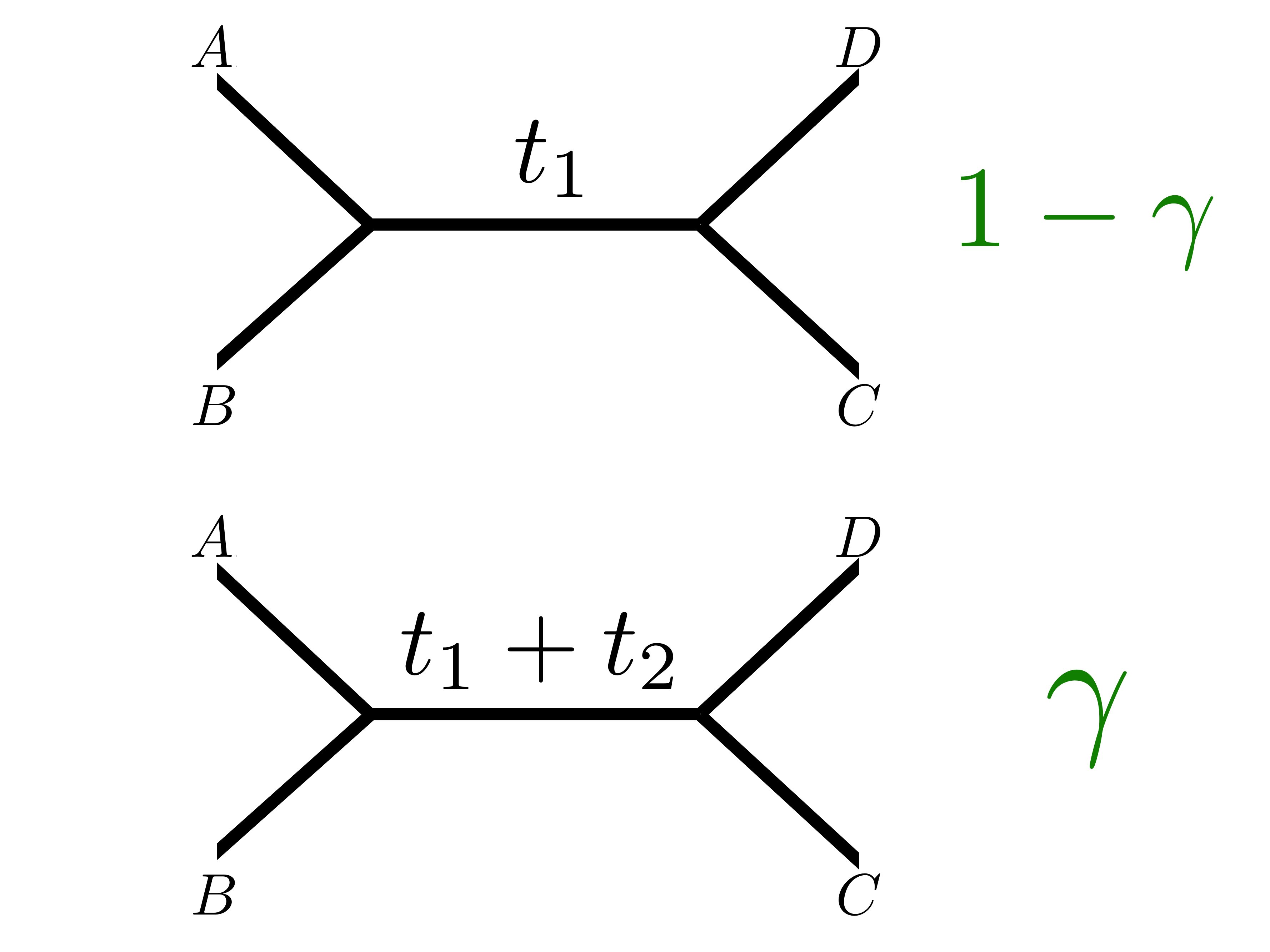}
    \caption{Rooted 4-taxon network (left) and its semi-directed
      version (center). Quartet CFs expected under the network do not
      depend on the root placement, and are weighted averages of
      quartet CFs expected under the unrooted trees (right).}
    \label{cfs}
\end{figure}

\subsection{CF formulas for all possible level-1 4-taxon networks}

Derived in \cite{Solis-Lemus2016}, the theoretical CFs for the five
types of level-1 4-taxon network are used in the computation of the
pseudolikelihood of any level-1 network. These formulas appear in
\cite{Solis-Lemus2016}, but we re-write them here in Figure
\ref{5quartets}.

\begin{figure}
%---------------------- Type 1
\noindent\begin{minipage}{0.35\textwidth}
\centering
\includegraphics[scale=0.1]{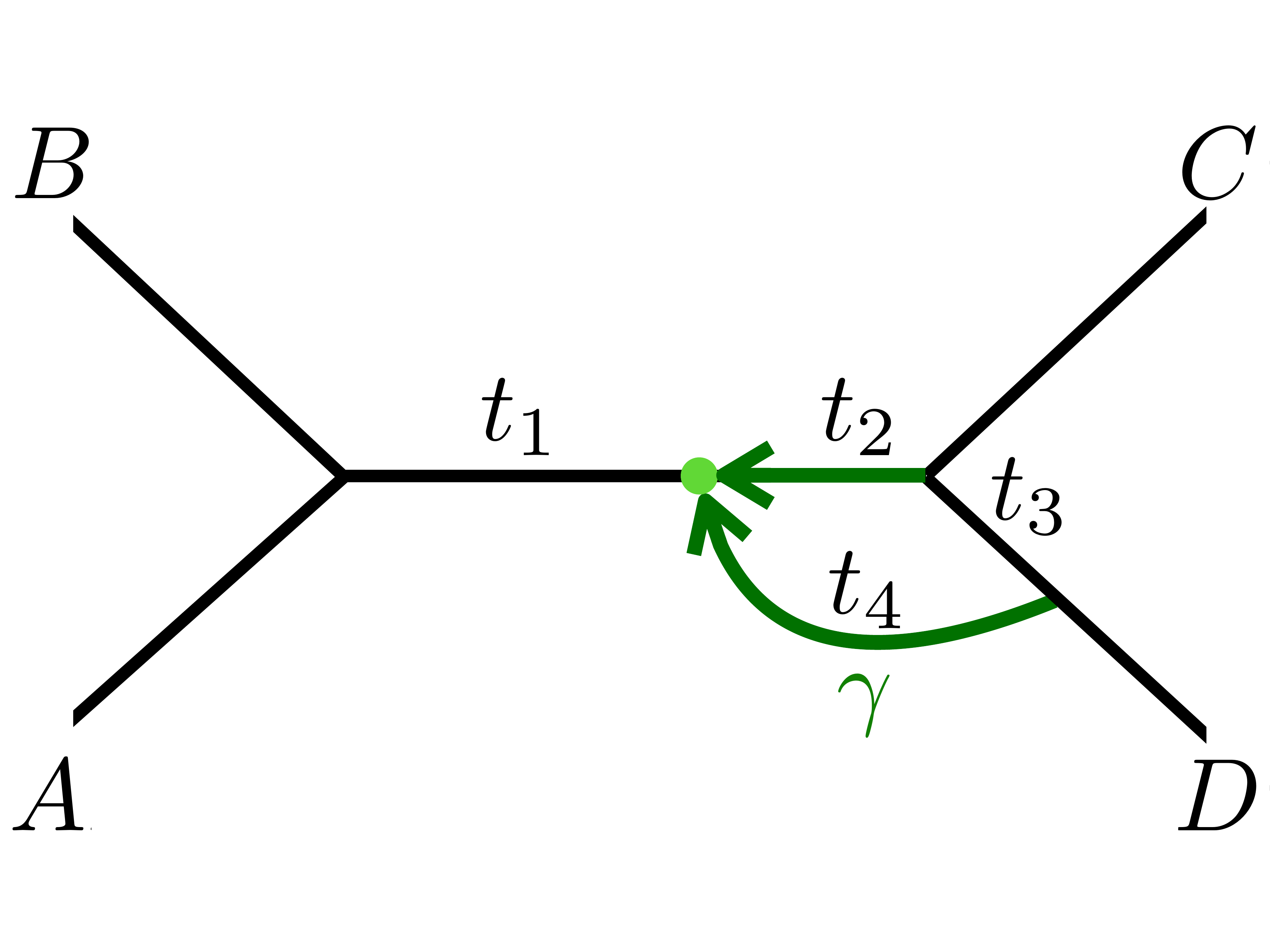}
\\Type 1%\captionof{figure}{Network 1}\label{fig:caseI2}
\end{minipage}
\begin{minipage}{0.55\textwidth}
\footnotesize
\begin{align*}
CF_{AB|CD}&=(1-\gamma)^2(1-2/3\exp(-t_1-t_2))\\
&+2\gamma(1-\gamma)(1-\exp(-t_1)+1/3\exp(-t_1-t_3))\\
&+\gamma^2(1-2/3\exp(-t_1-t_4))\\
CF_{AC|BD}&=(1-\gamma)^2(1/3\exp(-t_1-t_2))\\
&+\gamma(1-\gamma)\exp(-t_1)(1-1/3\exp(-t_3))\\
&+\gamma^2(1/3\exp(-t_1-t_4))\\
CF_{AD|BC}&=(1-\gamma)^2(1/3\exp(-t_1-t_2))\\
&+\gamma(1-\gamma)\exp(-t_1)(1-1/3\exp(-t_3))\\
&+\gamma^2(1/3\exp(-t_1-t_4))
\end{align*}
\end{minipage}

%---------------------- Type 2
\noindent\begin{minipage}{0.35\textwidth}
\centering
\includegraphics[scale=0.1]{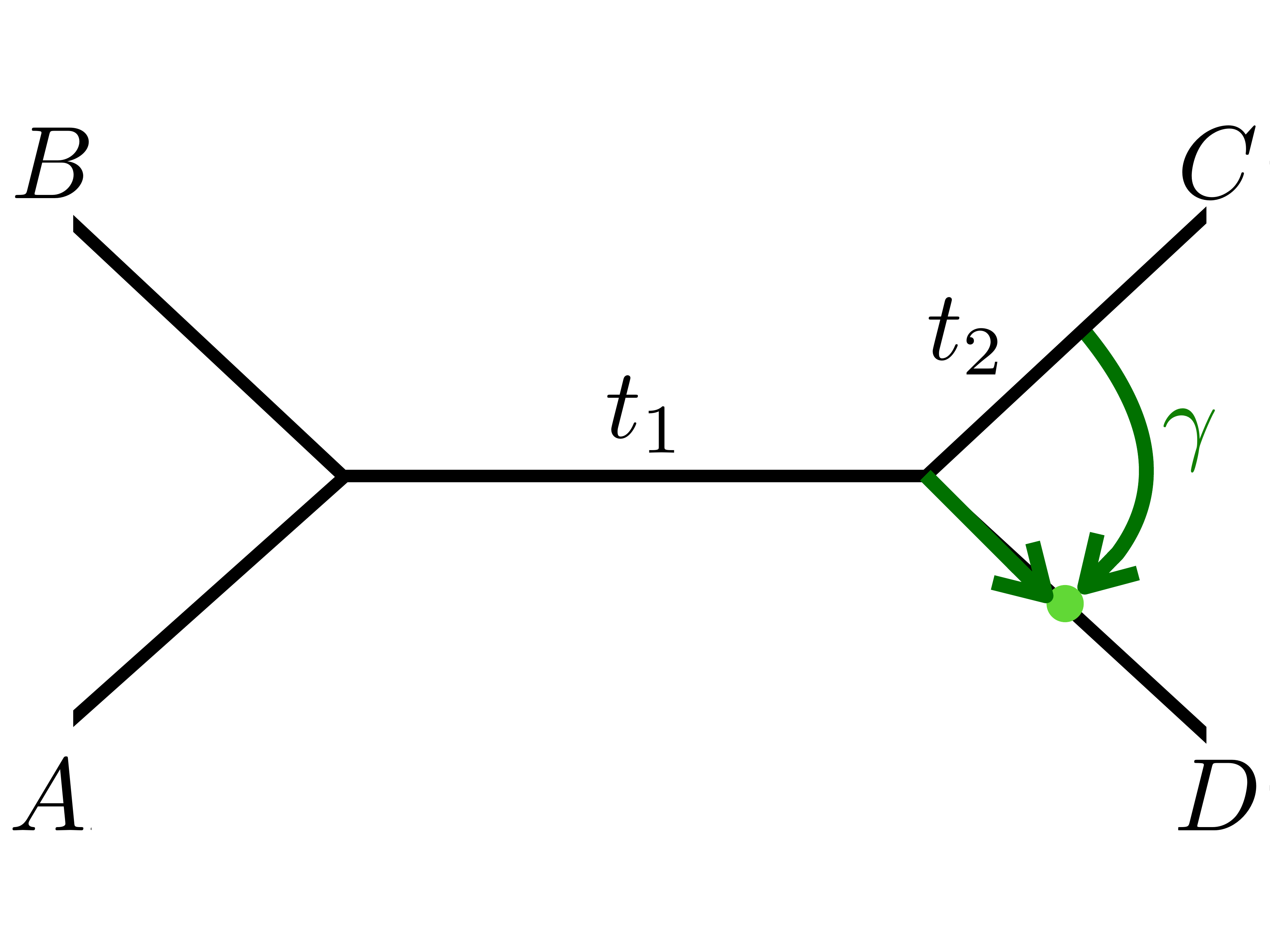}
\\Type 2%\captionof{figure}{Network 2}\label{fig:caseII}
\end{minipage}
\begin{minipage}{0.55\textwidth}
\footnotesize
\begin{align*}
CF_{AB|CD}&=(1-\gamma)(1-2/3\exp(-t_1))+\gamma(1-2/3\exp(-t_1-t_2))\\
CF_{AC|BD}&=(1-\gamma)1/3\exp(-t_1)+\gamma1/3\exp(-t_1-t_2)\\
CF_{AD|BC}&=(1-\gamma)1/3\exp(-t_1)+\gamma1/3\exp(-t_1-t_2)
\end{align*}
\end{minipage}

%---------------------- Type 3
\noindent\begin{minipage}{0.35\textwidth}
\centering
\includegraphics[scale=0.1]{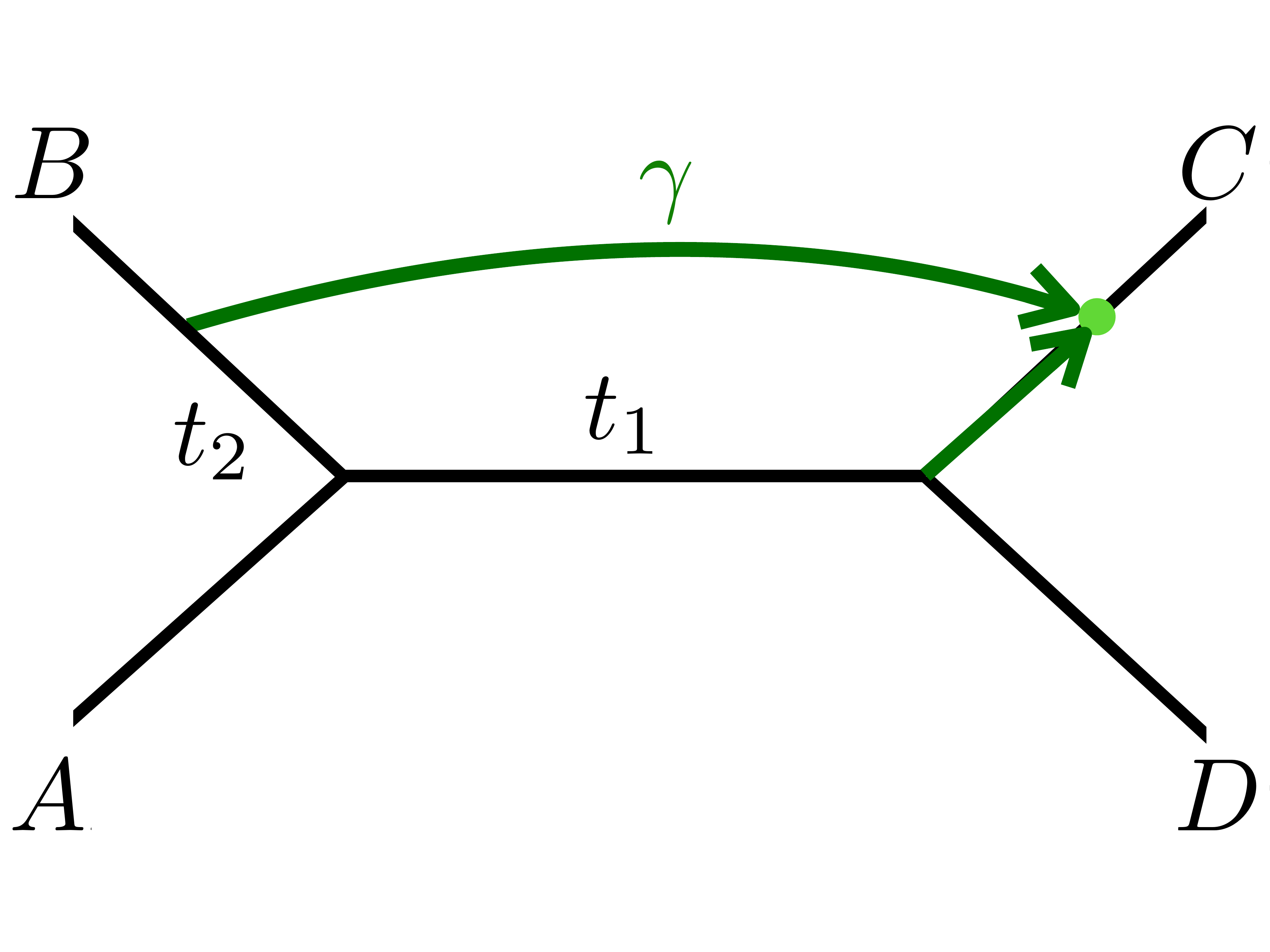}
\\Type 3%\captionof{figure}{Network 3}\label{fig:caseIII}
\end{minipage}
\begin{minipage}{0.55\textwidth}
\footnotesize
\begin{align*}
CF_{AB|CD}&=(1-\gamma)(1-2/3\exp(-t_1))+\gamma(1/3\exp(-t_2))\\
CF_{AC|BD}&=(1-\gamma)1/3\exp(-t_1)+\gamma(1-2/3\exp(-t_2))\\
CF_{AD|BC}&=(1-\gamma)1/3\exp(-t_1)+\gamma1/3\exp(-t_2)
\end{align*}
\end{minipage}

%---------------------- Type 4
\noindent\begin{minipage}{0.35\textwidth}
\centering
\includegraphics[scale=0.1]{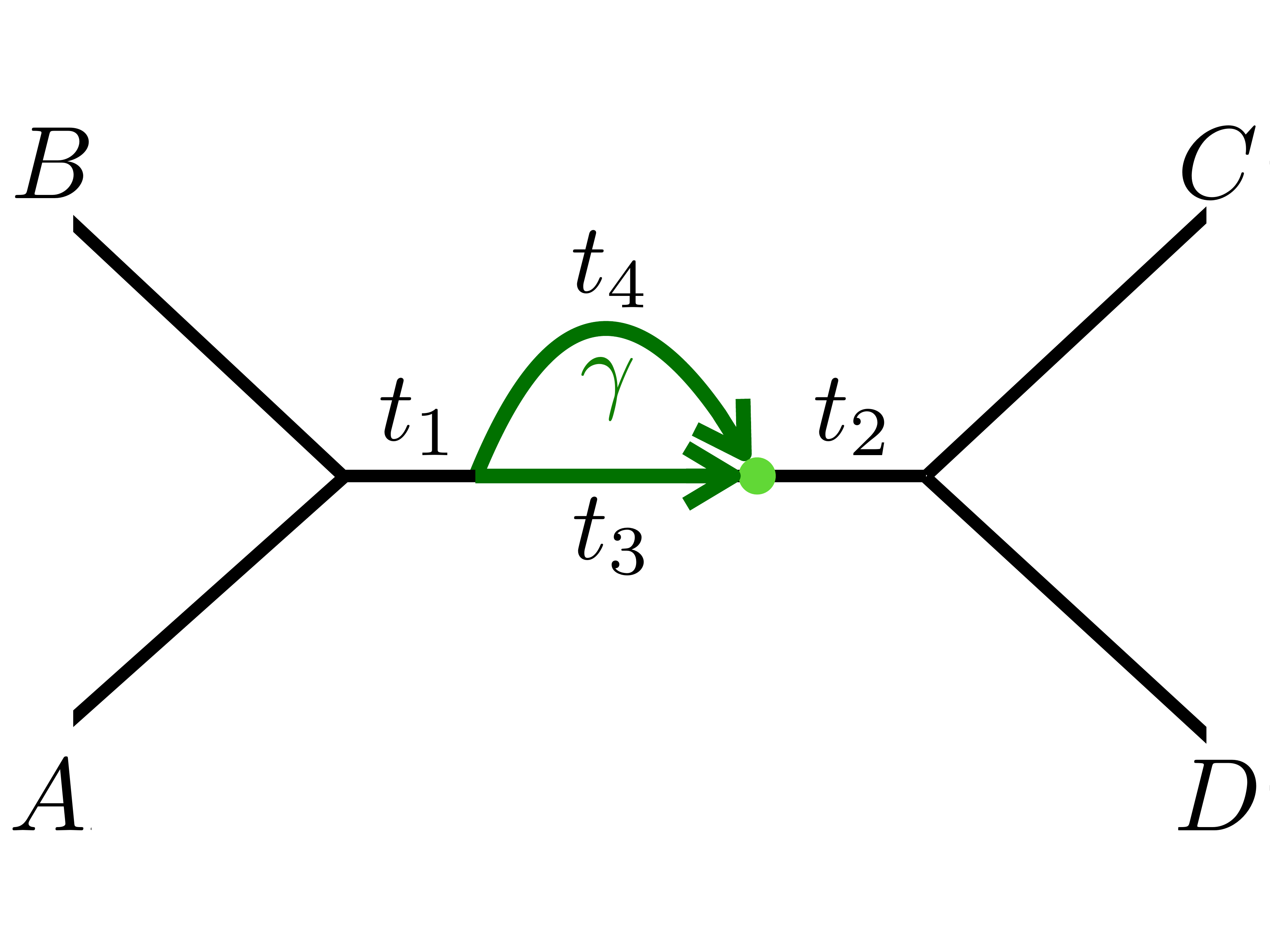}
\\Type 4%\captionof{figure}{Network 4}\label{fig:caseIV}
\end{minipage}
\begin{minipage}{0.55\textwidth}
\footnotesize
\begin{align*}
CF_{AB|CD}&=(1-\gamma)^2(1-2/3\exp(-t_1-t_2-t_3))\\
&+2\gamma(1-\gamma)(1-2/3\exp(-t_1-t_2))\\
&+\gamma^2(1-2/3\exp(-t_1-t_2-t_4))\\
CF_{AC|BD}&=(1-\gamma)^2(1/3\exp(-t_1-t_2-t_3))\\
&+2\gamma(1-\gamma)(1/3\exp(-t_1-t_2))\\
&+\gamma^2(1/3\exp(-t_1-t_2-t_4))\\
CF_{AD|BC}&=(1-\gamma)^2(1/3\exp(-t_1-t_2-t_3))\\
&+2\gamma(1-\gamma)(1/3\exp(-t_1-t_2))\\
&+\gamma^2(1/3\exp(-t_1-t_2-t_4))
\end{align*}
\end{minipage}

%---------------------- Type 5
\noindent\begin{minipage}{0.35\textwidth}
\centering
\includegraphics[scale=0.1]{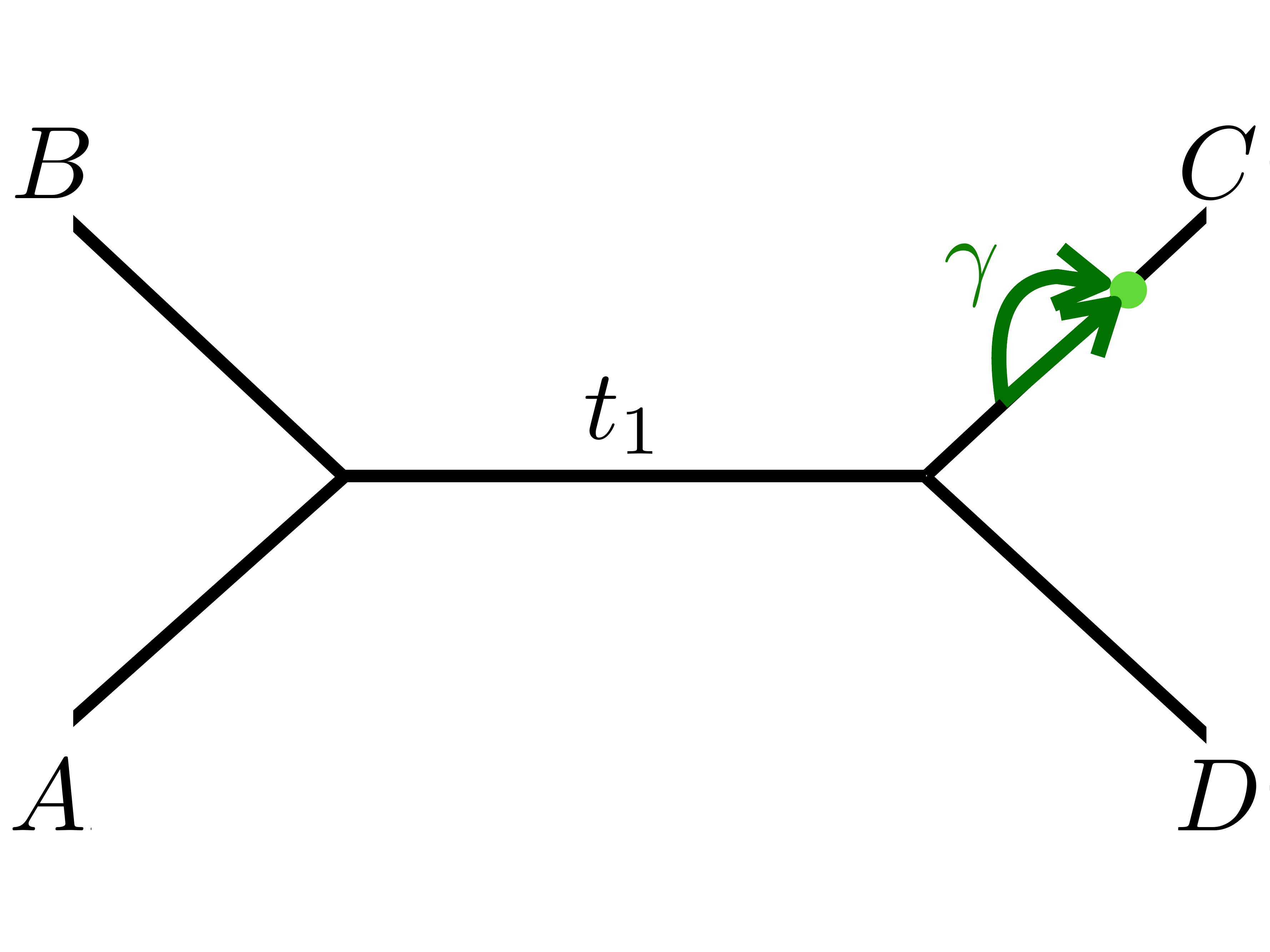}
\\Type 5%\captionof{figure}{Network 5}\label{fig:caseV}
\end{minipage}
\begin{minipage}{0.55\textwidth}
\footnotesize
\begin{align*}
CF_{AB|CD}&=1-2/3\exp(-t_1)\\
CF_{AC|BD}&=1/3\exp(-t_1)\\
CF_{AD|BC}&=1/3\exp(-t_1)
\end{align*}
\end{minipage}
\caption{Five different semi-directed level-1 4-taxon networks with one
  hybridization event, up to tip re-labelling.}
\label{5quartets}
\end{figure}

Note that there seems to be one type missing (type 1 with the
direction of the hybrid edge flipped). This network, however, has the
same theoretical CF formulas as the Type 2 network.

\subsubsection{Equivalence to quarnet types in \cite{Huber2018}}
\label{equivalence}
Note that our five quarnet types are related to the four quarnet types
defined in \cite{Huber2018} (their Figure 4).

In their work, \cite{Huber2018} define four types of quarnets:
\begin{itemize}
\item their type 1 corresponds to a quartet, which is equivalent to
  our type 5,
\item their type 2 is the undirected version of our type 1 and type 2,
\item their type 3 is not any of our types as we are restricting to
  the case of only one hybridization,
\item their type 4 corresponds to our type 3.
\end{itemize}
We have one extra type (type 4), since in \cite{Huber2018}, the authors
suppress any parallel edges and thus, our type 4 corresponds to a
quartet in \cite{Huber2018}.

%We list these CF formulas again in the Supplementary Material.

% \begin{figure*}
%   \centering
%   \includegraphics[scale=0.08]{figures/networks-hyb1.pdf}
%   \begin{tabular}{ccccc}
%     Type 1 \hspace{2cm} & Type 2 \hspace{1.5cm} & Type 3 \hspace{1.5cm} &
%                                                                       Type 4 \hspace{1.25cm} & Type 5\\
%   \end{tabular}
%     \caption{Five different semi-directed level-1 4-taxon networks with one
%       hybridization event, up to tip re-labelling.}
%     \label{5quartets}
% \end{figure*}

%%\section*{CF formulas for all possible level-1 4-taxon networks}

%\newpage

\section{Figures and tables for proof of Theorem \ref{topId}}
\label{figstabs}

  \begin{figure}[h!]
    \centering
    \begin{minipage}{.5\textwidth}
      \centering
      \includegraphics[scale=.15]{k3-net}
    \end{minipage}%
    \begin{minipage}{.5\textwidth}
      \centering
      \includegraphics[scale=.15]{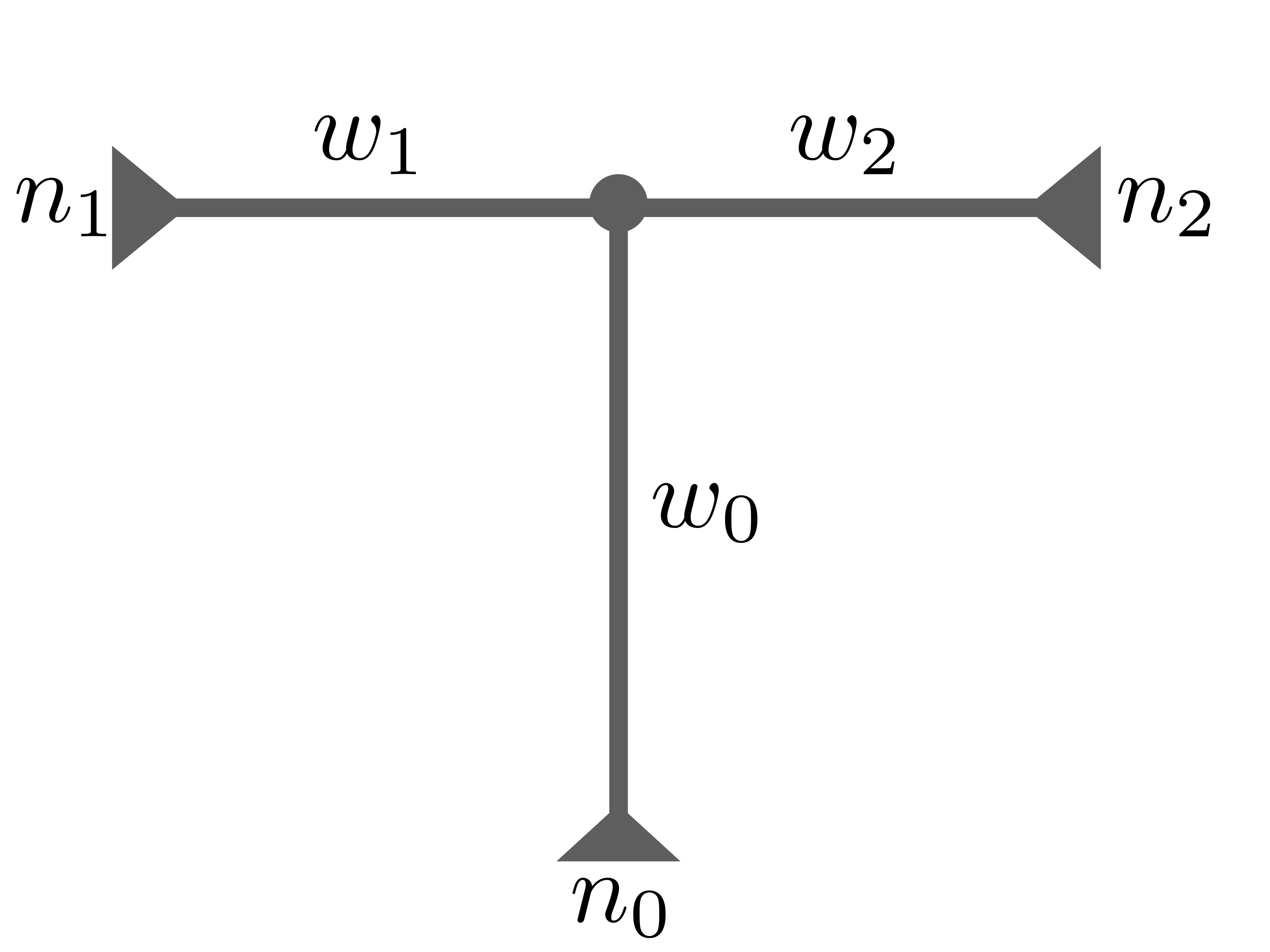}
    \end{minipage}
    \caption{3-cycle network and its corresponding major tree (see
      Lemma \ref{h1} for justification on comparison to the major
      tree). We find solutions to the equations
      $CF(\mathcal{N},\bs{z},\bs{\gamma})=CF(\mathcal{T},\bs{w})$ that
      would imply that the same set of CFs could be produced by both
      the network $\mathcal{N}$ and the tree $\mathcal{T}$. Here
      $z_i=\exp(-t_i)$ for branch length $t_i$ in $\mathcal{N}$, and
      $w_i=\exp(-t'_i)$ for branch length $t'_i$ in $\mathcal{T}$.}
    \label{k3nettree}
  \end{figure}

  \begin{table}[ht!]
    \caption{Systems of CF polynomial equations for the case of
      $k_i=3$. Here, $n=(n_0,n_1,n_2)$ in Figure \ref{k3nettree}, and
      ``Type'' corresponds to the type of quarnet in Figure
      \ref{5quartets}.}
    \label{k3table}
    \centering
    \begin{tabular}{llll}
      \toprule
      $ n $ & Type & $CF(\mathcal{N},\bs{z},\bs{\gamma})$  &
                                                             $CF(\mathcal{T},\bs{w})$ \\
      \midrule
      $ (0,2,2) $ & Tree & $ 1-\frac{2}{3} z_{1}z_{1,2}z_{2} $ & $1-\frac{2}{3}w_{2}w_{1}$
      \\
            && $\frac{1}{3} z_{1}z_{1,2}z_{2} $ & $\frac{1}{3}w_{2}w_{1}$ \\
            && $\frac{1}{3} z_{1}z_{1,2}z_{2} $ & $\frac{1}{3}w_{2}w_{1}$ \\
      \midrule
      $ (1,1,2) $ & $ 2 $ & $ (1- \gamma )\left(1-\frac{2}{3} z_{2}z_{1,2} \right) + \gamma \left(1-\frac{2}{3} z_{2} \right) $ & $ 1-\frac{2}{3} w_{2} $ \\ 
            & & $ (1- \gamma )\frac{1}{3} z_{2}z_{1,2} + \gamma \frac{1}{3}
        z_{2} $ & $ \frac{1}{3} w_{2} $ \\
            & & $ (1- \gamma )\frac{1}{3} z_{2}z_{1,2} + \gamma \frac{1}{3}
        z_{2} $ & $ \frac{1}{3} w_{2} $ \\
      \midrule
      $ (1,2,1) $ & $ 2 $ & $ (1- \gamma )\left(1-\frac{2}{3} z_{1} \right) + \gamma \left(1-\frac{2}{3} z_{1}z_{1,2} \right) $ & $ 1-\frac{2}{3} w_{1} $ \\ 
    & & $ (1- \gamma )\frac{1}{3} z_{1} + \gamma \frac{1}{3}
        z_{1}z_{1,2} $ & $ \frac{1}{3} w_{1} $ \\
    & & $ (1- \gamma )\frac{1}{3} z_{1} + \gamma \frac{1}{3}
        z_{1}z_{1,2} $ & $ \frac{1}{3} w_{1} $ \\
      \midrule
      $ (2,0,2) $ & $ 4 $ & $ (1- \gamma )^2\left(1-\frac{2}{3} z_{2}z_{0}z_{1,2}z_{0,1} \right) +2 \gamma (1- \gamma )\left(1-\frac{2}{3} z_{2}z_{0} \right) + \gamma ^2\left(1-\frac{2}{3} z_{2}z_{0}z_{0,2} \right) $ & $ 1-\frac{2}{3} w_{2}w_{0} $ \\ 
            & & $ (1- \gamma )^2\frac{1}{3} z_{2}z_{0}z_{1,2}z_{0,1}
                +2 \gamma (1- \gamma )\frac{1}{3} z_{2}z_{0} + \gamma
                ^2\frac{1}{3} z_{2}z_{0}z_{0,2} $ & $ \frac{1}{3}
                                                     w_{2}w_{0} $ \\
            & & $ (1- \gamma )^2\frac{1}{3} z_{2}z_{0}z_{1,2}z_{0,1}
                +2 \gamma (1- \gamma )\frac{1}{3} z_{2}z_{0} + \gamma
                ^2\frac{1}{3} z_{2}z_{0}z_{0,2} $ & $ \frac{1}{3}
                                                     w_{2}w_{0} $ \\
      \midrule
      $ (2,1,1) $ & $ 1 $ & $ (1- \gamma )^2\left(1-\frac{2}{3} z_{0}z_{0,1} \right) +2 \gamma (1- \gamma )\left(1- z_{0} +\frac{1}{3} z_{0}z_{1,2} \right) + \gamma ^2\left(1-\frac{2}{3} z_{0}z_{0,2} \right) $ & $ 1-\frac{2}{3} w_{0} $ \\ 
            & & $ (1- \gamma )^2\frac{1}{3} z_{0}z_{0,1} + \gamma (1-
                \gamma ) z_{0} \left(1-\frac{1}{3} z_{1,2} \right) +
                \gamma ^2\frac{1}{3} z_{0}z_{0,2} $ & $ \frac{1}{3}
                                                       w_{0} $ \\
            & & $ (1- \gamma )^2\frac{1}{3} z_{0}z_{0,1} + \gamma (1-
                \gamma ) z_{0} \left(1-\frac{1}{3} z_{1,2} \right) +
                \gamma ^2\frac{1}{3} z_{0}z_{0,2} $ & $ \frac{1}{3}
                                                       w_{0} $ \\
      \midrule
      $ (2,2,0) $ & $ 4 $ & $ (1- \gamma )^2\left(1-\frac{2}{3} z_{1}z_{0}z_{0,1} \right) +2 \gamma (1- \gamma )\left(1-\frac{2}{3} z_{1}z_{0} \right) + \gamma ^2\left(1-\frac{2}{3} z_{1}z_{0}z_{1,2}z_{0,2} \right) $ & $ 1-\frac{2}{3} w_{1}w_{0} $ \\ 
            & & $ (1- \gamma )^2\frac{1}{3} z_{1}z_{0}z_{0,1} +2
                \gamma (1- \gamma )\frac{1}{3} z_{1}z_{0} + \gamma
                ^2\frac{1}{3} z_{1}z_{0}z_{1,2}z_{0,2} $ & $
                                                            \frac{1}{3}
                                                            w_{1}w_{0}
                                                            $ \\
            & & $ (1- \gamma )^2\frac{1}{3} z_{1}z_{0}z_{0,1} +2
                \gamma (1- \gamma )\frac{1}{3} z_{1}z_{0} + \gamma
                ^2\frac{1}{3} z_{1}z_{0}z_{1,2}z_{0,2} $ & $
                                                            \frac{1}{3}
                                                            w_{1}w_{0}
                                                            $ \\
      \bottomrule
    \end{tabular}
    \end{table}

    \begin{figure}
      \centering
    \begin{minipage}{.5\textwidth}
      \centering
      \includegraphics[scale=.15]{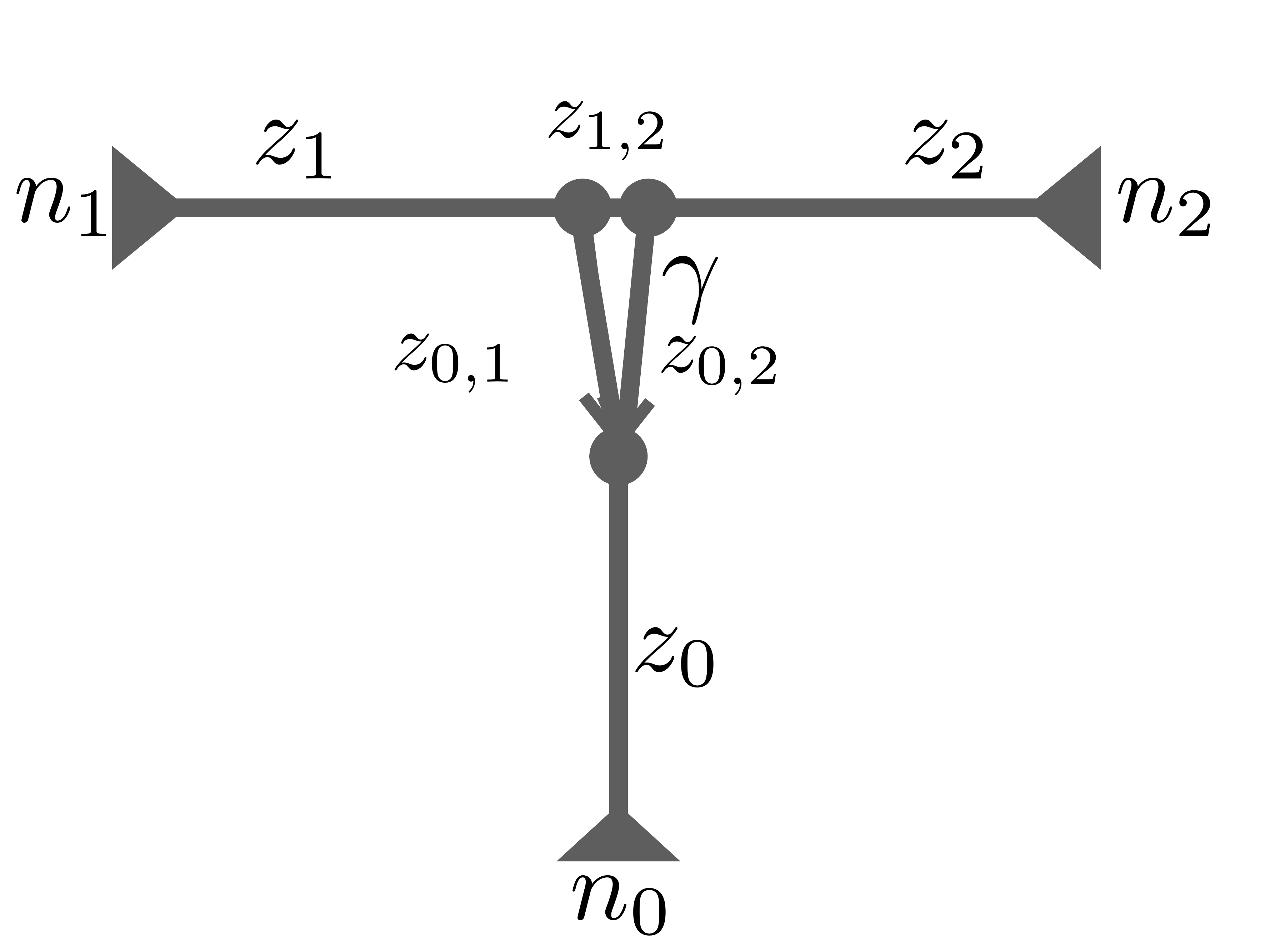} \\ %23
      \begin{align*}
        z_{1,2}&=1 \\
        (t_{1,2}&=0)
      \end{align*}
    \end{minipage}%
    \begin{minipage}{.5\textwidth}
      \centering
      \includegraphics[scale=.15]{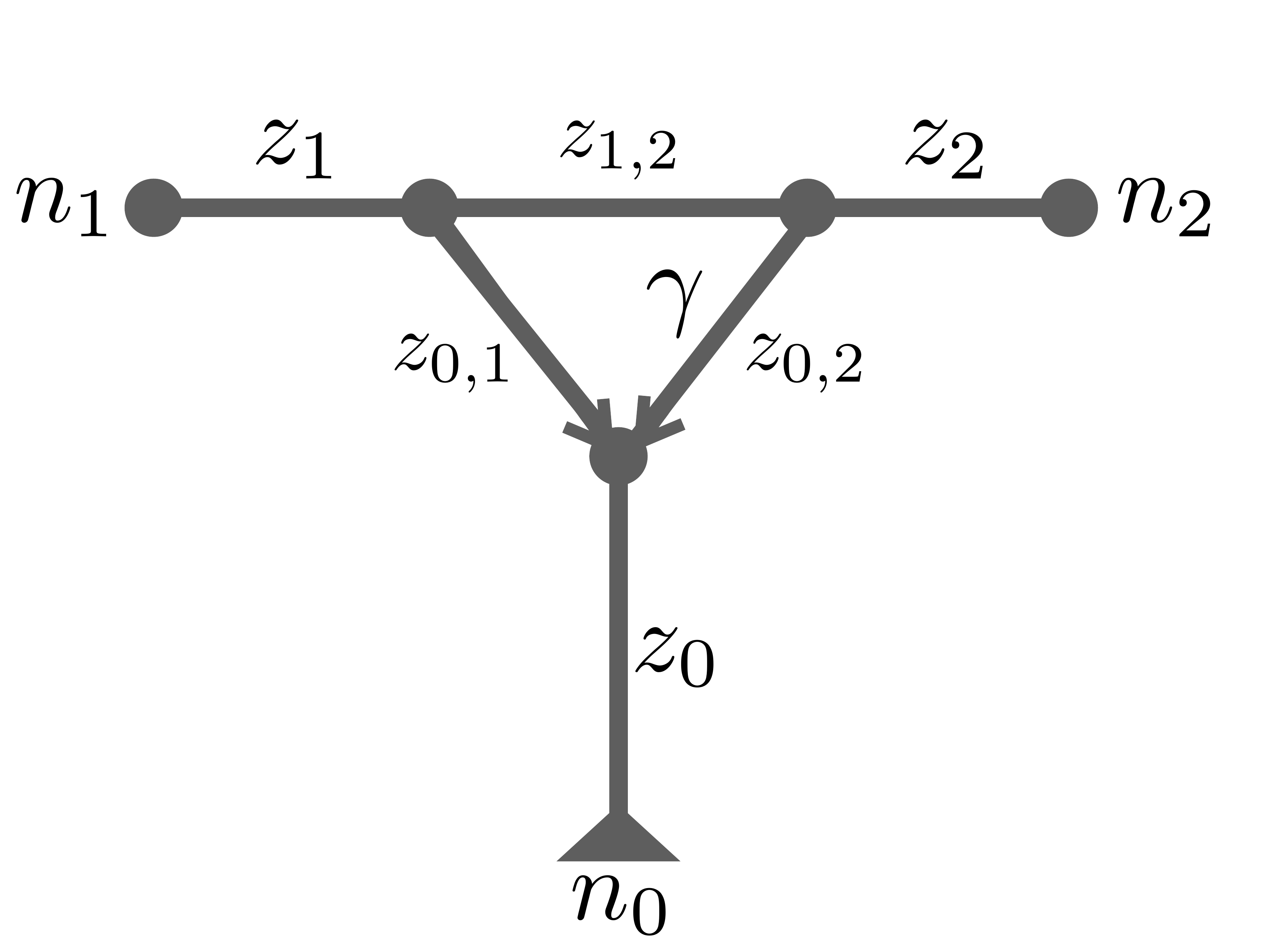} \\%24
      \begin{align*}
        z_1  =0&, z_2 =0\\
        (t_1  =\infty&, t_2 =\infty)
      \end{align*}
    \end{minipage}

    \begin{minipage}{.5\textwidth}
      \centering
      \includegraphics[scale=.15]{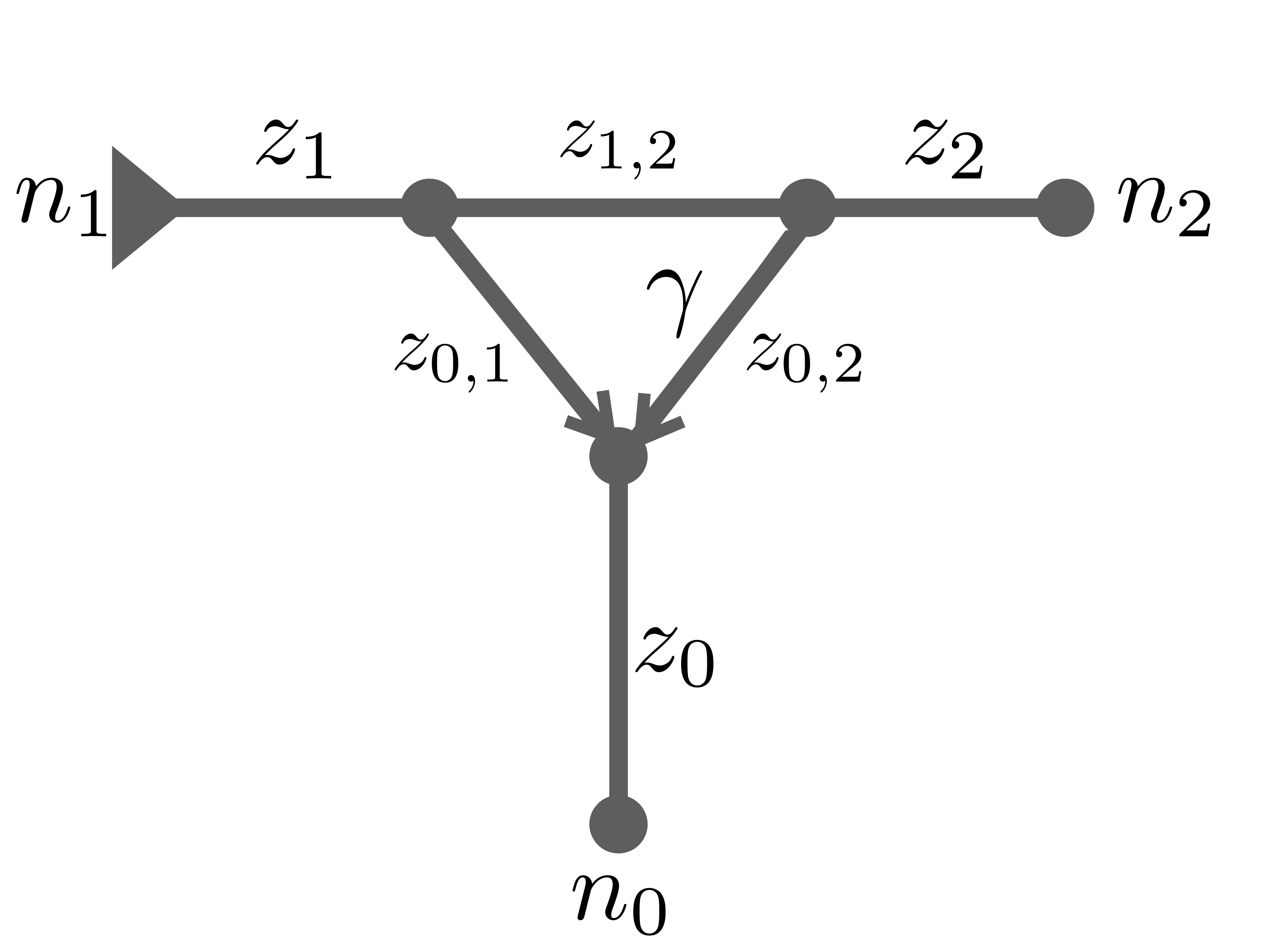} \\%27
      \begin{align*}
        z_0  =0&, z_2=0\\
        (t_0  =\infty&, t_2=\infty)
      \end{align*}
    \end{minipage}%
    \begin{minipage}{.5\textwidth}
      \centering
      \includegraphics[scale=.15]{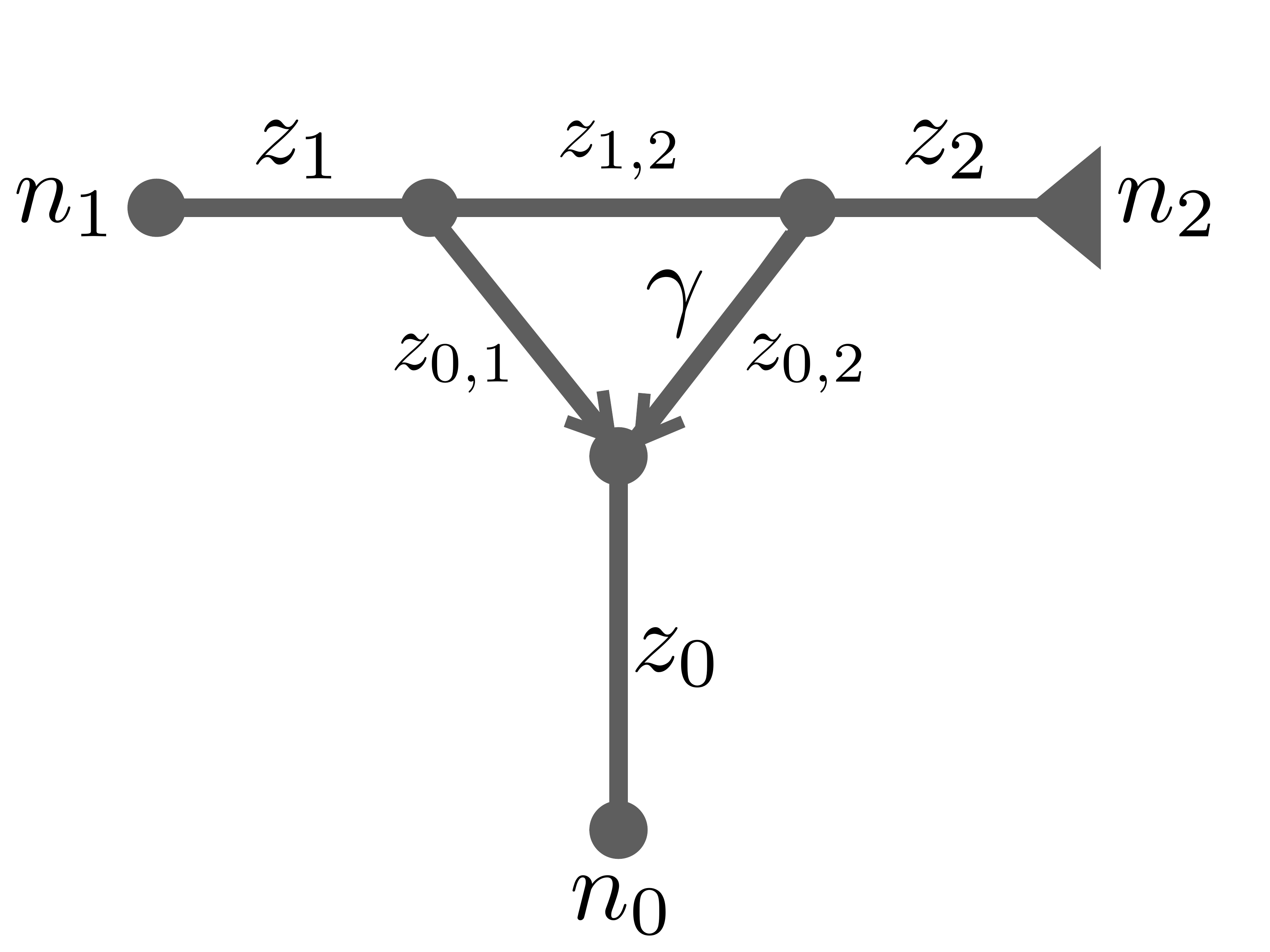} \\%28
      \begin{align*}
        z_0 =0&, z_1=0\\
        (t_0 =\infty&, t_1=\infty)
      \end{align*}
    \end{minipage}

    \begin{minipage}{.5\textwidth}
      \centering
      \includegraphics[scale=.15]{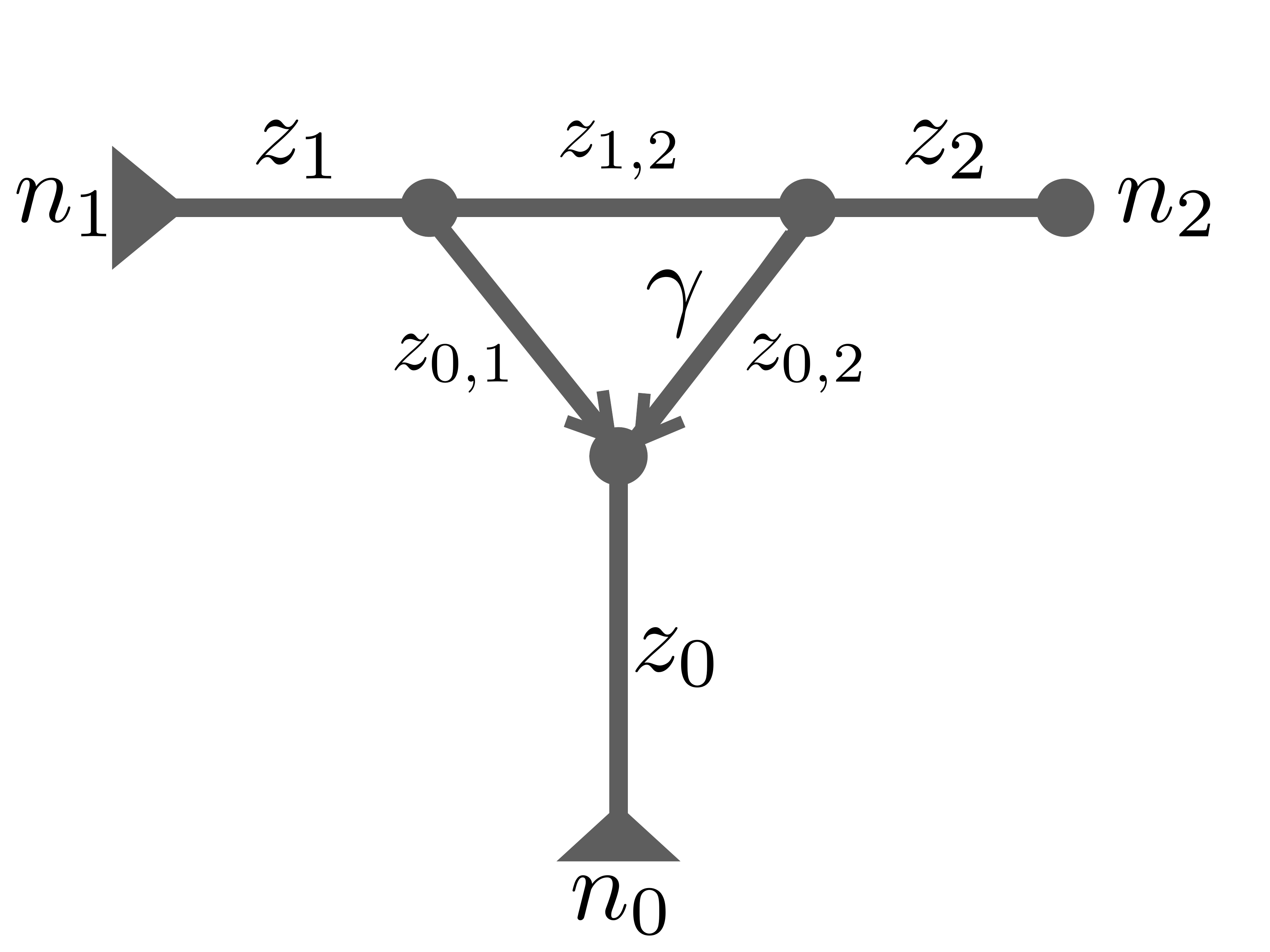} \\%25
      \begin{align*}
      z_2  &=0 (t_2=\infty) \\
        (1-\gamma)(z_{0,1}-1)&=\gamma(z_{0,2}+z_{1,2}-2)
      \end{align*}
    \end{minipage}%
    \begin{minipage}{.5\textwidth}
      \centering
      \includegraphics[scale=.15]{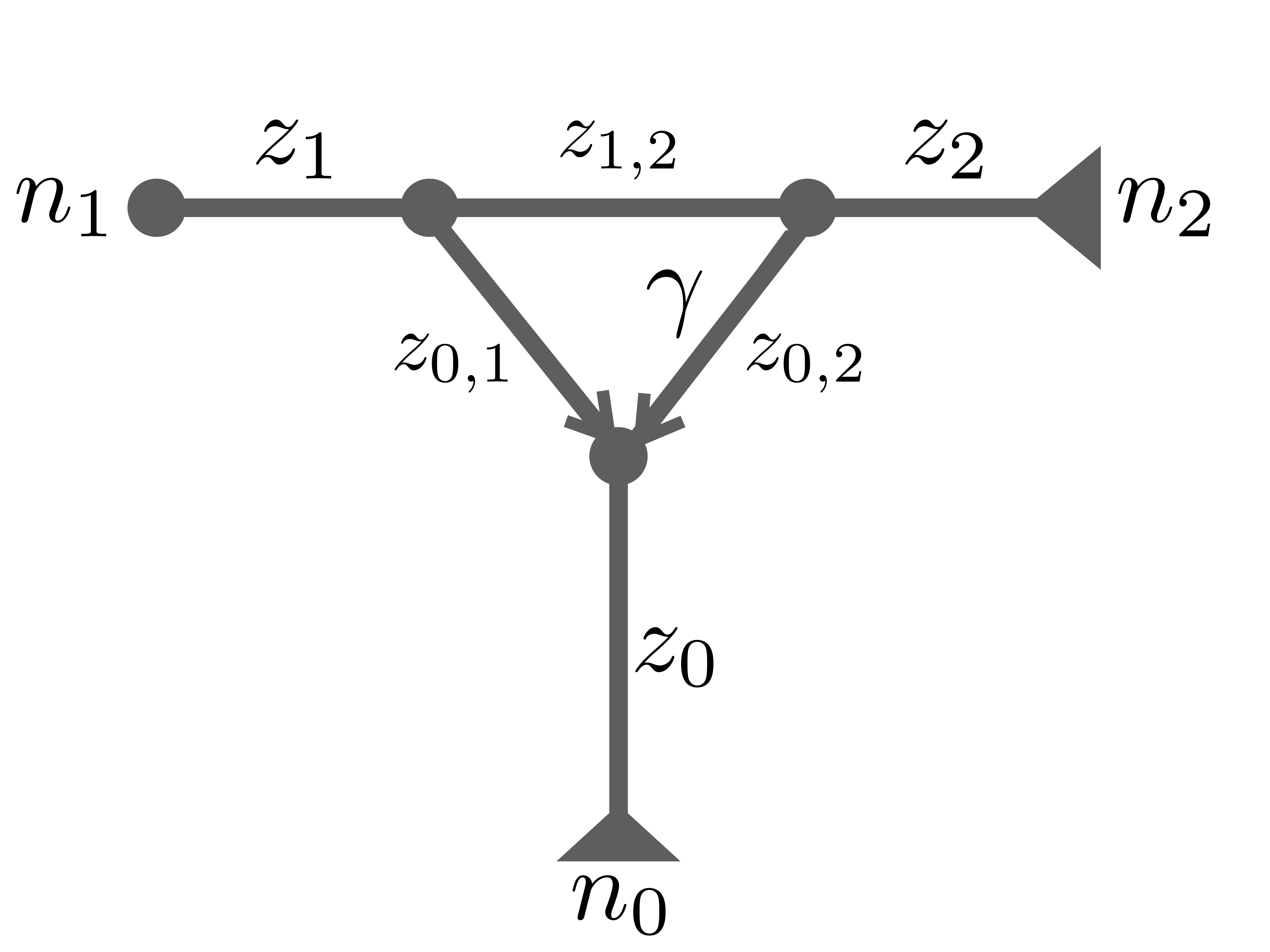} \\%26
      \begin{align*}
      z_1 &=0 (t_1=\infty) \\
      \gamma (z_{0,2}-1) &= (1-\gamma)(z_{1,2}+z_{0,1}-2)
      \end{align*}
    \end{minipage}%
    \caption{Non-detectable 3-cycles: Cases when the hybridization on
      a 3-cycle network is not detectable for $\gamma \in (0,1)$. Note
      that infinite branch lengths refer to the case of no ILS, so for
      example, on the second network on the first row, all individuals
      from the subgraphs labeled $n_1$ and $n_2$ would have coalesced
      on an infinitely long branch, and thus, the subgraph triangles
      are replaced by nodes which represent $n_1=n_2=1$. These cases
      with infinitely long branches violate the A1-A2 assumptions.}
    \label{k3path}
  \end{figure}

    \begin{figure}[h!]
    \centering
    \begin{minipage}{.5\textwidth}
      \centering
      \includegraphics[scale=.15]{k4-net}
    \end{minipage}%
    \begin{minipage}{.5\textwidth}
      \centering
      \includegraphics[scale=.15]{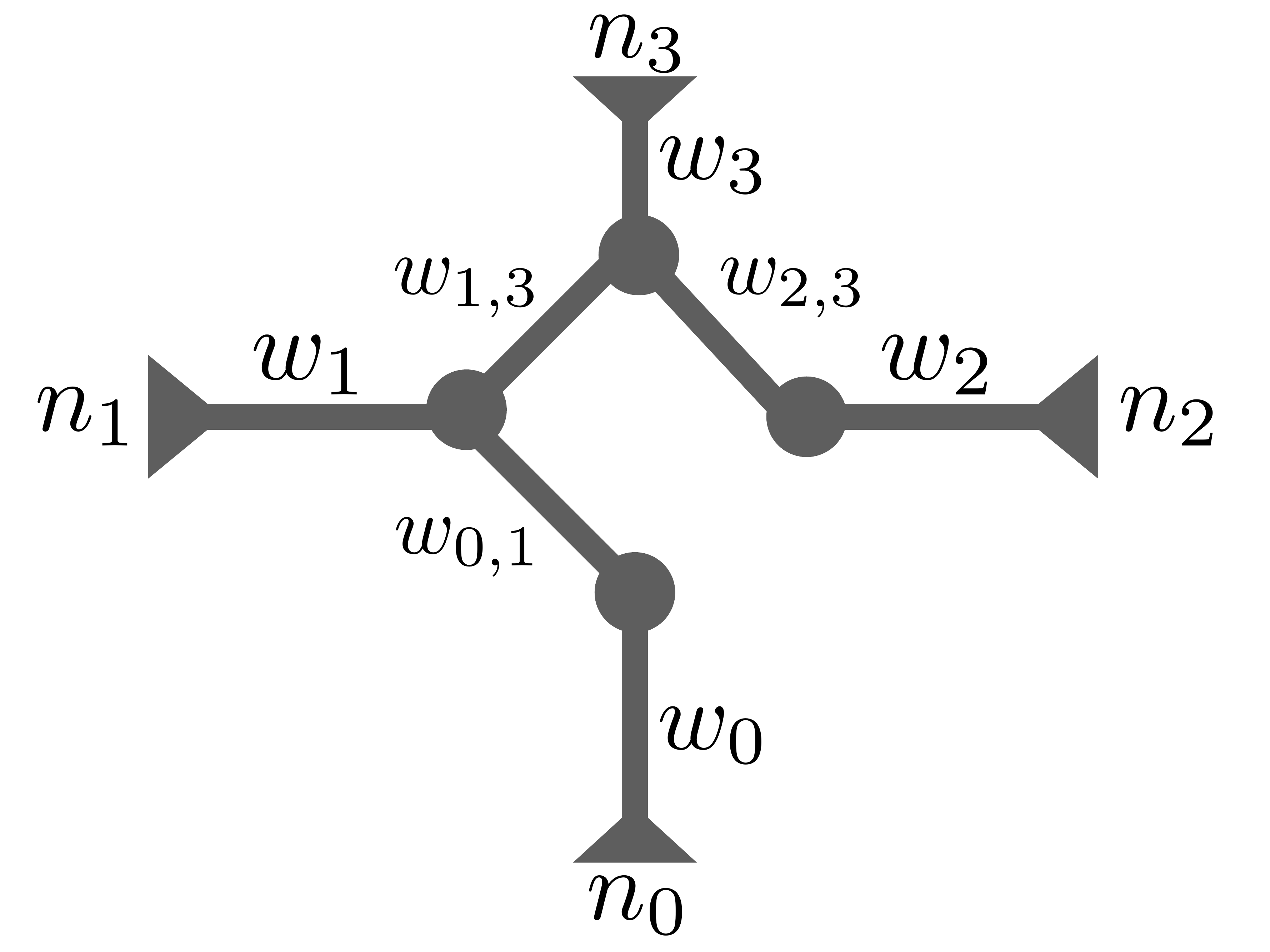}
    \end{minipage}
    \caption{4-cycle network and its corresponding tree representation
      (see Lemma \ref{h1} for justification on comparison to a
      tree). We find solutions to the equations
      $CF(\mathcal{N},\bs{z},\bs{\gamma})=CF(\mathcal{T},\bs{w})$ that
      would imply that the same set of CFs could be produced by both
      the network $\mathcal{N}$ and the tree $\mathcal{T}$. Note that
      we do not eliminate degree-2 nodes in the tree $\mathcal{T}$
      simply for ease of notation and comparison to
      $\mathcal{N}$. Here $z_i=\exp(-t_i)$ for branch length $t_i$ in
      $\mathcal{N}$, and $w_i=\exp(-t'_i)$ for branch length $t'_i$ in
      $\mathcal{T}$.}
    \label{k4nettree}
  \end{figure}

   \begin{figure}[h!]
    \centering
    \begin{minipage}{.5\textwidth}
      \centering
      \includegraphics[scale=.15]{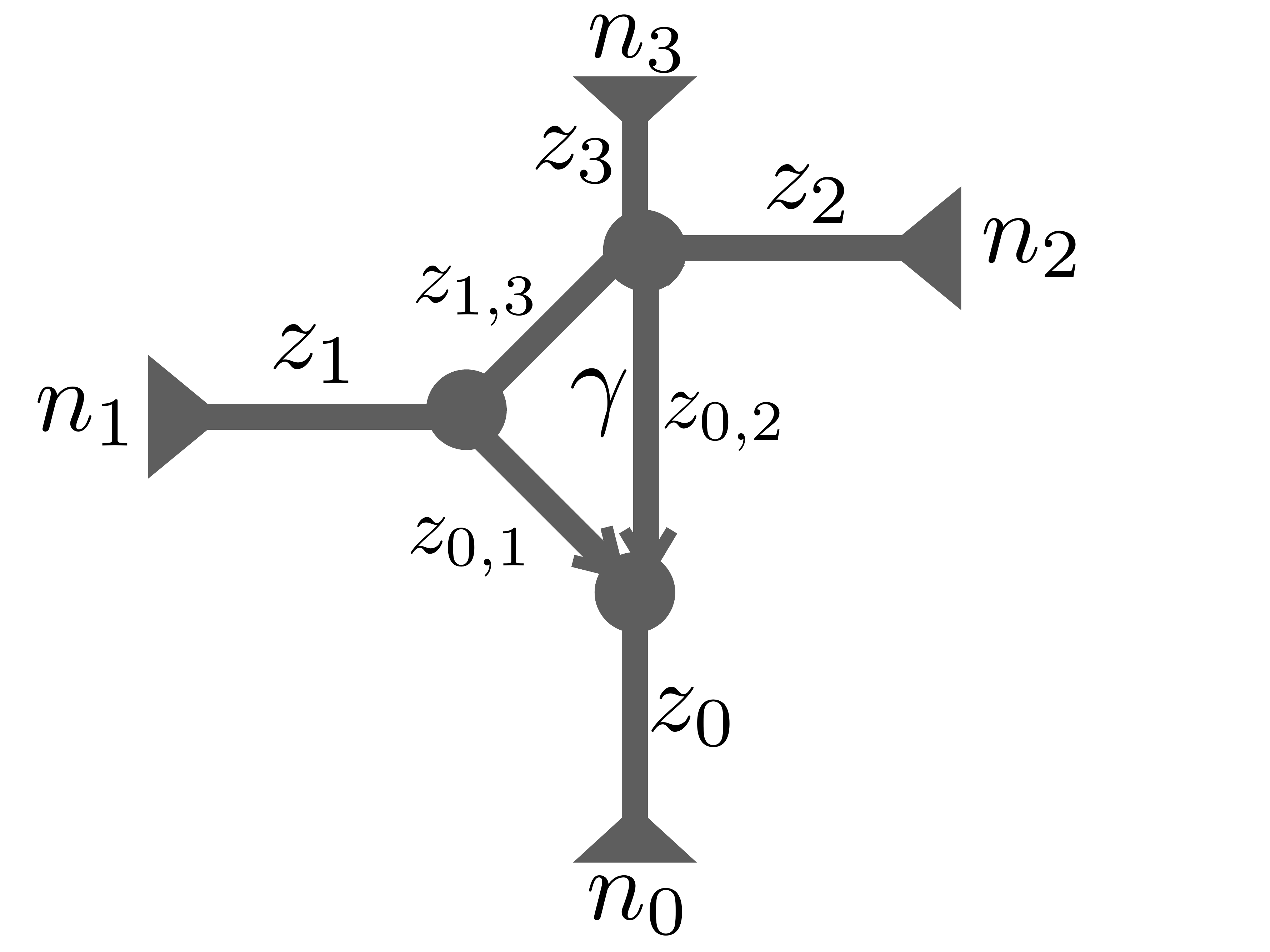}
    \end{minipage}%
    \begin{minipage}{.5\textwidth}
      \centering
      \includegraphics[scale=.15]{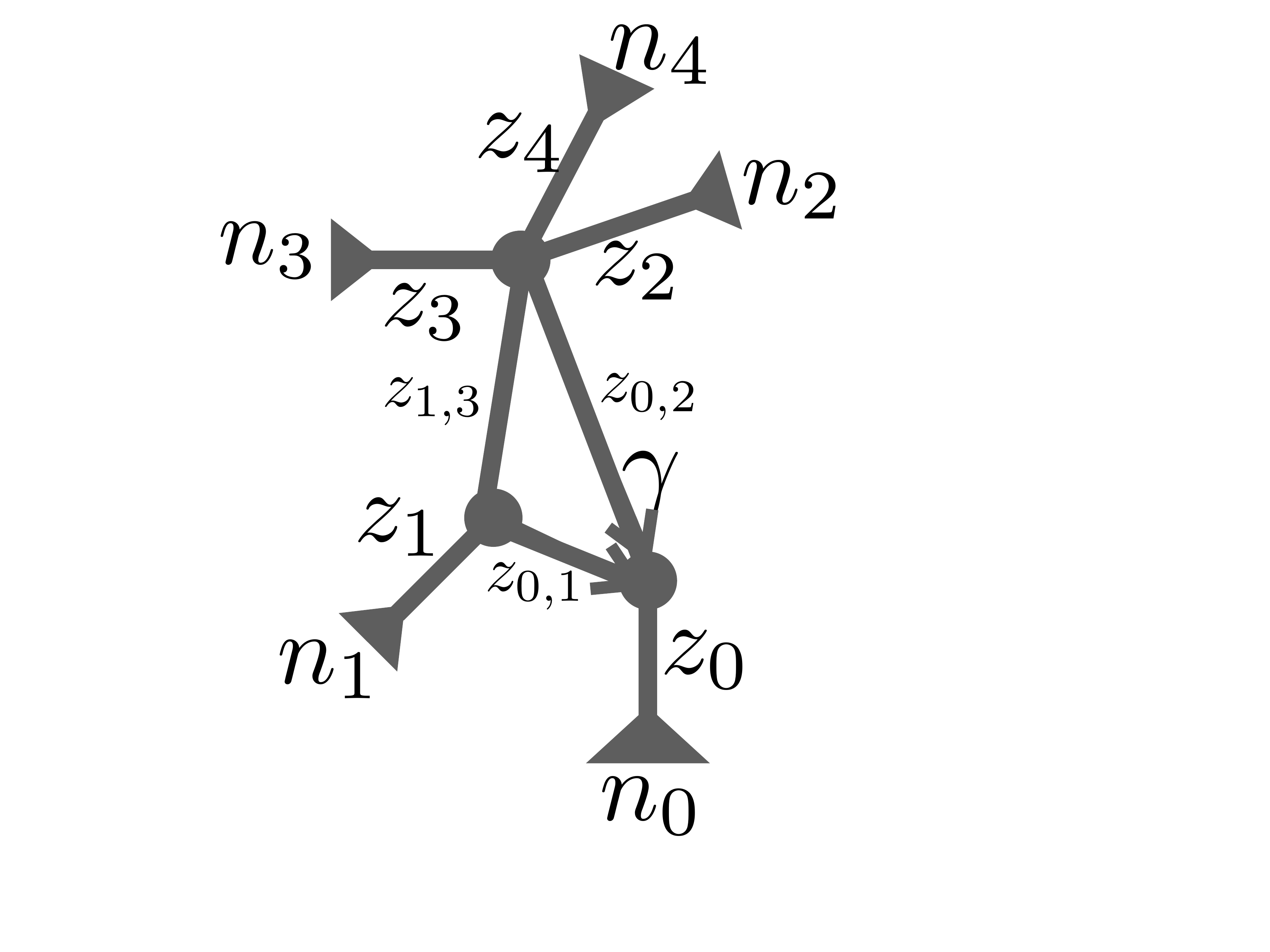}
    \end{minipage}
    \caption{Left: Non-detectable 4-cycle network: Cases when the
      hybridization on a 4-cycle network is not detectable for
      $\gamma \in (0,1)$ correspond to the case when $t_{2,3}=0$ in
      Figure \ref{k4nettree} and thus, the 4-cycle becomes a
      3-cycle. Right: Non-detectable $k$-cycle network ($k \geq 5$):
      Cases when the hybridization on a $k$-cycle network for
      $k \geq 5$ is not detectable for $\gamma \in (0,1)$ correspond
      to the case when $\tilde{t}=t_{2,4}=0$ in Figure \ref{k5nettree}
      and thus, the $k$-cycle becomes a 3-cycle.}
    \label{k4path}
  \end{figure}

    \begin{table}[ht!]
      \caption{Systems of CF polynomial equations for the case of
        $k_i=4$. Here we do not repeat the minor CF equations for the
        cases when there are two equal. Also, $n=(n_0,n_1,n_2,n_3)$ as
        in Figure \ref{k4nettree}, and ``Type'' corresponds to the
        type of quarnet in Figure \ref{5quartets}.}
    \label{k4table}
    \centering
    \begin{tabular}{llll}
      \toprule
      $ n $ & Type & $CF(\mathcal{N},\bs{z},\bs{\gamma})$  &
                                                             $CF(\mathcal{T},\bs{w})$ \\
      \midrule
$ (0,0,2,2) $ &  Tree  & $ 1-\frac{2}{3} z_{2}z_{2,3}z_{3} $ & $
                                                               1-\frac{2}{3}
                                                               w_{2}w_{2,3}w_{3}
                                                               $ \\
    &  & $ \frac{1}{3} z_{2}z_{2,3}z_{3} $ & $ \frac{1}{3}
                                             w_{2}w_{2,3}w_{3} $ \\ \midrule
			$ (0,1,2,1) $ &  Tree  & $ 1-\frac{2}{3}
                                                 z_{2,3}z_{2} $ & $
                                                                  1-\frac{2}{3}
                                                                  w_{2,3}w_{2}
                                                                  $ \\
      && $ \frac{1}{3} z_{2,3}z_{2} $ & $ \frac{1}{3} w_{2,3}w_{2} $ \\      \midrule
			$ (0,1,1,2) $ &  Tree  & $ 1-\frac{2}{3} z_{3}
                                                 $ & $ 1-\frac{2}{3}
                                                     w_{3} $ \\
      && $ \frac{1}{3} z_{3} $ & $ \frac{1}{3} w_{3} $ \\      \midrule
			$ (0,2,2,0) $ &  Tree  & $ 1-\frac{2}{3}
                                                 z_{2}z_{2,3}z_{1,3}z_{1}
                                                 $ & $ 1-\frac{2}{3}
                                                     w_{2}w_{2,3}w_{1,3}w_{1}
                                                     $ \\ 
      && $ \frac{1}{3} z_{2}z_{2,3}z_{1,3}z_{1} $ & $ \frac{1}{3} w_{2}w_{2,3}w_{1,3}w_{1} $ \\ \midrule
			$ (0,2,1,1) $ &  Tree  & $ 1-\frac{2}{3} z_{1}z_{1,3}
                                                 $ & $ 1-\frac{2}{3}
                                                     w_{1,3}w_{1} $ \\
      && $ \frac{1}{3} z_{1}z_{1,3} $ & $ \frac{1}{3} w_{1,3}w_{1} $ \\      \midrule
			$ (0,2,0,2) $ &  Tree  & $ 1-\frac{2}{3}
                                                 z_{3}z_{1,3}z_{1} $ &
                                                                       $ 1-\frac{2}{3} w_{3}w_{1,3}w_{1} $ \\
           & & $ \frac{1}{3} z_{3}z_{1,3}z_{1} $ & $ \frac{1}{3} w_{3}w_{1,3}w_{1} $ \\ \midrule
			$ (1,0,2,1) $ & $ 2 $ & $ (1- \gamma )\left(1-\frac{2}{3} z_{2,3}z_{2} \right) + \gamma \left(1-\frac{2}{3} z_{2} \right) $ & $ 1-\frac{2}{3} w_{2,3}w_{2} $ \\ 
			& & $ (1- \gamma )\frac{1}{3} z_{2,3}z_{2} + \gamma \frac{1}{3} z_{2} $ & $ \frac{1}{3} w_{2,3}w_{2} $ \\ \midrule
			$ (1,0,1,2) $ & $ 2 $ & $ (1- \gamma )\left(1-\frac{2}{3} z_{3} \right) + \gamma \left(1-\frac{2}{3} z_{2,3}z_{3} \right) $ & $ 1-\frac{2}{3} w_{3} $ \\ 
			& & $ (1- \gamma )\frac{1}{3} z_{3} + \gamma \frac{1}{3} z_{2,3}z_{3} $ & $ \frac{1}{3} w_{3} $ \\ \midrule
			$ (1,1,2,0) $ & $ 2 $ & $ (1- \gamma )\left(1-\frac{2}{3} z_{1,3}z_{2,3}z_2 \right) + \gamma \left(1-\frac{2}{3} z_{2} \right) $ & $ 1-\frac{2}{3} w_{1,3}w_{2,3}w_{2} $ \\ 
			& & $ (1- \gamma )\frac{1}{3} z_{1,3}z_{2,3}z_2 + \gamma \frac{1}{3} z_{2} $ & $ \frac{1}{3} w_{1,3}w_{2,3}w_{2} $ \\ \midrule
			$ (1,1,1,1) $ & $ 3 $ & $ (1- \gamma
                                                )\left(1-\frac{2}{3}
                                                z_{1,3} \right) +
                                                \gamma \frac{1}{3}
                                                z_{2,3} $ & $
                                                            1-\frac{2}{3}
                                                            w_{1,3} $
      \\
      			& & $ (1- \gamma )\frac{1}{3} z_{1,3} + \gamma \left(1-\frac{2}{3} z_{2,3} \right) $ & $ \frac{1}{3} w_{1,3} $ \\ 
			& & $ (1- \gamma )\frac{1}{3} z_{1,3} + \gamma \frac{1}{3} z_{2,3} $ & $ \frac{1}{3} w_{1,3} $ \\
      \bottomrule
    \end{tabular}
    \end{table}

\begin{landscape}
    \begin{table}[ht!]
      \caption{Systems of CF polynomial equations for the case of
        $k_i=4$. Here we do not repeat the minor CF equations for the
        cases when there are two equal. Also, $n=(n_0,n_1,n_2,n_3)$ as
        in Figure \ref{k4nettree}, and ``Type'' corresponds to the
        type of quarnet in Figure \ref{5quartets}.}
      \label{k4table-2}
      \centering
      \begin{tabular}{llll}
        \toprule
      $ n $ & Type & $CF(\mathcal{N},\bs{z},\bs{\gamma})$  &
                                                             $CF(\mathcal{T},\bs{w})$ \\
      \midrule
			$ (1,1,0,2) $ & $ 2 $ & $ (1- \gamma )\left(1-\frac{2}{3} z_{1,3}z_{3} \right) + \gamma \left(1-\frac{2}{3} z_{3} \right) $ & $ 1-\frac{2}{3} w_{1,3}w_{3} $ \\ 
			& & $ (1- \gamma )\frac{1}{3} z_{1,3}z_{3} + \gamma \frac{1}{3} z_{3} $ & $ \frac{1}{3} w_{1,3}w_{3} $ \\ \midrule
			$ (1,2,1,0) $ & $ 2 $ & $ (1- \gamma )\left(1-\frac{2}{3} z_{1} \right) + \gamma \left(1-\frac{2}{3} z_{2,3}z_{1,3}z_{1} \right) $ & $ 1-\frac{2}{3} w_{1} $ \\ 
			& & $ (1- \gamma )\frac{1}{3} z_{1} + \gamma \frac{1}{3} z_{2,3}z_{1,3}z_{1} $ & $ \frac{1}{3} w_{1} $ \\ \midrule
			$ (1,2,0,1) $ & $ 2 $ & $ (1- \gamma )\left(1-\frac{2}{3} z_{1} \right) + \gamma \left(1-\frac{2}{3} z_{1,3}z_{1} \right) $ & $ 1-\frac{2}{3} w_{1} $ \\ 
			& & $ (1- \gamma )\frac{1}{3} z_{1} + \gamma \frac{1}{3} z_{1,3}z_{1} $ & $ \frac{1}{3} w_{1} $ \\ \midrule
			$ (2,0,2,0) $ & $ 4 $ & $ (1- \gamma )^2\left(1-\frac{2}{3} z_{2}z_{0}z_{0,1}z_{1,3}z_{2,3} \right) +2 \gamma (1- \gamma )\left(1-\frac{2}{3} z_{2}z_{0} \right) + \gamma ^2\left(1-\frac{2}{3} z_{2}z_{0}z_{0,2} \right) $ & $ 1-\frac{2}{3} w_{0}w_{0,1}w_{1,3}w_{2,3}w_{2} $ \\ 
			& & $ (1- \gamma )^2\frac{1}{3} z_{2}z_{0}z_{0,1}z_{1,3}z_{2,3} +2 \gamma (1- \gamma )\frac{1}{3} z_{2}z_{0} + \gamma ^2\frac{1}{3} z_{2}z_{0}z_{0,2} $ & $ \frac{1}{3} w_{0}w_{0,1}w_{1,3}w_{2,3}w_{2} $ \\ \midrule
			$ (2,0,1,1) $ & $ 1 $ & $ (1- \gamma )^2\left(1-\frac{2}{3} z_{0}z_{1,3}z_{0,1} \right) +2 \gamma (1- \gamma )\left(1- z_{0} +\frac{1}{3} z_{0}z_{2,3} \right) + \gamma ^2\left(1-\frac{2}{3} z_{0}z_{0,2} \right) $ & $ 1-\frac{2}{3} w_{1,3}w_{0,1}w_{0} $ \\ 
			& & $ (1- \gamma )^2\frac{1}{3} z_{0}z_{1,3}z_{0,1} + \gamma (1- \gamma ) z_{0} \left(1-\frac{1}{3} z_{2,3} \right) + \gamma ^2\frac{1}{3} z_{0}z_{0,2} $ & $ \frac{1}{3} w_{1,3}w_{0,1}w_{0} $ \\ \midrule
			$ (2,0,0,2) $ & $ 4 $ & $ (1- \gamma )^2\left(1-\frac{2}{3} z_{3}z_{0}z_{1,3}z_{0,1} \right) +2 \gamma (1- \gamma )\left(1-\frac{2}{3} z_{3}z_{0} \right) + \gamma ^2\left(1-\frac{2}{3} z_{3}z_{0}z_{2,3}z_{0,2} \right) $ & $ 1-\frac{2}{3} w_{3}w_{1,3}w_{0,1}w_{0} $ \\ 
			& & $ (1- \gamma )^2\frac{1}{3} z_{3}z_{0}z_{1,3}z_{0,1} +2 \gamma (1- \gamma )\frac{1}{3} z_{3}z_{0} + \gamma ^2\frac{1}{3} z_{3}z_{0}z_{2,3}z_{0,2} $ & $ \frac{1}{3} w_{3}w_{1,3}w_{0,1}w_{0} $ \\ \midrule
			$ (2,1,1,0) $ & $ 1 $ & $ (1- \gamma )^2\left(1-\frac{2}{3} z_{0}z_{0,1} \right) +2 \gamma (1- \gamma )\left(1- z_{0} +\frac{1}{3} z_{0}z_{2,3}z_{1,3} \right) + \gamma ^2\left(1-\frac{2}{3} z_{0}z_{0,2} \right) $ & $ 1-\frac{2}{3} w_{0,1}w_{0} $ \\ 
			& & $ (1- \gamma )^2\frac{1}{3} z_{0}z_{0,1} + \gamma (1- \gamma ) z_{0} \left(1-\frac{1}{3} z_{2,3}z_{1,3} \right) + \gamma ^2\frac{1}{3} z_{0}z_{0,2} $ & $ \frac{1}{3} w_{0,1}w_{0} $ \\ \midrule
			$ (2,1,0,1) $ & $ 1 $ & $ (1- \gamma )^2\left(1-\frac{2}{3} z_{0}z_{0,1} \right) +2 \gamma (1- \gamma )\left(1- z_{0} +\frac{1}{3} z_{0}z_{1,3} \right) + \gamma ^2\left(1-\frac{2}{3} z_{0}z_{0,2}z_{2,3} \right) $ & $ 1-\frac{2}{3} w_{0,1}w_{0} $ \\ 
			& & $ (1- \gamma )^2\frac{1}{3} z_{0}z_{0,1} + \gamma (1- \gamma ) z_{0} \left(1-\frac{1}{3} z_{1,3} \right) + \gamma ^2\frac{1}{3} z_{0}z_{0,2}z_{2,3} $ & $ \frac{1}{3} w_{0,1}w_{0} $ \\ \midrule
			$ (2,2,0,0) $ & $ 4 $ & $ (1- \gamma )^2\left(1-\frac{2}{3} z_{1}z_{0}z_{0,1} \right) +2 \gamma (1- \gamma )\left(1-\frac{2}{3} z_{1}z_{0} \right) + \gamma ^2\left(1-\frac{2}{3} z_{1}z_{0}z_{0,2}z_{2,3}z_{1,3} \right) $ & $ 1-\frac{2}{3} w_{1}w_{0,1}w_{0} $ \\ 
			& & $ (1- \gamma )^2\frac{1}{3} z_{1}z_{0}z_{0,1} +2 \gamma (1- \gamma )\frac{1}{3} z_{1}z_{0} + \gamma ^2\frac{1}{3} z_{1}z_{0}z_{0,2}z_{2,3}z_{1,3} $ & $ \frac{1}{3} w_{1}w_{0,1}w_{0} $ \\
      \bottomrule
    \end{tabular}
    \end{table}
  \end{landscape}

    \begin{figure}[h!]
    \centering
    \begin{minipage}{.5\textwidth}
      \centering
      \includegraphics[scale=.15]{k5-net}
    \end{minipage}%
    \begin{minipage}{.5\textwidth}
      \centering
      \includegraphics[scale=.15]{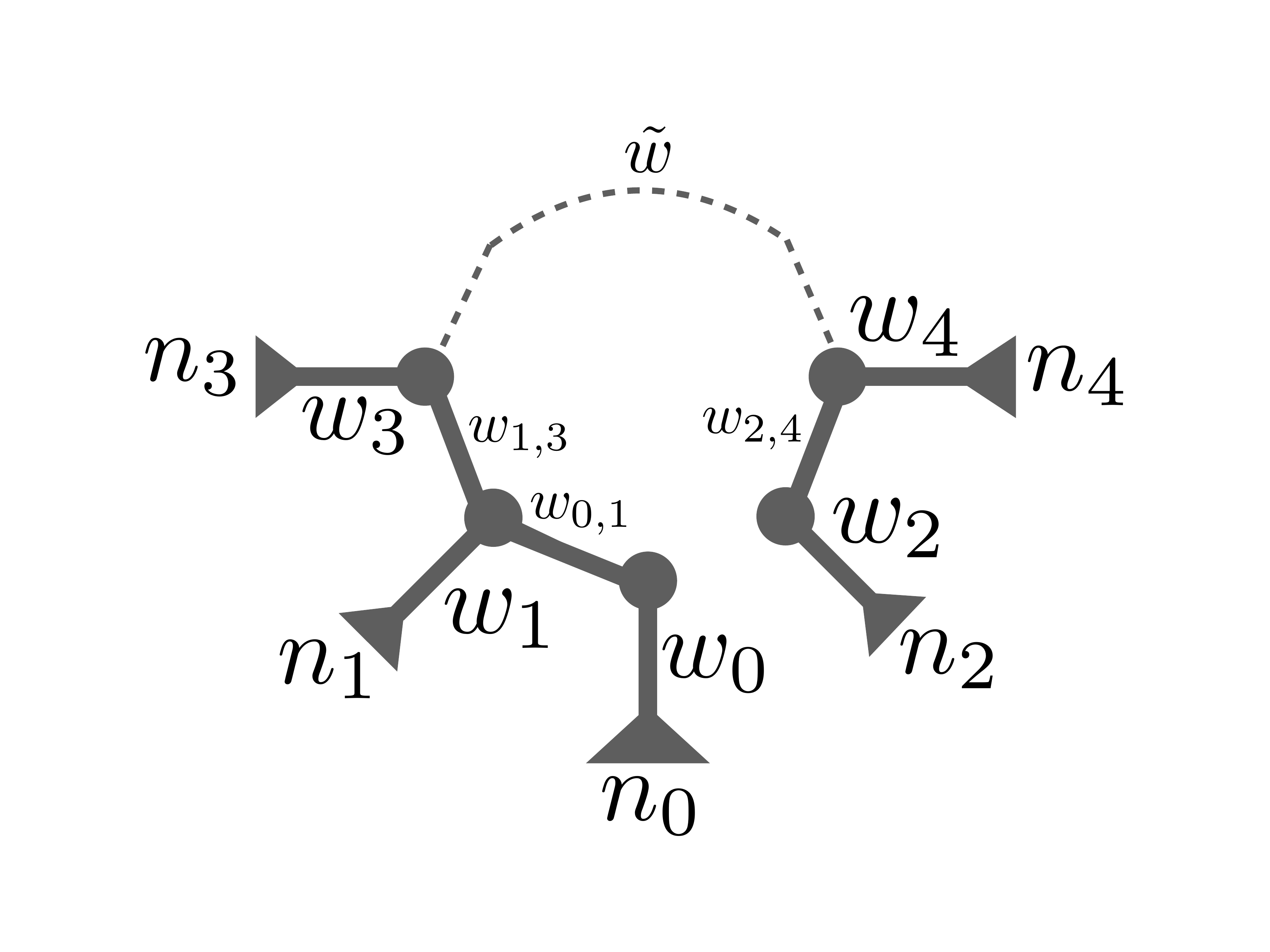}
    \end{minipage}
    \caption{5-cycle network and its corresponding tree representation
      (see Lemma \ref{h1} for justification on comparison to tree). We
      find solutions to the equations
      $CF(\mathcal{N},\bs{z},\bs{\gamma})=CF(\mathcal{T},\bs{w})$ that
      would imply that the same set of CFs could be produced by both
      the network $\mathcal{N}$ and the tree $\mathcal{T}$. The dotted
      path represent the possibility of a large cycle ($k>5$) which by
      Lemma \ref{k5} has the same hybridization detectability
      properties. Note that we do not eliminate degree-2 nodes in the
      tree $\mathcal{T}$ simply for ease of notation and comparison to
      $\mathcal{N}$. Here $z_i=\exp(-t_i)$ for branch length $t_i$ in
      $\mathcal{N}$, and $w_i=\exp(-t'_i)$ for branch length $t'_i$ in
      $\mathcal{T}$.}
    \label{k5nettree}
  \end{figure}

\begin{landscape}
    \begin{table}[ht!]
      \caption{Systems of CF polynomial equations for the case of
        $k_i\geq 5$. Here we do not repeat the minor CF equations for
        the cases when there are two equal. Also,
        $n=(n_0,n_1,n_2,n_3,n_4)$ as in Figure \ref{k5nettree}, and
        ``Type'' corresponds to the type of quarnet in Figure
        \ref{5quartets}.}
      \label{k5table-1}
      \centering
      \begin{tabular}{llll}
        \toprule
$ (0,2,0,2,0) $ & Tree & $ 1-\frac{2}{3} z_{3}z_{1,3}z_{1} $ & $
                                                                1-\frac{2}{3}
                                                                w_{3}w_{1,3}w_{1}
                                                                $ \\
        & & $ \frac{1}{3} z_{3}z_{1,3}z_{1} $ & $ \frac{1}{3} w_{3}w_{1,3}w_{1} $ \\ \midrule
			$ (1,1,0,2,0) $ & $ 2 $ & $ (1- \gamma )\left(1-\frac{2}{3} z_{1,3}z_{3} \right) + \gamma \left(1-\frac{2}{3} z_{3} \right) $ & $ 1-\frac{2}{3} w_{3}w_{1,3} $ \\ 
			& & $ (1- \gamma )\frac{1}{3} z_{1,3}z_{3} + \gamma \frac{1}{3} z_{3} $ & $ \frac{1}{3} w_{3}w_{1,3} $ \\ \midrule
			$ (1,2,0,1,0) $ & $ 2 $ & $ (1- \gamma )\left(1-\frac{2}{3} z_{1} \right) + \gamma \left(1-\frac{2}{3} z_{1,3}z_{1} \right) $ & $ 1-\frac{2}{3} w_{1} $ \\ 
			& & $ (1- \gamma )\frac{1}{3} z_{1} + \gamma \frac{1}{3} z_{1,3}z_{1} $ & $ \frac{1}{3} w_{1} $ \\ \midrule
			$ (2,0,0,2,0) $ & $ 4 $ & $ (1- \gamma )^2\left(1-\frac{2}{3} z_{3}z_{0}z_{0,1}z_{1,3} \right) +2 \gamma (1- \gamma )\left(1-\frac{2}{3} z_{3}z_{0} \right) + \gamma ^2\left(1-\frac{2}{3} z_{3}z_{0}z_{0,2}z_{2,4}z_{4,6}z_{6,7}z_{5,7}z_{3,5} \right) $ & $ 1-\frac{2}{3} w_{3}w_{1,3}w_{0,1}w_{0} $ \\ 
			& & $ (1- \gamma )^2\frac{1}{3} z_{3}z_{0}z_{0,1}z_{1,3} +2 \gamma (1- \gamma )\frac{1}{3} z_{3}z_{0} + \gamma ^2\frac{1}{3} z_{3}z_{0}z_{0,2}z_{2,4}z_{4,6}z_{6,7}z_{5,7}z_{3,5} $ & $ \frac{1}{3} w_{3}w_{1,3}w_{0,1}w_{0} $ \\ \midrule
			$ (2,1,0,1,0) $ & $ 1 $ & $ (1- \gamma )^2\left(1-\frac{2}{3} z_{0}z_{0,1} \right) +2 \gamma (1- \gamma )\left(1- z_{0} +\frac{1}{3} z_{0}z_{1,3} \right) + \gamma ^2\left(1-\frac{2}{3} z_{0}z_{0,2}z_{2,4}z_{4,6}z_{6,7}z_{5,7}z_{3,5} \right) $ & $ 1-\frac{2}{3} w_{0,1}w_{0} $ \\ 
			& & $ (1- \gamma )^2\frac{1}{3} z_{0}z_{0,1} + \gamma (1- \gamma ) z_{0} \left(1-\frac{1}{3} z_{1,3} \right) + \gamma ^2\frac{1}{3} z_{0}z_{0,2}z_{2,4}z_{4,6}z_{6,7}z_{5,7}z_{3,5} $ & $ \frac{1}{3} w_{0,1}w_{0} $ \\ \midrule
			$ (2,2,0,0,0) $ & $ 4 $ & $ (1- \gamma )^2\left(1-\frac{2}{3} z_{1}z_{0}z_{0,1} \right) +2 \gamma (1- \gamma )\left(1-\frac{2}{3} z_{1}z_{0} \right) + \gamma ^2\left(1-\frac{2}{3} z_{1}z_{0}z_{0,2}z_{2,4}z_{4,6}z_{6,7}z_{5,7}z_{3,5}z_{1,3} \right) $ & $ 1-\frac{2}{3} w_{1}w_{0,1}w_{0} $ \\ 
			& & $ (1- \gamma )^2\frac{1}{3} z_{1}z_{0}z_{0,1} +2 \gamma (1- \gamma )\frac{1}{3} z_{1}z_{0} + \gamma ^2\frac{1}{3} z_{1}z_{0}z_{0,2}z_{2,4}z_{4,6}z_{6,7}z_{5,7}z_{3,5}z_{1,3} $ & $ \frac{1}{3} w_{1}w_{0,1}w_{0} $ \\ \midrule
			$ (0,1,1,2,0) $ & Tree & $ 1-\frac{2}{3}
                                                 z_{3} $ & $
                                                           1-\frac{2}{3}
                                                           w_{3} $ \\
        			& & $ \frac{1}{3} z_{3} $ & $ \frac{1}{3} w_{3} $ \\ \midrule
			$ (0,2,1,1,0) $ & Tree & $ 1-\frac{2}{3}
                                                 z_{1,3}z_{1} $ & $
                                                                  1-\frac{2}{3}
                                                                  w_{1,3}w_{1}
                                                                  $ \\

        & & $ \frac{1}{3} z_{1,3}z_{1} $ & $ \frac{1}{3} w_{1,3}w_{1} $ \\ \midrule
			$ (1,0,1,2,0) $ & $ 2 $ & $ (1- \gamma )\left(1-\frac{2}{3} z_{3} \right) + \gamma \left(1-\frac{2}{3} z_{2,4}z_{4,6}z_{6,7}z_{5,7}z_{3,5}z_{3} \right) $ & $ 1-\frac{2}{3} w_{3} $ \\ 
			& & $ (1- \gamma )\frac{1}{3} z_{3} + \gamma
                            \frac{1}{3}
                            z_{2,4}z_{4,6}z_{6,7}z_{5,7}z_{3,5}z_{3} $
                                                              & $
                                                                \frac{1}{3}
                                                                w_{3}
                                                                $ \\
        \midrule
        $ (1,1,1,1,0) $ & $ 3 $ & $ (1- \gamma )\left(1-\frac{2}{3} z_{1,3} \right) + \gamma \frac{1}{3} z_{2,4}z_{4,6}z_{6,7}z_{5,7}z_{3,5} $ & $ 1-\frac{2}{3} w_{1,3} $ \\ 
			& & $ (1- \gamma )\frac{1}{3} z_{1,3} + \gamma \left(1-\frac{2}{3} z_{2,4}z_{4,6}z_{6,7}z_{5,7}z_{3,5} \right) $ & $ \frac{1}{3} w_{1,3} $ \\ 
			& & $ (1- \gamma )\frac{1}{3} z_{1,3} + \gamma
                            \frac{1}{3}
                            z_{2,4}z_{4,6}z_{6,7}z_{5,7}z_{3,5} $ & $
                                                                    \frac{1}{3}
                                                                    w_{1,3}
                                                                    $
        \\ 
                \bottomrule
    \end{tabular}
    \end{table}

  \begin{table}[ht!]
      \caption{Systems of CF polynomial equations for the case of
        $k_i\geq 5$. Here we do not repeat the minor CF equations for the
        cases when there are two equal. Also,
        $n=(n_0,n_1,n_2,n_3,n_4)$ as in Figure \ref{k5nettree}, and
        ``Type'' corresponds to the type of quarnet in Figure
        \ref{5quartets}.}
      \label{k5table-2}
      \centering
      \begin{tabular}{llll}
        \toprule
			$ (1,2,1,0,0) $ & $ 2 $ & $ (1- \gamma )\left(1-\frac{2}{3} z_{1} \right) + \gamma \left(1-\frac{2}{3} z_{2,4}z_{4,6}z_{6,7}z_{5,7}z_{3,5}z_{1,3}z_{1} \right) $ & $ 1-\frac{2}{3} w_{1} $ \\ 
			& & $ (1- \gamma )\frac{1}{3} z_{1} + \gamma \frac{1}{3} z_{2,4}z_{4,6}z_{6,7}z_{5,7}z_{3,5}z_{1,3}z_{1} $ & $ \frac{1}{3} w_{1} $ \\ \midrule
			$ (2,0,1,1,0) $ & $ 1 $ & $ (1- \gamma )^2\left(1-\frac{2}{3} z_{0}z_{1,3}z_{0,1} \right) +2 \gamma (1- \gamma )\left(1- z_{0} +\frac{1}{3} z_{0}z_{2,4}z_{4,6}z_{6,7}z_{5,7}z_{3,5} \right) + \gamma ^2\left(1-\frac{2}{3} z_{0}z_{0,2} \right) $ & $ 1-\frac{2}{3} w_{1,3}w_{0,1}w_{0} $ \\ 
			& & $ (1- \gamma )^2\frac{1}{3} z_{0}z_{1,3}z_{0,1} + \gamma (1- \gamma ) z_{0} \left(1-\frac{1}{3} z_{2,4}z_{4,6}z_{6,7}z_{5,7}z_{3,5} \right) + \gamma ^2\frac{1}{3} z_{0}z_{0,2} $ & $ \frac{1}{3} w_{1,3}w_{0,1}w_{0} $ \\ \midrule
			$ (2,1,1,0,0) $ & $ 1 $ & $ (1- \gamma )^2\left(1-\frac{2}{3} z_{0}z_{0,1} \right) +2 \gamma (1- \gamma )\left(1- z_{0} +\frac{1}{3} z_{0}z_{2,4}z_{4,6}z_{6,7}z_{5,7}z_{3,5}z_{1,3} \right) + \gamma ^2\left(1-\frac{2}{3} z_{0}z_{0,2} \right) $ & $ 1-\frac{2}{3} w_{0,1}w_{0} $ \\ 
			& & $ (1- \gamma )^2\frac{1}{3} z_{0}z_{0,1} + \gamma (1- \gamma ) z_{0} \left(1-\frac{1}{3} z_{2,4}z_{4,6}z_{6,7}z_{5,7}z_{3,5}z_{1,3} \right) + \gamma ^2\frac{1}{3} z_{0}z_{0,2} $ & $ \frac{1}{3} w_{0,1}w_{0} $ \\ \midrule
			$ (0,0,2,2,0) $ & Tree & $ 1-\frac{2}{3}
                                                 z_{2}z_{2,4}z_{4,6}z_{6,7}z_{5,7}z_{3,5}z_{3}
                                                 $ & $ 1-\frac{2}{3}
                                                     w_{2}w_{2,4}w_{4,6}w_{6,7}w_{5,7}w_{3,5}w_{3}
                                                     $ \\ 
        & & $ \frac{1}{3} z_{2}z_{2,4}z_{4,6}z_{6,7}z_{5,7}z_{3,5}z_{3} $ & $ \frac{1}{3} w_{2}w_{2,4}w_{4,6}w_{6,7}w_{5,7}w_{3,5}w_{3} $ \\ \midrule
			$ (0,1,2,1,0) $ & Tree & $ 1-\frac{2}{3} z_{3,5}z_{5,7}z_{6,7}z_{4,6}z_{2,4}z_{2} $ & $ 1-\frac{2}{3} w_{3,5}w_{5,7}w_{6,7}w_{4,6}w_{2,4}w_{2} $ \\
			& & $ \frac{1}{3} z_{3,5}z_{5,7}z_{6,7}z_{4,6}z_{2,4}z_{2} $ & $ \frac{1}{3} w_{3,5}w_{5,7}w_{6,7}w_{4,6}w_{2,4}w_{2} $ \\ \midrule
			$ (0,2,2,0,0) $ & Tree & $ 1-\frac{2}{3} z_{2}z_{2,4}z_{4,6}z_{6,7}z_{5,7}z_{3,5}z_{1,3}z_{1} $ & $ 1-\frac{2}{3} w_{2}w_{2,4}w_{4,6}w_{6,7}w_{5,7}w_{3,5}w_{1,3}w_{1} $ \\ 
			& & $ \frac{1}{3} z_{2}z_{2,4}z_{4,6}z_{6,7}z_{5,7}z_{3,5}z_{1,3}z_{1} $ & $ \frac{1}{3} w_{2}w_{2,4}w_{4,6}w_{6,7}w_{5,7}w_{3,5}w_{1,3}w_{1} $ \\ \midrule
			$ (1,0,2,1,0) $ & $ 2 $ & $ (1- \gamma )\left(1-\frac{2}{3} z_{3,5}z_{5,7}z_{6,7}z_{4,6}z_{2,4}z_{2} \right) + \gamma \left(1-\frac{2}{3} z_{2} \right) $ & $ 1-\frac{2}{3} w_{3,5}w_{5,7}w_{6,7}w_{4,6}w_{2,4}w_{2} $ \\ 
			& & $ (1- \gamma )\frac{1}{3} z_{3,5}z_{5,7}z_{6,7}z_{4,6}z_{2,4}z_{2} + \gamma \frac{1}{3} z_{2} $ & $ \frac{1}{3} w_{3,5}w_{5,7}w_{6,7}w_{4,6}w_{2,4}w_{2} $ \\ \midrule
			$ (1,1,2,0,0) $ & $ 2 $ & $ (1- \gamma )\left(1-\frac{2}{3} z_{1,3}z_{3,5}z_{5,7}z_{6,7}z_{4,6}z_{2,4}z_{2} \right) + \gamma \left(1-\frac{2}{3} z_{2} \right) $ & $ 1-\frac{2}{3} w_{1,3}w_{3,5}w_{5,7}w_{6,7}w_{4,6}w_{2,4}w_{2} $ \\ 
			& & $ (1- \gamma )\frac{1}{3} z_{1,3}z_{3,5}z_{5,7}z_{6,7}z_{4,6}z_{2,4}z_{2} + \gamma \frac{1}{3} z_{2} $ & $ \frac{1}{3} w_{1,3}w_{3,5}w_{5,7}w_{6,7}w_{4,6}w_{2,4}w_{2} $ \\ \midrule
			$ (2,0,2,0,0) $ & $ 4 $ & $ (1- \gamma )^2\left(1-\frac{2}{3} z_{2}z_{0}z_{0,1}z_{1,3}z_{3,5}z_{5,7}z_{6,7}z_{4,6}z_{2,4} \right) +2 \gamma (1- \gamma )\left(1-\frac{2}{3} z_{2}z_{0} \right) + \gamma ^2\left(1-\frac{2}{3} z_{2}z_{0}z_{0,2} \right) $ & $ 1-\frac{2}{3} w_{2}w_{2,4}w_{4,6}w_{6,7}w_{5,7}w_{3,5}w_{1,3}w_{0,1}w_{0} $ \\ 
			& & $ (1- \gamma )^2\frac{1}{3} z_{2}z_{0}z_{0,1}z_{1,3}z_{3,5}z_{5,7}z_{6,7}z_{4,6}z_{2,4} +2 \gamma (1- \gamma )\frac{1}{3} z_{2}z_{0} + \gamma ^2\frac{1}{3} z_{2}z_{0}z_{0,2} $ & $ \frac{1}{3} w_{2}w_{2,4}w_{4,6}w_{6,7}w_{5,7}w_{3,5}w_{1,3}w_{0,1}w_{0} $ \\ \midrule
			$ (0,1,0,2,1) $ & Tree & $ 1-\frac{2}{3} z_{3}
                                                 $ & $ 1-\frac{2}{3}
                                                     w_{3} $ \\
        			& & $ \frac{1}{3} z_{3} $ & $ \frac{1}{3} w_{3} $ \\ \midrule
			$ (0,2,0,1,1) $ & Tree & $ 1-\frac{2}{3}
                                                 z_{1,3}z_{1} $ & $
                                                                  1-\frac{2}{3}
                                                                  w_{1,3}w_{1}
                                                                  $ \\
        & & $ \frac{1}{3} z_{1,3}z_{1} $ & $ \frac{1}{3} w_{1,3}w_{1}
                                           $ \\
                        \bottomrule
    \end{tabular}
    \end{table}

  \begin{table}[ht!]
      \caption{Systems of CF polynomial equations for the case of
        $k_i\geq 5$. Here we do not repeat the minor CF equations for the
        cases when there are two equal. Also,
        $n=(n_0,n_1,n_2,n_3,n_4)$ as in Figure \ref{k5nettree}, and
        ``Type'' corresponds to the type of quarnet in Figure
        \ref{5quartets}.}
      \label{k5table-3}
      \centering
      \begin{tabular}{llll}
        \toprule
			$ (1,0,0,2,1) $ & $ 2 $ & $ (1- \gamma )\left(1-\frac{2}{3} z_{3} \right) + \gamma \left(1-\frac{2}{3} z_{4,6}z_{6,7}z_{5,7}z_{3,5}z_{3} \right) $ & $ 1-\frac{2}{3} w_{3} $ \\ 
			& & $ (1- \gamma )\frac{1}{3} z_{3} + \gamma \frac{1}{3} z_{4,6}z_{6,7}z_{5,7}z_{3,5}z_{3} $ & $ \frac{1}{3} w_{3} $ \\ \midrule
			$ (1,1,0,1,1) $ & $ 3 $ & $ (1- \gamma )\left(1-\frac{2}{3} z_{1,3} \right) + \gamma \frac{1}{3} z_{2,4} $ & $ 1-\frac{2}{3} w_{1,3} $ \\ 
			& & $ (1- \gamma )\frac{1}{3} z_{1,3} + \gamma \left(1-\frac{2}{3} z_{2,4} \right) $ & $ \frac{1}{3} w_{1,3} $ \\ 
			& & $ (1- \gamma )\frac{1}{3} z_{1,3} + \gamma \frac{1}{3} z_{2,4} $ & $ \frac{1}{3} w_{1,3} $ \\ \midrule
			$ (1,2,0,0,1) $ & $ 2 $ & $ (1- \gamma )\left(1-\frac{2}{3} z_{1} \right) + \gamma \left(1-\frac{2}{3} z_{4,6}z_{6,7}z_{5,7}z_{3,5}z_{1,3}z_{1} \right) $ & $ 1-\frac{2}{3} w_{1} $ \\ 
			& & $ (1- \gamma )\frac{1}{3} z_{1} + \gamma \frac{1}{3} z_{4,6}z_{6,7}z_{5,7}z_{3,5}z_{1,3}z_{1} $ & $ \frac{1}{3} w_{1} $ \\ \midrule
			$ (2,0,0,1,1) $ & $ 1 $ & $ (1- \gamma )^2\left(1-\frac{2}{3} z_{0}z_{0,1}z_{1,3} \right) +2 \gamma (1- \gamma )\left(1- z_{0} +\frac{1}{3} z_{0}z_{4,6}z_{6,7}z_{5,7}z_{3,5} \right) + \gamma ^2\left(1-\frac{2}{3} z_{0}z_{0,2}z_{2,4} \right) $ & $ 1-\frac{2}{3} w_{1,3}w_{0,1}w_{0} $ \\ 
			& & $ (1- \gamma )^2\frac{1}{3} z_{0}z_{0,1}z_{1,3} + \gamma (1- \gamma ) z_{0} \left(1-\frac{1}{3} z_{4,6}z_{6,7}z_{5,7}z_{3,5} \right) + \gamma ^2\frac{1}{3} z_{0}z_{0,2}z_{2,4} $ & $ \frac{1}{3} w_{1,3}w_{0,1}w_{0} $ \\ \midrule
			$ (2,1,0,0,1) $ & $ 1 $ & $ (1- \gamma )^2\left(1-\frac{2}{3} z_{0}z_{0,1} \right) +2 \gamma (1- \gamma )\left(1- z_{0} +\frac{1}{3} z_{0}z_{4,6}z_{6,7}z_{5,7}z_{3,5}z_{1,3} \right) + \gamma ^2\left(1-\frac{2}{3} z_{0}z_{0,2}z_{2,4} \right) $ & $ 1-\frac{2}{3} w_{0,1}w_{0} $ \\ 
			& & $ (1- \gamma )^2\frac{1}{3} z_{0}z_{0,1} + \gamma (1- \gamma ) z_{0} \left(1-\frac{1}{3} z_{4,6}z_{6,7}z_{5,7}z_{3,5}z_{1,3} \right) + \gamma ^2\frac{1}{3} z_{0}z_{0,2}z_{2,4} $ & $ \frac{1}{3} w_{0,1}w_{0} $ \\ \midrule
			$ (0,0,1,2,1) $ & Tree & $ 1-\frac{2}{3} z_{4,6}z_{6,7}z_{5,7}z_{3,5}z_{3} $ & $ 1-\frac{2}{3} w_{4,6}w_{6,7}w_{5,7}w_{3,5}w_{3} $ \\
			& & $ \frac{1}{3} z_{4,6}z_{6,7}z_{5,7}z_{3,5}z_{3} $ & $ \frac{1}{3} w_{4,6}w_{6,7}w_{5,7}w_{3,5}w_{3} $ \\ \midrule
			$ (0,1,1,1,1) $ & Tree & $ 1-\frac{2}{3} z_{4,6}z_{6,7}z_{5,7}z_{3,5} $ & $ 1-\frac{2}{3} w_{4,6}w_{6,7}w_{5,7}w_{3,5} $ \\
			& & $ \frac{1}{3} z_{4,6}z_{6,7}z_{5,7}z_{3,5} $ & $ \frac{1}{3} w_{4,6}w_{6,7}w_{5,7}w_{3,5} $ \\ \midrule
			$ (0,2,1,0,1) $ & Tree & $ 1-\frac{2}{3} z_{4,6}z_{6,7}z_{5,7}z_{3,5}z_{1,3}z_{1} $ & $ 1-\frac{2}{3} w_{4,6}w_{6,7}w_{5,7}w_{3,5}w_{1,3}w_{1} $ \\
			& & $ \frac{1}{3} z_{4,6}z_{6,7}z_{5,7}z_{3,5}z_{1,3}z_{1} $ & $ \frac{1}{3} w_{4,6}w_{6,7}w_{5,7}w_{3,5}w_{1,3}w_{1} $ \\ \midrule
			$ (1,0,1,1,1) $ & $ 3 $ & $ (1- \gamma )\left(1-\frac{2}{3} z_{4,6}z_{6,7}z_{5,7}z_{3,5} \right) + \gamma \frac{1}{3} z_{2,4} $ & $ 1-\frac{2}{3} w_{4,6}w_{6,7}w_{5,7}w_{3,5} $ \\ 
			& & $ (1- \gamma )\frac{1}{3} z_{4,6}z_{6,7}z_{5,7}z_{3,5} + \gamma \left(1-\frac{2}{3} z_{2,4} \right) $ & $ \frac{1}{3} w_{4,6}w_{6,7}w_{5,7}w_{3,5} $ \\ 
			& & $ (1- \gamma )\frac{1}{3}
                            z_{4,6}z_{6,7}z_{5,7}z_{3,5} + \gamma
                            \frac{1}{3} z_{2,4} $ & $ \frac{1}{3}
                                                    w_{4,6}w_{6,7}w_{5,7}w_{3,5}
                                                    $ \\ \midrule
			$ (1,1,1,0,1) $ & $ 3 $ & $ (1- \gamma )\left(1-\frac{2}{3} z_{4,6}z_{6,7}z_{5,7}z_{3,5}z_{1,3} \right) + \gamma \frac{1}{3} z_{2,4} $ & $ 1-\frac{2}{3} w_{4,6}w_{6,7}w_{5,7}w_{3,5}w_{1,3} $ \\ 
			& & $ (1- \gamma )\frac{1}{3} z_{4,6}z_{6,7}z_{5,7}z_{3,5}z_{1,3} + \gamma \left(1-\frac{2}{3} z_{2,4} \right) $ & $ \frac{1}{3} w_{4,6}w_{6,7}w_{5,7}w_{3,5}w_{1,3} $ \\ 
			& & $ (1- \gamma )\frac{1}{3} z_{4,6}z_{6,7}z_{5,7}z_{3,5}z_{1,3} + \gamma \frac{1}{3} z_{2,4} $ & $ \frac{1}{3} w_{4,6}w_{6,7}w_{5,7}w_{3,5}w_{1,3} $ \\
                                \bottomrule
    \end{tabular}
    \end{table}

  \begin{table}[ht!]
      \caption{Systems of CF polynomial equations for the case of
        $k_i\geq 5$. Here we do not repeat the minor CF equations for the
        cases when there are two equal. Also,
        $n=(n_0,n_1,n_2,n_3,n_4)$ as in Figure \ref{k5nettree}, and
        ``Type'' corresponds to the type of quarnet in Figure
        \ref{5quartets}.}
      \label{k5table-4}
      \centering
      \begin{tabular}{llll}
        \toprule
			$ (2,0,1,0,1) $ & $ 1 $ & $ (1- \gamma )^2\left(1-\frac{2}{3} z_{0}z_{0,1}z_{1,3}z_{3,5}z_{5,7}z_{6,7}z_{4,6} \right) +2 \gamma (1- \gamma )\left(1- z_{0} +\frac{1}{3} z_{0}z_{2,4} \right) + \gamma ^2\left(1-\frac{2}{3} z_{0}z_{0,2} \right) $ & $ 1-\frac{2}{3} w_{4,6}w_{6,7}w_{5,7}w_{3,5}w_{1,3}w_{0,1}w_{0} $ \\ 
			& & $ (1- \gamma )^2\frac{1}{3} z_{0}z_{0,1}z_{1,3}z_{3,5}z_{5,7}z_{6,7}z_{4,6} + \gamma (1- \gamma ) z_{0} \left(1-\frac{1}{3} z_{2,4} \right) + \gamma ^2\frac{1}{3} z_{0}z_{0,2} $ & $ \frac{1}{3} w_{4,6}w_{6,7}w_{5,7}w_{3,5}w_{1,3}w_{0,1}w_{0} $ \\ \midrule
			$ (0,0,2,1,1) $ & Tree & $ 1-\frac{2}{3}
                                                 z_{2,4}z_{2} $ & $
                                                                  1-\frac{2}{3}
                                                                  w_{2,4}w_{2}
                                                                  $ \\
       & & $ \frac{1}{3} z_{2,4}z_{2} $ & $ \frac{1}{3} w_{2,4}w_{2} $ \\ \midrule
			$ (0,1,2,0,1) $ & Tree & $ 1-\frac{2}{3}
                                                 z_{2,4}z_{2} $ & $
                                                                  1-\frac{2}{3}
                                                                  w_{2,4}w_{2}
                                                                  $ \\
        		& & $ \frac{1}{3} z_{2,4}z_{2} $ & $ \frac{1}{3} w_{2,4}w_{2} $ \\ \midrule
			$ (1,0,2,0,1) $ & $ 2 $ & $ (1- \gamma )\left(1-\frac{2}{3} z_{2,4}z_{2} \right) + \gamma \left(1-\frac{2}{3} z_{2} \right) $ & $ 1-\frac{2}{3} w_{2,4}w_{2} $ \\ 
			& & $ (1- \gamma )\frac{1}{3} z_{2,4}z_{2} + \gamma \frac{1}{3} z_{2} $ & $ \frac{1}{3} w_{2,4}w_{2} $ \\ \midrule
			$ (0,0,0,2,2) $ & Tree & $ 1-\frac{2}{3} z_{4}z_{4,6}z_{6,7}z_{5,7}z_{3,5}z_{3} $ & $ 1-\frac{2}{3} w_{4}w_{4,6}w_{6,7}w_{5,7}w_{3,5}w_{3} $ \\
			& & $ \frac{1}{3} z_{4}z_{4,6}z_{6,7}z_{5,7}z_{3,5}z_{3} $ & $ \frac{1}{3} w_{4}w_{4,6}w_{6,7}w_{5,7}w_{3,5}w_{3} $ \\ \midrule
			$ (0,1,0,1,2) $ & Tree & $ 1-\frac{2}{3} z_{3,5}z_{5,7}z_{6,7}z_{4,6}z_{4} $ & $ 1-\frac{2}{3} w_{3,5}w_{5,7}w_{6,7}w_{4,6}w_{4} $ \\ 
			& & $ \frac{1}{3} z_{3,5}z_{5,7}z_{6,7}z_{4,6}z_{4} $ & $ \frac{1}{3} w_{3,5}w_{5,7}w_{6,7}w_{4,6}w_{4} $ \\ \midrule
			$ (0,2,0,0,2) $ & Tree & $ 1-\frac{2}{3} z_{4}z_{4,6}z_{6,7}z_{5,7}z_{3,5}z_{1,3}z_{1} $ & $ 1-\frac{2}{3} w_{4}w_{4,6}w_{6,7}w_{5,7}w_{3,5}w_{1,3}w_{1} $ \\ 
			& & $ \frac{1}{3} z_{4}z_{4,6}z_{6,7}z_{5,7}z_{3,5}z_{1,3}z_{1} $ & $ \frac{1}{3} w_{4}w_{4,6}w_{6,7}w_{5,7}w_{3,5}w_{1,3}w_{1} $ \\ \midrule
			$ (1,0,0,1,2) $ & $ 2 $ & $ (1- \gamma )\left(1-\frac{2}{3} z_{3,5}z_{5,7}z_{6,7}z_{4,6}z_{4} \right) + \gamma \left(1-\frac{2}{3} z_{4} \right) $ & $ 1-\frac{2}{3} w_{3,5}w_{5,7}w_{6,7}w_{4,6}w_{4} $ \\ 
			& & $ (1- \gamma )\frac{1}{3} z_{3,5}z_{5,7}z_{6,7}z_{4,6}z_{4} + \gamma \frac{1}{3} z_{4} $ & $ \frac{1}{3} w_{3,5}w_{5,7}w_{6,7}w_{4,6}w_{4} $ \\ \midrule
			$ (1,1,0,0,2) $ & $ 2 $ & $ (1- \gamma )\left(1-\frac{2}{3} z_{1,3}z_{3,5}z_{5,7}z_{6,7}z_{4,6}z_{4} \right) + \gamma \left(1-\frac{2}{3} z_{4} \right) $ & $ 1-\frac{2}{3} w_{4}w_{4,6}w_{6,7}w_{5,7}w_{3,5}w_{1,3} $ \\ 
			& & $ (1- \gamma )\frac{1}{3} z_{1,3}z_{3,5}z_{5,7}z_{6,7}z_{4,6}z_{4} + \gamma \frac{1}{3} z_{4} $ & $ \frac{1}{3} w_{4}w_{4,6}w_{6,7}w_{5,7}w_{3,5}w_{1,3} $ \\ \midrule
			$ (2,0,0,0,2) $ & $ 4 $ & $ (1- \gamma )^2\left(1-\frac{2}{3} z_{4}z_{0}z_{0,1}z_{1,3}z_{3,5}z_{5,7}z_{6,7}z_{4,6} \right) +2 \gamma (1- \gamma )\left(1-\frac{2}{3} z_{4}z_{0} \right) + \gamma ^2\left(1-\frac{2}{3} z_{4}z_{0}z_{0,2}z_{2,4} \right) $ & $ 1-\frac{2}{3} w_{4}w_{4,6}w_{6,7}w_{5,7}w_{3,5}w_{1,3}w_{0,1}w_{0} $ \\ 
			& & $ (1- \gamma )^2\frac{1}{3} z_{4}z_{0}z_{0,1}z_{1,3}z_{3,5}z_{5,7}z_{6,7}z_{4,6} +2 \gamma (1- \gamma )\frac{1}{3} z_{4}z_{0} + \gamma ^2\frac{1}{3} z_{4}z_{0}z_{0,2}z_{2,4} $ & $ \frac{1}{3} w_{4}w_{4,6}w_{6,7}w_{5,7}w_{3,5}w_{1,3}w_{0,1}w_{0} $ \\ \midrule
			$ (0,0,1,1,2) $ & Tree & $ 1-\frac{2}{3} z_{4}
                                                 $ & $ 1-\frac{2}{3}
                                                     w_{4} $ \\
                                        & & $ \frac{1}{3} z_{4} $ & $ \frac{1}{3} w_{4} $ \\
        \bottomrule
    \end{tabular}
    \end{table}
  \end{landscape}

  \begin{table}[ht!]
      \caption{Systems of CF polynomial equations for the case of
        $k_i\geq 5$. Here we do not repeat the minor CF equations for the
        cases when there are two equal. Also,
        $n=(n_0,n_1,n_2,n_3,n_4)$ as in Figure \ref{k5nettree}, and
        ``Type'' corresponds to the type of quarnet in Figure
        \ref{5quartets}.}
      \label{k5table-5}
      \centering
      \begin{tabular}{llll}
        \toprule
			$ (0,1,1,0,2) $ & Tree & $ 1-\frac{2}{3} z_{4}
                                                 $ & $ 1-\frac{2}{3}
                                                     w_{4} $ \\
        & & $ \frac{1}{3} z_{4} $ & $ \frac{1}{3} w_{4} $ \\ \midrule
			$ (1,0,1,0,2) $ & $ 2 $ & $ (1- \gamma )\left(1-\frac{2}{3} z_{4} \right) + \gamma \left(1-\frac{2}{3} z_{2,4}z_{4} \right) $ & $ 1-\frac{2}{3} w_{4} $ \\ 
			& & $ (1- \gamma )\frac{1}{3} z_{4} + \gamma \frac{1}{3} z_{2,4}z_{4} $ & $ \frac{1}{3} w_{4} $ \\ \midrule
			$ (0,0,2,0,2) $ & Tree & $ 1-\frac{2}{3}
                                                 z_{2}z_{2,4}z_{4} $ &
                                                                       $ 1-\frac{2}{3} w_{2}w_{2,4}w_{4} $ \\
        & & $ \frac{1}{3} z_{2}z_{2,4}z_{4} $ & $ \frac{1}{3} w_{2}w_{2,4}w_{4} $ \\
        \bottomrule
    \end{tabular}
    \end{table}

\end{document}